\documentclass[draftclsnofoot,journal,onecolumn]{IEEEtran}

\usepackage{graphicx}
\usepackage{amsmath}
\usepackage{amssymb}
\usepackage{amsfonts}
\usepackage[lined, boxed]{algorithm2e}
\usepackage[dvips,bookmarks=false]{hyperref}
\usepackage{algorithmic}

\begin{document}

\newtheorem{theorem}{Theorem}
\newtheorem{lemma}{Lemma}
\newtheorem{claim}{Claim}
\newtheorem{corollary}{Corollary}
\newtheorem{definition}{Definition}
\newtheorem{proposition}{Proposition}
\title{Universal Denoising of Discrete-time Continuous-Amplitude Signals}

%

\thispagestyle{empty}
\def\x{{\mathbf x}}
\def\L{{\cal L}}

\author{Kamakshi~Sivaramakrishnan \dag ~and~Tsachy~Weissman \dag,\S,~\IEEEmembership{Senior Member,~IEEE}
\thanks{Work supported in part by National Science Foundation through grants CCR-0311633 and the NSF CAREER.}
\thanks{\dag \; The research was partially supported by NASA New Horizons program through grant No. 399840Q.} \thanks{The material in this correspondence was presented in part at the IEEE International Symposium on Information Theory, Seattle, WA, July 2006} \thanks{\dag \; Department of Electrical Engineering, Stanford University, Stanford, CA 94305 USA (e-mails:\{ksivaram,tsachy\}@stanford.edu)} \thanks{\S  \; Department of Electrical Engineering, Technion -- Israel Institute of Technology, Haifa, Israel. } 
} \maketitle
\begin{abstract}
We consider the problem of reconstructing a discrete-time signal (sequence)
with continuous-valued components corrupted by a known memoryless
channel. When performance is measured using a per-symbol
loss function satisfying mild regularity conditions, we develop a
sequence of denoisers that, although independent of the
distribution of the underlying `clean' sequence, is universally optimal in the limit of large sequence length. This sequence of denoisers is universal in the sense of
performing as well as any sliding window denoising scheme which
may be optimized for the underlying clean signal. Our results are
initially developed in a ``semi-stochastic'' setting, where the
noiseless signal is an unknown individual sequence, and the only
source of randomness is due to the channel noise. It is
subsequently shown that in the fully stochastic setting, where the
noiseless sequence is a stationary stochastic process, our schemes
universally  attain optimum performance. The proposed schemes draw
from nonparametric density estimation techniques and are practically implementable. We demonstrate efficacy of the
proposed schemes in denoising gray-scale images in the
conventional additive white Gaussian noise setting, with additional
promising results for less conventional noise distributions.
\end{abstract}


\section{Introduction}
Consider the problem of estimating a clean discrete-time signal (sequence) $\{X_t\}_{t \in
\mathbb{T}}$, $X_t \in [a,b] \subset \mathbb{R}$, based on its noisy
observations $\{Z_t\}_{t \in \mathbb{T}}$, $Z_t \in \mathbb{R}$,
where $\{Z_t\}$ is  the output of a corruption mechanism, a memoryless channel. This problem finds applications in areas ranging
from engineering, cryptography and statistics, to bioinformatics and beyond.
There is significant literature on particular instantiations of
this problem, most notably for the case where signal and noise
components are real-valued and the noise is additive, most
commonly Gaussian (cf.\ \cite{Donoho_1} and references therein).
Solutions to this problem in \cite{Donoho_1} are based on
wavelet-based soft thresholding and have various asymptotic
optimality properties under a minimax criterion. The scope of
wavelet-based thresholding in \cite{Donoho_1} has been extended
beyond the additive white Gaussian case in \cite{Gao},
\cite{Averkamp} where optimality is again established in an
asymptotic minimax sense. The soft-thresholding scheme proposed in
\cite{Averkamp} is among the few denoisers found in the literature
\cite{Gao, Kolaczyk} that are designed for the case of a
non-Gaussian corruption mechanism. Even in this case, restrictions to additive noise and symmetry assumptions on
the noise distribution are made in order to provide asymptotic
performance guarantees. For the case of a random vector $Y = X + Z$, where $X$ is independent of $Z$
(with known distribution). The Minimum Mean Squared Estimate (MMSE) of $X$ is well-known to be given by $\hat{X} = \psi(Y) = E\{X|Y\}$. It was shown in \cite{Raphan} that, for $Z \sim \mathcal{N} \left(\mu, \Sigma \right)$, $\psi(\cdot)$ satisfies $\psi(Y) = \frac{(Y- \mu)  - \triangledown_y \ln f_Y(Y)}{f_Y} $, where $f_Y(y)$
is the marginal density of $Y$, which can be learned from the noisy samples $Y^n  =\{Y_1,\cdots,Y_n\}$ of $Y$. Using techniques for nonparametric density estimation in \cite{Luc}, an estimate of $f_Y(y)$, $\hat{f}_Y(y)$, can be computed, the (appropriate) gradient of which leads to the following estimate:
\begin{equation}\hat{\psi}(Y) =  \frac{(Y- \mu)  - \triangledown_y \ln \hat{f}_Y(Y)}{\hat{f}_Y} \label{eqn:psi}\end{equation}
The authors in \cite{Raphan} also discuss expressions for $\hat{\psi}(Y)$ for a certain class of non-Gaussian noise distributions with the corruption mechanism continuing to be additive. This leaves room for universal denoising schemes for continuous valued data for a general class of  noise distributions where the corruption mechanism is also arbitrary. Compression based approaches pioneered in
(cf., e. g., \cite{Natarajan} and \cite{Donoho_2}), as discussed in
\cite{Tsachy}, are provably sub-optimal and suffer
from non-practicality of implementation of optimal lossy
compression schemes. The wavelet-based Bayesian estimation
approach in \cite{Portilla}, has demonstrated significant
improvement in image denoising. However, despite much recent
progress, the problem of universal denoising for discrete-time
continuous-amplitude data is still a largely open problem of both
theoretical and practical value. The problem is particularly
relevant in new emerging areas as microarray imaging \cite{Wang},
array-based comparative genomic hybridization (array-CGH)
\cite{Hsu} and medical imaging \cite{Wagner,Gudbjartsson,
McVeigh}, where parametric noise models that are currently used
often fail to capture the true nature of the noise.

Recently, universal  denoising for discrete signals and channels was considered
  in \cite{Tsachy}. The results of \cite{Tsachy}, and the denoising scheme DUDE proposed therein, although attractive theoretically,
  are restricted in their
  practicality to problems with small alphabets. This is a result of 
\begin{itemize}
\item computational issues involved
with collecting higher-order joint distributions from the noisy data.

\item mapping an estimated channel output distribution to an estimated
channel input distribution.

\item count
statistics being too sparse to be reliable for even moderately
large alphabet sizes. 
\end{itemize}
This leaves open challenges in the
application of DUDE to problems like gray-scale image denoising.
More recently, a modified DUDE, using ideas from lossless
compression, was presented in \cite{Motta}. As discussed in that
work,  in spite of circumventing some of the computational issues
mentioned above, the approach leaves room for improvement in the
denoising performance. The problem was further extended to the
discrete-valued input and general output alphabet setting in
\cite{Dembo}. This approach proposes quantization of the output
alphabet space and proceeds on an a similar line to that in
\cite{Tsachy}, showing that there is no essential loss of
optimality in quantizing the channel output before denoising (insofar as learning the statistics of the underlying data is concerned). In
spite of its theoretical elegance, this approach faces similar
issues as the scheme of \cite{Tsachy}, limiting its scope of
applications to small channel input alphabets. The authors of
\cite{Dembo}, while conjecturing the need for mild restrictions on
the channel, suggest an extension of the proposed scheme to the
case where both the input and output alphabet space is
continuous-valued and general. The present work proposes an
extension of the two-stage DUDE-like approach in \cite{Tsachy,
Dembo} to the case of denoising for general alphabets. A natural extension would have been to quantize both
the input and the output space and apply a similar count-statistic
based two-pass approach. The vast literature on nonparametric
density estimation (cf. \cite{Luc} and references therein),
however, points to the opportunity of extracting more reliable
statistics from the observed data, that would lead to better denoising (as measured under a specified loss function). We do, however, maintain the sliding window approach of \cite{Dembo, Tsachy} and show asymptotic universal optimality of our schemes with increasing context lengths in the limit of large sequence lengths.

Recent developments in universal denoising in the particular
context of images have also been reported in \cite{Frenchie}.
Their approach is based on local smoothing methods that make
assumptions on the underlying structure of the data which are more
relevant in image denoising due to the inherent redundancy of natural
images. The consistency results showed the convergence of the
denoising rule to the conditional expected value of the clean
symbol given the noisy neighborhood sans the particular noisy symbol in question. There is potential to improve this result by
incorporating the information from the noisy pixel that is being
denoised too, an approach  at the heart of the denoisers we
present below. We establish the universal optimality of the
suggested denoisers in a generality that applies to arbitrarily
distributed noiseless signals, arbitrary memoryless channels, and
arbitrary loss functions (with some benign regularity conditions).

The remainder of the paper is organized as follows. In section II,
we discuss the problem setup and notations. This is followed by a
description of the technical results that are key to the
construction of the denoisers in section III. In section IV, we establish
universality of a family of denoisers that we develop for the
semi-stochastic setting, in which the clean data is an individual sequence and provide bounds
on the difference between the performance of this proposed family of denoisers and that of the best `symbol-by-symbol' denoiser chosen by a genie with full knowledge of the distribution (or probability  law) of the clean data. Section V details an extension of this proposed family of denoisers to a genie that can select the best sliding window scheme, of any order, with knowledge of the underlying clean data. Section VI discusses the implication of the performance guarantees in the semi-stochastic setting to the fully stochastic setting where the clean data is generated by a stationary stochastic process, rather than an individual sequence. A slightly modified version of the proposed denoiser is shown 
to reduce to the scheme of \cite{Dembo} when the underlying clean data have finite alphabet size. The proposed family of denoisers can, hence, be seen
as a natural extension of those in \cite{Dembo}
to the current setting of denoising continuous valued symbols
corrupted by a continuous memoryless channel where the clean data components may take values in a continuum. In section \ref{sec:results}, we present some preliminary experimental results of applying the proposed schemes to denoising of gray-scale images. We conclude in section \ref{sec:conclusion} with a summary of some propositions for future research directions. Throughout this paper, we maintain the flow by stating the Theorems and Lemmas corresponding to the optimality results in the main body of the paper relegating most of the proofs to the appendices.

\section{Problem Setting and Notations}
\label{sec:problem_setting} Let $\mathbf{x}=(x_1, x_2, \cdots)$ be
an individual (deterministic) noise-free source signal \footnote{throughout the paper we will be using the terms `signal' and `sequence' interchangeably} with components taking
values in $[a,b] \subset \mathbb{R}$ and $\mathbf{Y} = (Y_1,Y_2,
\cdots)$, $Y_i \in \mathbb{R}$ be the corresponding noisy
observations, also referred to as the `output of the channel'
(corruption source). This setting, where both the underlying clean
sequence and the noisy sequence are continuous valued, is the
continuous-amplitude analog of the semi-stochastic setting
discussed in \cite{Dembo}. The channel is specified by a family of
distribution functions $\mathcal{C} = \{F_{Y|x}\}_{x \in [a,b]}$,
where $F_{Y|x}$ denotes the distribution of the channel output
symbol when the input symbol is $x$. Also, we denote the probability measure on $\mathbb{R}$ corresponding to $F_{Y|x}$ by $\mu_x$. We make the following assumptions about the channel,
\begin{itemize}
\item[C1.] A memoryless channel, which is to say that the components of
$\mathbf{Y}$ are independent with $Y_i \sim F_{Y|x_i}$.
\item[C2.] The family of measures, $\{ \mu_x \}_{x \in [a,b]} $, associated
with the channel, $\mathcal{C}$, is uniformly tight in the sense
\begin{equation*}
\sup_{x \in [a,b]} \mu_x([-T,T]^c) \rightarrow 0 \quad \text{as}
\quad  T \rightarrow \infty.
\end{equation*}
This condition will be needed to guarantee that one can consistently track the
evolution of the marginal density of the noisy symbols at the
output of the memoryless channel, regardless of the underlying $\mathbf{x}$, using nonparametric Kernel
density estimation techniques.
\item[C3.] The distribution functions $F_{Y|x}$ are absolutely continuous for all $x \in [a,b]$ w.r.t the Lebesgue
measure and $\{ f_{Y|x} \}$ denotes the corresponding densities. This assumption is not crucial for the validity of our approach but is made for concreteness in the construction of our schemes and the development of their performance guarantees. 
\item[C4.] The conditional densities of the channel form a set of
linearly independent functions. This is equivalent to the ``invertibility'' condition of \cite{Tsachy} which ensures that, to any distribution on the input to the channel there corresponds a unique channel output.
\item[C5.] The mapping, w.r.t a metric that will be detailed in section~\ref{sec:denoiser_construct}, from the space of channel input distributions to the corresponding channel output distributions is continuous. The precise analytical expression describing this condition is discussed in Appendix \ref{app:channel conds}. 
\item[C6.] The expected loss, for reasonably well-behaved loss functions (conditions L1-L2 listed subsequently in this section), induced by two output distributions that are close (under the metric discussed in section~\ref{sec:denoiser_construct}) is continuous. Again, the analytical expression describing this condition is in the Appendix \ref{app:channel conds}. 
\end{itemize}
The above, are rather benign conditions obeyed by most
channels arising in practice, an example of this being the most commonly
addressed channel, viz., the Additive White Gaussian Noise Channel
(AWGN). It is easy to verify that
even the multiplicative (non-additive) Gaussian channel with a
finite variance and mean satisfies these requirements. In this
case, the channel input (underlying clean signal) affects the
variance of the channel. The fact that the underlying clean signal
takes only bounded values implies that the tightness condition, C2, is satisfied. In fact, any additive noise channel with distribution functions that are absolutely continuous and the corresponding densities (of finite mean and variance) satisfying conditions C4-7 (C7 discussed in Appendix \ref{app:channel conds}) will satisfy the above requirements.

An \textit{n-block} denoiser is a measurable mapping taking
$\mathbb{R}^n$ into $[a,b]^n$. We assume a loss function $\Lambda
: [a,b]^2 \rightarrow [0,\infty)$ and denote the normalized
cumulative loss of an $n$-block denoiser $\hat{X}^n$, when the underlying sequence is $x^n$ and the observed sequence is $y^n$, by

\begin{equation}
L_{\hat{X}^n}(x^n,y^n) = \frac{1}{n} \sum_{i=1}^n \Lambda(x_i,
\hat{X}^n(y^n)[i])
\end{equation}
where $\hat{X}^n(y^n)[i]$ denotes the $i$-th component of
$\hat{X}^n(y^n)$. In addition to the constraints on the channel, we impose some
conditions on the permissible loss functions, $\Lambda$. We assume the loss function, $\Lambda$,
\begin{enumerate}
\item[L1.] to be bounded,i.e.,  $\Lambda_{\max} < \infty$ where $\Lambda_{\max} = \sup_{x, \hat{x} \in [a,b]} \Lambda(x,\hat{x})$
\item[L2.] to be a bounded Lipschitz function. More formally, we require the Lipschitz norm, $\|\Lambda \|_L < \infty$. The Lipschitz norm of the loss function, is defined as
\begin{equation}
\parallel \Lambda \parallel_L = \sup_{0< \Delta < (b-a)} \frac{\lambda \left( \Delta
\right)}{\Delta}
\end{equation}
where,
\begin{equation}
\lambda(\Delta,x) = \sup_{y \in [a,b]} \sup_{x':|x-x'| < \Delta} \left|
\Lambda(x,y) - \Lambda(x',y) \right| \label{loss_modcont}
\end{equation}
and
\begin{equation}
\lambda \left( \Delta \right) = \sup_{x \in [a,b]} \lambda \left( \Delta, x
\right)
\label{loss_modcont1}
\end{equation}
In words, this condition necessitates continuity of the mapping that
takes the estimates of the underlying symbol to the corresponding loss incurred. We require that estimates of the underlying clean symbol that are close together have corresponding loss values that are also close to each other. 
\end{enumerate}
It can be easily verified that the commonly used loss functions of
$L_2$, $L_1$ norms satisfy the aforementioned condition.\\ 

	Let $\mathcal{F}^{[a,b]}$ denote the set of all probability distribution functions with support contained in the interval $[a,b]$. For $F \in \mathcal{F}^{[a,b]}$, we let
\begin{equation}
\mathcal{U} (F) = \min_{\hat{x} \in [a,b]} \int_{x \in [a,b]}
\Lambda(x,\hat{x})dF(x)
\end{equation}
denote its `Bayes envelope' (our assumptions on the loss function
will imply existence of the minimum). In other words,
$\mathcal{U}(F)$ denotes the minimum achievable expected loss when guessing the value of $X \sim F$. Define the symbol-by-symbol minimum loss of $x^n$ by
\begin{equation}
D_0(x^n) = \min_{g} E \left[ \frac{1}{n} \sum_{i=1}^n \Lambda(x_i,
g(Y_i) ) \right] \label{1symb_denoisability}
\end{equation}
where the minimum is over all measurable maps $g: \mathbb{R}
\rightarrow [a,b]$. $D_0\left(x^n \right)$ denotes the minimum
expected loss in denoising the sequence $x^n$, using a time-invariant symbol-by-symbol rule.  This can be attained by a ``genie'' with access to the clean sequence $x^n$. $D_0(x^n)$, which is the expected per-symbol loss of the optimal symbol-by-symbol rule for the individual sequence $x^n$, will be our benchmark for assessing the performance of the universal symbol-by-symbol denoiser that we construct in the next section. The same benchmark was used also in \cite{Dembo}. This is slightly different than the benchmark used in \cite{Tsachy}, which corresponded to a genie that can choose the best symbol-by-symbol rule with knowledge not only of the individual sequence $x^n$, but also of the noisy sequence realization $Y^n$. The latter is irrelevant for our current setting where each of the components of $Y^n$ will take on a different value, with probability one.
For $x^n \in [a,b]^n$, define
\begin{equation}
F_{x^n}(x) = \frac{|\{1 \le i \le n: x_i \le x \} |}{n},
\label{dist_defn}
\end{equation}
i.e., the CDF associated with  the empirical distribution of
$x^n$. Note that $D_0(x^n)$ can be expressed as
\begin{equation}
D_0(x^n) = \min_{g} \int_{[a,b]} E_{x} \Lambda(x,g(Y)) dF_{X^n}(x)
\label{eqn_1}
\end{equation}
where $E_x$ denotes expectation when the underlying clean symbol
is $x$, the expectation being over the channel noise
\begin{equation}
E_x \Lambda(x,g(Y)) = \int \Lambda(x,g(y)) f_{Y|x}(y)dy
\label{eqn_2}
\end{equation}
For $F \in \mathcal{F}^{[a,b]} $, let $F \otimes \mathcal{C}$ and
$E_{F \otimes \mathcal{C}}$ denote, respectively, probability and
expectation when the channel input $X \sim F$ and $Y$ is the
channel output. So that,
\begin{eqnarray}
E_{F \otimes \mathcal{C}} \Lambda(X, g(Y)) &=& \int_{[a,b]}
E_x \Lambda(x,g(Y))dF(x) \nonumber \\
 &=& \int_{[a,b]} \left[ \int_{\mathbb{R}} \Lambda(x,g(y))
 f_{Y|x}(y)dy\right] dF(x)
\label{eqn_3}
\end{eqnarray}
Letting $[F \otimes \mathcal{C}]_{X|y}$ denote the conditional
distribution of $X$ given $Y=y$ under $F \otimes \mathcal{C}$,
we have
\begin{equation}
\min_{g} E_{F \otimes \mathcal{C}} \Lambda(X,g(Y)) = E_{F \otimes
\mathcal{C}} \mathcal{U} \left( [F \otimes \mathcal{C}]_{X|Y}
\right)
\label{eqn:bayes_denoiser}
\end{equation}
with $\mathcal{U}$ denoting the Bayes envelope as defined above.
Letting $g_{\text{opt}} \left[ F\right]$ denote the achiever of the minimum in (\ref{eqn:bayes_denoiser}), we note that is given by the Bayes
response to $[F \otimes \mathcal{C}]_{X|y}$, namely,
\begin{eqnarray}
g_{\text{opt}}[F](y) &=& \arg \min_{\hat{x} \in [a,b]}
\int_{[a,b]} \Lambda(x, \hat{x} ) d [F \otimes
\mathcal{C}]_{X|y}(x) \nonumber \\
 &=& \arg \min_{\hat{x} \in [a,b]}
\int_{[a,b]} \Lambda(x, \hat{x} ) f_{Y|x}(y) dF(x)
\label{g_opt_defn}
\end{eqnarray}
In Lemma \ref{lem:bayes_sequence}, we will establish the concavity of  $\mathcal{U}(F)$, and minimizing this bounded (by
our assumption of bounded $\Lambda$) concave function over a
closed compact interval, $[a,b]$, guarantees the existence of the
minimizer, $g_{\text{opt}}$. Note that from (\ref{eqn_1}),
(\ref{eqn_2}) and (\ref{eqn_3}) we have
\begin{equation}
D_0(x^n) = \min_{g} E_{F_{x^n} \otimes \mathcal{C}} \Lambda(X,g(Y))
\label{eqn:sym_sym_denoisability}
\end{equation}
where $F_{x^n}$ was defined in (\ref{dist_defn}) and the
minimum is attained by $g_{\text{opt}} \left[ F_{x^n}\right]$.
Thus, only a ``genie'' with access to the empirical distribution
of the noiseless sequence could employ $g_{\text{opt}}[F_{x^n}]$.

\section{Construction of Universal `Symbol-by-symbol' Denoiser and Preliminaries}
\label{sec:denoiser_construct} $F_{x^n}$ and, hence,
$g_{\text{opt}}[F_{x^n}]$ are not known to an observer of the
noisy sequence. The first step towards constructing an estimate of
$g_{\text{opt}}[F_{x^n}]$ is to estimate the input empirical
distribution from the observable noisy sequence, $Y^n$, and
knowledge of the channel, $\mathcal{C}$. We approach this problem
by first estimating a function that tracks the evolution of the
`average' density function according to which the noisy symbols
are distributed. For an input sequence $x^n$, given the memoryless
nature of the channel, the output symbols will be independent with
respective distributions, $\{F_{Y|x_1}, \cdots, F_{Y|x_n} \}$ and
have the corresponding density functions, $\{f_{Y|x_1}, \cdots,
f_{Y|x_n} \}$. The function we are interested in estimating is
\begin{equation} \label{eq: creature to be estimated}
 f_Y^n(y) = \frac{1}{n}\sum_{i=1}^n f_{Y|x_i}(y) 
\end{equation}
which can be thought of as the marginal density, $ f_Y^n$, of the noisy symbols in
the semi-stochastic setting where $x^n$ is the unknown
deterministic sequence. The estimation of this function is done by
exploiting the vast literature on density estimation techniques
\cite{Luc}, \cite{Luc_2}, the details of which are discussed in
Subsection \ref{sec:density_est} below. Once we have an estimate
$f_Y^n = f_Y^n[Y^n]$ for this function, we use it to estimate the
input empirical distribution by
\begin{equation}
\hat{F}_{x^n}[Y^n]=\arg\min_{F \in \mathcal{F}^{[a,b]}_n} d
\left(f_Y^n, \underbrace{\int f_{Y|x}dF(x)}_{[F \otimes
\mathcal{C}]_Y} \right) \label{chan_inv_defn}
\end{equation}
where $\mathcal{F}_n^{[a,b]} \subseteq \mathcal{F}^{[a,b]}$
denotes the set of empirical distributions induced by $n$-tuples
with $[a,b]$-valued components and $\left[ F \otimes \mathcal{C} \right]_Y$ denotes the marginal density induced at the output of the channel by an input distribution $F$. That is, every member, $F(x)$, of
$\mathcal{F}_n^{[a,b]}$ is of the form
\begin{equation}
F(x) = \frac{1}{n} \sum_{i=1}^n \mathbf{1}_{(x \le x_i)}
\end{equation}
for some $n$-tuple, $x^n = \left( x_1, x_2, \cdots, x_n \right)$, with
$[a,b]$-valued components. The norm, $d$, in (\ref{chan_inv_defn})
is defined as
\begin{equation}
d \left(f, g \right) = \int \left| f(y)-g(y) \right|dy
\label{norm_defn}
\end{equation}
The channel, $\mathcal{C}$, induces a set of `feasible' densities of
the output noisy symbol corresponding to the family of empirical
distributions of the underlying clean sequence at the input of the
channel. The density estimate, $f_Y^n$, which is constructed only
from the noisy sequence, $Y^n$, is oblivious to the set of
achievable marginal densities and hence could lie outside this
set. It is thus natural to estimate the unobserved $F_{x^n}$ by the member of $\mathcal{F}_n^{[a,b]}$ leading to a channel output distribution closest to the estimated one, $f_{Y}^n$. This is exactly the estimate in (\ref{chan_inv_defn}). The uniqueness of the minimizer in (\ref{chan_inv_defn}) follows from the fact that the objective function being minimized is a norm-function and hence convex, coupled with the linear independence assumption of the channel, C4. The assumption, C4, implies a one-to-one correspondence between channel input and channel output distributions (i.e., ``invertibility'' of the channel).  Additionally, the search for the minimizer is conducted on
a convex set of distribution functions, $\mathcal{F}^{[a,b]}_n$,
resulting in uniquely achieving the minimizer or in other words,
the candidate input empirical distribution estimate. 

A two-stage quantization of both, the support of the underlying
clean symbol, $[a,b]$, and the levels of the estimate of its empirical
distribution function, $\hat{F}_{x^n}$, is carried out to give the
corresponding quantized probability mass function that has mass points only at the quantized symbols. 
\begin{itemize}
\item[Q1.] The quantization of the  interval $[a,b]$ is depicted in Fig.~\ref{fig:dist_quant} below.
\begin{figure}[h]
\begin{center}
\includegraphics[width=4.5in,height=2in]{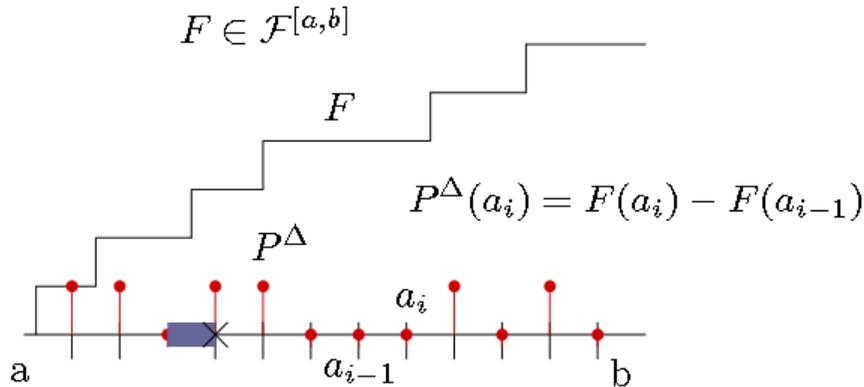}
\caption{Quantization of the support of a distribution function, $F \in \mathcal{F}^{[a,b]}$ }
\label{fig:dist_quant}
\end{center}
\end{figure}
For a given quantization step size, $\Delta$, the quantized symbols, $a_i$ in the interval $[a,b]$ are constructed in the following manner.\\
For $\Delta >0$, $N(\Delta) = \frac{(b-a)}{\Delta}$, if $m =
\lfloor \frac{b-a}{\Delta} \rfloor$, consider a family of vectors,
\begin{equation*}
\begin{split}
\mathcal{F}^{\Delta} &= \{ P^{\Delta}: P^{\Delta} =
\left(P(a_0), P(a_1), \cdots, P(a_{N \left( \Delta \right)})
\right) \} \\ \mathcal{A}^{\Delta} &= \{a_i = a+i\Delta, i = 0,
\cdots, N(\Delta) \} \\ &\text{s.t. } \quad \sum_{i=1}^{N(\Delta)} P(a_i)=1
\end{split}
\end{equation*}
else,
define the family of vectors as $\mathcal{F}^{\Delta} =
\{P^{\Delta}$: $P^{\Delta} = \left(P(a_0), P(a_1), \cdots, P(a_{N
\left( \Delta \right)-1}),\right.$  $\left. P(a_{N \left( \Delta
\right)}) \right) \}$, $\mathcal{A}^{\Delta}= \{a_i = a+i\Delta, i =
0, \cdots, N(\Delta)-1 \}, a_{N(\Delta)} =
b,\sum_{i=1}^{N(\Delta)} P(a_i)=1$.\\
As indicated in Fig.~\ref{fig:dist_quant}, the probability mass function, $P^{\Delta}$, that we propose is constructed by allocating the mass of the distribution function, $F$, in any quantization interval (of length $\Delta$) to the higher end point in that interval. More precisely,
\begin{equation}
P^{\Delta}(a_i) = F(a_i)-F(a_{i-1}) \label{P_delta_0}
\end{equation}
where $a_i$'s as defined above and note that
\begin{eqnarray*}
P^{\Delta} \left(B\right) &=& \sum_{a_i \in B} P(a_i)
\end{eqnarray*}
with any $B \in \mathcal{B}^{[a,b]}$, $\mathcal{B}^{[a,b]}$ is the
Borel sigma-algebra generated by open sets in $[a,b]$. \\

Applying this quantization of the support of the underlying clean
symbol to the estimate, $\hat{F}_{x^n}$, we construct now, the
corresponding probability mass function,
$\hat{P}_{x^n}^{\Delta}$ 
\begin{equation}
\hat{P}_{x^n}^{\Delta}(a_i) = \hat{F}_{x^n}(a_i) -
\hat{F}_{x^n}(a_{i-1})
\label{eqn:F_P_map}
\end{equation}
where, $a_i \in \mathcal{A}^{\Delta}$.
\item[Q2.] The quantization of the values $\hat{P}_{x^n}$ is carried out using a uniform quantizer, $Q_{\delta}$
\begin{equation}
\hat{P}^{\delta, \Delta}_{x^n}  =
Q_{\delta}(\hat{P}^{\Delta}_{x^n})
\label{eqn:F_P_delta_map}
\end{equation}
where, $\delta$ denotes the quantization step-size on the interval $[0,1]$.
\end{itemize}
This is primarily motivated by tractability of the proof of the asymptotic optimality results. But, it can also be argued that any practical implementation of this proposed denoiser only has a finite precision representation of the underlying clean symbol and the distribution function values itself. Analysis of the asymptotic optimality results also lends itself nicely to viewing the distribution of the underlying clean symbol, $\hat{F}_{x^n}$, as the asymptotic limit attained by its quantized, finite precision representation, $\hat{P}^{\delta, \Delta}_{x^n}$. This is formalized in section \ref{sec:dist_approx} where we discuss the precise convergence notion of $\hat{P}^{\Delta}_{x^n}$ to the un-quantized probability measure. 

The minimizer of the Bayes envelope in (\ref{g_opt_defn}) is then constructed from the quantized probability mass function, $P_{x^n}^{\delta, \Delta}$,as $g_{\text{opt}}\left[P_{x^n}^{\delta,\Delta} \right]$, where
$g_{\text{opt}}$ for the quantized clean symbol is,
\begin{eqnarray}
g_{\text{opt}}[P](y) &=& \arg \min_{\hat{x} \in \mathcal{A}^{\Delta}} \sum_{a \in \mathcal{A}^{\Delta}} \Lambda \left( a, \hat{x} \right) \cdot f_{Y|x=a}(y) \cdot P \left( X = a \right)
 \label{eqn:discrete_denoiser}
\end{eqnarray}
$\mathcal{A}^{\Delta}$ is finite alphabet approximation of
$[a,b]$ corresponding to the quantization step size of $\Delta$.
Note that we have extended the definition of  $g_{\text{opt}}$ to accommodate the case when $P$ is not a valid probability, i.e., $\hat{P}_{x^n}^{\delta,
\Delta}$ (it does not sum up to 1). Equipped with $\hat{P}_{x^n}^{\delta, \Delta}$, the candidate for the $n$-block symbol-by-symbol denoiser is now given by
\begin{equation}
\tilde{X}^{n, \delta, \Delta}[y^n](i) = g_{\text{opt}}\left[\hat{P}^{\delta,
\Delta}_{x^n}[y^n]\right](y_i), \qquad 1 \le i \le n
\label{eqn:sym_sym_denoiser}
\end{equation}
where, $g_{\text{opt}}$ is given in (\ref{eqn:discrete_denoiser}). We now proceed to discuss in detail the construction and consistency results of the estimate, $f_Y^n$, $\hat{F}_{x^n}$ and its quantized version, $\hat{P}^{\delta,
\Delta}_{x^n}$.
\subsection{Density Estimation for  independent and non identically
distributed random variables} \label{sec:density_est} We now obtain an estimator $f_Y^n$, for the function in (\ref{eq: creature to be estimated}) which depends on $x^n$ and therefore unknown to the denoiser. Given the memoryless nature
of the channel, the sequence of output symbols, $Y_1, Y_2, \cdots,
Y_n$ are independent random variables taking values in
$\mathbb{R}$, having conditional densities, $f_{Y|x_1}, f_{Y|x_2},
\cdots, f_{Y|x_n}$ respectively. A density estimate is a sequence
$f^1, f^2, \cdots, f^n$, where for each $n$, $f^n_Y(y) = f^n(y;Y_1, \cdots, Y_n)$ is a real-valued Borel
measurable function of its arguments, and for fixed $n$, $f^n$ is
a density estimate on $\mathbb{R}$. The \textit{kernel density
estimate} is given by
\begin{equation}
f^n_Y(y) = \frac{1}{nh^d}\sum_{i=1}^n K \left( \frac{y-Y_i}{h}
\right) \label{density_est}
\end{equation}
where $h = h_n$ is a sequence of positive numbers and $K$ is a
Borel measurable function satisfying $K \ge 0$, $\int K = 1$. The
$L_1$ distance, $J_n$, is defined as
\begin{equation}
J_n = \int \left| f^n_Y(y) -\frac{1}{n} \sum_{i=0}^n f_{Y|x_i}(y)
\right| dy
\label{eqn:J_n}
\end{equation}
The choice of $L_1$ distance as elaborated by the authors in
\cite{Luc} is motivated by its invariance under monotone
transformations of the coordinate axes and the fact that it is
always well-defined. Before proceeding to discuss convergence
results for $J_n$, we present definitions of certain types of
kernel functions, $K$, that are the backbone to
kernel density estimation techniques, \cite{Luc_2}.\\

\begin{definition}
The class of kernels, $\mathcal{K}$ s.t. $ \forall K \in
\mathcal{K}$, we have
\begin{eqnarray*}
\int K = 1
\end{eqnarray*}
and $K$ is symmetric about $0$ are called \textit{class 0
kernels}. \\
\end{definition}

\begin{definition}
A \textit{class s} kernel is a class 0 kernel for which
\begin{eqnarray*}
\int |x|^s|K(x)|dx < \infty
\end{eqnarray*}
and
\begin{eqnarray*}
\int x^i K(x)dx = 0
\end{eqnarray*}
for all $i=1,\cdots, s-1$. Most class 0 kernels are in fact class
2 kernels, the only additional condition being that $\int |x|^2
K(x) < \infty$. However, nonnegative class 0 kernels cannot
possibly of class $s \ge 3$.\\
\end{definition}

\begin{theorem}
 Let $K$ be a nonnegative Borel measurable function on
$\mathbb{R}$ with $\int K =1$ of class $s=2$. Let $f^n_Y$ be the
kernel estimate in (\ref{density_est}) and $J_n$, the
corresponding error as defined in (\ref{eqn:J_n}). Consider
\begin{enumerate}
\item  $J_n  \rightarrow 0$ in probability  as $n \rightarrow
\infty$, for some sequence $\mathbf{x}=(x_1, x_2, \cdots)$
\item  $J_n  \rightarrow 0$ in probability  as $n \rightarrow
\infty$, for all sequences $\mathbf{x}=(x_1, x_2, \cdots)$
\item $J_n  \rightarrow 0$ almost surely  as $n \rightarrow
\infty$, for all sequences $\mathbf{x}=(x_1, x_2, \cdots)$
\item For all $\epsilon > 0$, there exist $r$, $n_0 > 0$
such that $P(J_n \ge \epsilon) \le e^{-rn} $, $n \ge n_0$, for all
sequences $\mathbf{x}$. \item  $\lim_{n \rightarrow \infty} h =
0$, $\lim_{n
\rightarrow \infty} nh = \infty$ \\
\end{enumerate}
Then, 5 $\Rightarrow$ 4 $\Rightarrow$ 3
$\Rightarrow$ 2 $\Rightarrow$ 1.\\
\label{thm:density_consist}
\end{theorem}

The following lemma is key to the proof of Theorem
\ref{thm:density_consist}. \\

\begin{lemma}
\textit{For any family of channel probability density  functions,\\
$\{ f_{Y|x}\}_{x \in [a,b]} $ on $\mathbb{R}$, satisfying assumptions C1-C7, and any
non-negative, integrable function $K$, with $\int K(x)dx=1$,}
condition 4) in Theorem \ref{thm:density_consist} \textit{ holds
whenever}
\begin{equation}
\lim_{n \rightarrow \infty} h_n= 0 \mbox{   and   }  \lim_{n
\rightarrow \infty} nh^d= \infty
\end{equation}
\label{lem:lem_4_pf_dens_const}
\end{lemma}
\begin{proof}[Proof of Theorem \ref{thm:density_consist}]

The implication that 5 $\Rightarrow$ 4 is proved in Lemma
\ref{lem:lem_4_pf_dens_const}. Since clearly, 4 $\Rightarrow$ 3
$\Rightarrow$ 2 $\Rightarrow$ 1, the proof of Theorem
\ref{thm:density_consist} is complete.
\end{proof}

\subsection{Channel Inversion}

The mapping in (\ref{chan_inv_defn}) projects the kernel density estimate
of $\frac{1}{n}\sum_{i=1}^n f_{Y|x_i}(y)$ to an estimate of the
empirical distribution, $F_{x^n}$. This projection is such that it
best approximates (in the $L_1$ sense), the kernel density
estimate with a member in the set of achievable channel output
distributions. From the construction of $f_Y^n$ in
(\ref{density_est}), it is clear that $f_Y^n$ is a bona fide
density on $\mathbb{R}$. Additionally, from the construction of
$\hat{F}_{x^n}$ in (\ref{chan_inv_defn}), we see that for every $F
\in \mathcal{F}_n^{[a,b]}$, $\left[F \otimes \mathcal{C}
\right]_Y$ is also a valid density in $\mathbb{R}$. Finally, from
the definition of the norm, $d$, in (\ref{norm_defn}), it
is true that for $f_Y^n$ and $[F \otimes \mathcal{C}]_Y$ being
bona fide densities on $\mathbb{R}$, $0 \le d \left(f_Y^n, [F
\otimes \mathcal{C}]_Y \right) \le 2$, $\forall$, $n$. These
facts, together with the convexity of $\mathcal{F}_n^{[a,b]}$ show
that the estimator in (\ref{chan_inv_defn}) is well defined. With the Levy metric defined as:
 
\begin{definition}[Levy metric]
The Levy distance $\lambda \left(F,G \right)$ between any two
distributions $F$ and $G$ is defined as
\end{definition}

\begin{equation*}
\lambda \left(F,G \right) = \inf \{ \varepsilon
>0 :F(x-\varepsilon) -\varepsilon \le G(x) \le F(x+\varepsilon)+\varepsilon \text{ for all } x \}
\end{equation*}

we have:\\

\begin{theorem}
For the estimator, $\hat{F}_{x^n}$ defined in equation
(\ref{chan_inv_defn}) we have $\lambda \left(
F_{x^n},\hat{F}_{x^n} \right) \rightarrow 0$ a.s. for all $\mathbf{x} \in [a,b]^{\infty}$
\label{cont_map}
\end{theorem}
The proof of Theorem \ref{cont_map} is discussed in detail in the Appendix \ref{app:proofofthmcont_map}.

\subsection{Distribution-independent Approximation of the Estimate
of the Input empirical distribution} \label{sec:dist_approx}
In this section, we discuss the convergence notion of $\hat{P}_{x^n}^{\Delta}$ to the law corresponding to the un-quantized distribution function $\hat{F}_{x^n}$.
\begin{definition}[$\beta$ metric]
For any two laws $P$ and $Q$ on $S$, $f:S \rightarrow \mathbb{R}$
let $ \int f d \left(P-Q \right) := \int fdP - \int f dQ$, for
bounded $\int f dP$ and $\int f dQ$, the Prohorov metric is
defined as
\begin{equation*}
\beta \left(P,Q\right) = \sup \left\{ \left| \int f d \left(P-Q
\right)\right| : \parallel f \parallel_{BL} \le 1 \right\}
\end{equation*}
where
\begin{equation}
\parallel f \parallel_{BL} = \parallel f \parallel_{L} + \parallel f
\parallel_{\infty}
\end{equation}
and
\begin{equation}
\parallel f \parallel_{L} = := \sup_{x \neq y} \frac{\left| f(x)-
f(y)\right|}{|x-y|} ,\qquad \parallel f \parallel_{\infty} =
\sup_x \left| f(x) \right|
\end{equation}
\end{definition}

Equipped with this definition, we now state the following theorem,
\begin{theorem}
\begin{equation}
\lim_{\Delta \rightarrow \infty} \beta \left( \hat{P}_{x^n}, \hat{P}_{x^n}^{\Delta}\right) = 0
\end{equation}
where, $\hat{P}_{x^n}$ denotes the law associated with the distribution function $\hat{F}_{x^n}$.
\label{thm:dist_approx}
\end{theorem}
\begin{proof}
Follows directly from Lemma \ref{lem:dist_approx}.\\
\end{proof}

\begin{lemma} For any $F \in \mathcal{F}^{[a,b]}$,
\begin{equation}
\lim_{\Delta \rightarrow 0} \beta \left( P, P^{\Delta} \right)= 0
\label{thm_4}
\end{equation}
where $P$ is the law associated with distribution functions in
the family $\mathcal{F}^{[a,b]}$. Particularly, the $F$ and $P^{\Delta}$
that satisfies (\ref{thm_4}) is defined by,
\begin{equation}
P^{\Delta}(a_i) = F(a_i)-F(a_{i-1}) \label{P_delta}
\end{equation}
where $a_i \in \mathcal{A}^{\Delta}$ and $\mathcal{A}^{\Delta}$ is the finite alphabet approximation of $[a,b]$ discussed earlier.
\label{lem:dist_approx}
\end{lemma}
In words, any empirical distribution of the underlying clean sequence is approximated arbitrarily well with a PMF on the quantized set of points when the quantization is fine enough. 

Next we discuss the mechanics of the construction of the denoiser, which has the density estimation and the channel inversion steps as its core.
\subsection{Implementation of the symbol-by-symbol denoiser}
\label{sec:sym-sym denoiser}
The implementation of the denoiser in the previous section involves a discretization of the density estimation and the channel inversion steps. 
The discretized version of the kernel density estimate, $f^n_Y(y)$, in (\ref{density_est}) is evaluated at a set of discrete points, $\{y_1, \cdots, y_N\}$ . This gives an $N$-dimensional vector of the distribution function, $p^n_Y(y)$. The ``channel inversion'' in (\ref{chan_inv_defn}) is also discretized using the estimate, $p^n_Y(y)$. 
\subsubsection{Fast kernel density estimation}
\label{subsec:Fast Kernel Density Est}
The Kernel density estimation in (\ref{density_est}) for a given kernel function, $K$, although simple in construction, is faced with a significant computational burden for a brute-force computation of $O(Nn)$ corresponding to $n$ data points and $N$ points $\{y_1, \cdots, y_N\}$ at which $p_Y^n(y)$ is evaluated. The computational complexity can be greatly reduced by using FFT based methods \cite{Silverman}. Recently, there has been extensive work on the use of fast gauss transform-based techniques \cite{greengard} for reduction of computational complexity. These techniques reduce the complexity from $O(Nn)$ to $O(N+n)$. The cardinal factor in nonparametric density estimation procedures is the choice of the \textit{optimal} bandwidth, $h$, in (\ref{density_est}). There has been some recent work in \cite{GrayandMoore} on using dual-tree methods to derive fast methods for optimal bandwidth choice that continues to maintain the complexity of this step at $O(N+n)$. For $N=O(n)$, this reduces to $O(n)$.

\subsubsection{Channel inversion using linear programming techniques}
\label{subsec:Channel Inversion}
In solving the channel inversion problem in (\ref{chan_inv_defn}), we are looking for a vector in the probability simplex, $\mathcal{F}^{\Delta}=\{P:\sum_{i=1}^{N \left( \Delta \right)} P\left(a_i \right), a_i \in \mathcal{A}^{\Delta}\}$, for our candidate distribution function, $\hat{P}_{x^n}^{\delta,\Delta}$. The discretized version of (\ref{chan_inv_defn}) is given by,
\begin{equation}
\hat{P}_{x^n}^{\delta,\Delta} = \arg \min_{p \in \mathcal{F}^{\Delta}} \sum_{i=1}^N \left| p^n_Y(y_i) - \sum_{j=1}^{N(\Delta)} f_{Y|x=x_
j}(y_i)Q_{\delta} \left(p\left(x_j \right)\right) \right|
\label{eqn:chan_inv_lp}
\end{equation}
The objective function, being an $L_1$-norm, is clearly a convex function (of the input distribution, $p(\cdot)$) and the candidate minimizer also resides in the convex subspace, viz., the probability simplex $\mathcal{F}^{\Delta}$. This can be easily solved using well-studied linear programming algorithms in the broader area of convex optimization techniques. 
The particular reformulation of the problem solved is of the form
\begin{eqnarray}
\hat{P}_{x^n}^{\delta,\Delta}  &=& \arg \min_{p \in \mathcal{F}^{\Delta}} \sum_{i=1}^N \varepsilon_i \nonumber \\
\text{s.t. } & &p^n_Y(y_i) - \sum_{j=1}^{N(\Delta)} f_{Y|x=x_
j}(y_i)Q_{\delta} \left(p(x_j)\right) \le \varepsilon_i \nonumber \\
 & & \sum_{j=1}^{N(\Delta) } f_{Y|x=x_
j}(y_i)Q_{\delta} \left(p(x_j)\right) - p^n_Y(y_i) \le \varepsilon_i \qquad \forall i \in \{1, \cdots, N \}
\label{eqn:LP_symsym}
\end{eqnarray}
The computational complexity of solving this problem using the popular interior point methods \cite{Tsitsiklis} is $O((N+N(\Delta))^3)=O\left((N+\frac{1}{\Delta})^3\right)=O((N+\log n)^3)$. This again, for $N=O(n)$, reduces to $O \left( \left(n + \log n \right)^3 \right)= O (n^3)$.

The two-pronged quantization discussed in the previous section can be naturally built into the optimization problem in (\ref{eqn:chan_inv_lp}) by searching in \begin{equation}\mathcal{F}^{\delta, \Delta} = \left\{Q_{\delta}(P): P \in \mathcal{F}^{\Delta} \right\} \label{eqn:twoprongquant} \end{equation} the set of $N \left( \Delta \right)$-tuples with
components in [0,1] that are integer multiples of $\frac{1}{\delta}$ with
point masses on the set $\mathcal{A}^{\Delta}$. The formulation would then be
\begin{eqnarray*}
\hat{P}_{x^n}^{\delta,\Delta}  &=& \arg \min_{p \in \mathcal{F}^{\delta,\Delta}} \sum_{i=1}^N \varepsilon_i\\
\text{s.t. } & &p^n_Y(y_i) - \sum_{j=1}^{N(\Delta)} f_{Y|x=x_
j}(y_i)p(x_j)\le \varepsilon_i \\
 & & \sum_{j=1}^{N(\Delta) } f_{Y|x=x_
j}(y_i)p(x_j) - p^n_Y(y_i) \le \varepsilon_i \qquad \forall i \in \{1, \cdots, N \} \end{eqnarray*}
This channel inversion is at the heart of the denoiser in (\ref{eqn:discrete_denoiser}) and its simple formulation makes the scheme particularly elegant and practically implementable. The estimate of the empirical distribution in (\ref{eqn:chan_inv_lp}) is then plugged into (\ref{eqn:discrete_denoiser}) to finally give an estimate of the underlying clean symbol according to (\ref{eqn:sym_sym_denoiser}). The denoiser is described as Algorithm  \ref{alg:CUDE} below.
\incmargin{1em} 
\restylealgo{boxed}\linesnumbered 
\begin{algorithm}[h]
\label{alg:CUDE}
\SetKwData{Left}{left} 
\SetKwData{This}{this} 
\SetKwData{Up}{up} 
\SetKwFunction{Union}{Union} 
\SetKwFunction{FindCompress}{FindCompress} 
\SetKwInOut{Input}{input} 
\SetKwInOut{Output}{output} 
\caption{Symbol-by-symbol denoiser in Section \ref{sec:denoiser_construct}} 
\Input{Noisy sequence $y^n$,  channel $\mathcal{C}$} 
\Output{Denoised sequence, $\hat{x}^n$} 
\BlankLine 
\textbf{FIRST PASS} \\
\textbf{Density estimation step}\\
\Input{Noisy sequence, $y^n$}
\Output{Density estimate, $f^n_Y$}
 Determine the optimal bandwidth from any one of the techniques discussed in \cite{Silverman}, e.g., cross-validation\\
 Use techniques discussed in \cite{GrayandMoore} for \textit{fast} evaluation of (\ref{density_est}) \\
\textbf{Channel inversion step}\\
\Input{$f_Y^n$, Quantization resolutions, $\delta$,$\Delta$}
\Output{$\hat{P}_{x^n}^{\delta,\Delta}$}
Construct  an LP (Linear Program) as in (\ref{eqn:LP_symsym}) and use \texttt{linprog} (in MATLAB) or any complex program solver to solve it. Alternatively, use log-barrier methods discussed in \cite{SB}  to solve for the estimate, $\hat{F}_{x^n}$ \\
Use the quantization mapping in (\ref{eqn:F_P_map}) to map $\hat{F}_{x^n}$ to $\hat{P}_{x^n}^{ \Delta}$\\
Then use a uniform quantizer with resolution $\delta$ to get $\hat{P}_{x^n}^{\delta, \Delta} \leftarrow Q_{\delta} \left( \hat{P}_{x^n}^{ \Delta} \right)$\\
\textbf{SECOND PASS} \\
\Input{Noisy sequence, $y^n$, channel $\mathcal{C}$, estimate of input distribution $\hat{P}_{x^n}^{\delta,\Delta}$}
\Output{Denoised Sequence, $\hat{x}^n$}
Use equation (\ref{eqn:discrete_denoiser}), (\ref{eqn:sym_sym_denoiser}) to denoise at every location, $i$\\
\For{$i\leftarrow 1$ \KwTo $n$}{ 
$\hat{x}_i  \leftarrow g_{\text{opt}}[\hat{P}^{\delta,
\Delta}_{x^n}](y_i)$
} 
\end{algorithm}
\decmargin{1em}

\section{Performance guarantees for the Symbol by Symbol denoiser}
\label{sec:Analysis} The main result of this section is Theorem \ref{thm:mainresultsymsym} below, which establishes the universal asymptotic optimality of our proposed symbol-by-symbol denoiser in (\ref{eqn:sym_sym_denoiser}) with respect to the class of symbol-by-symbol schemes. The predominant technical result leading to Theorem \ref{thm:mainresultsymsym} is Theorem \ref{thm:loss_denoisability_dev}. We continue to restrict ourselves to the semi-stochastic setting where the
underlying clean sequence is an unknown, but deterministic,
sequence $\mathbf{x}$. The benchmark performance for the clean
sequence is the minimum possible symbol-by-symbol loss, $D_0\left(
x^n \right)$, defined in Section \ref{sec:problem_setting}. Theorem \ref{thm:mainresultsymsym} shows that our proposed denoiser, $g_{\text{opt}}\left[
\hat{P}_{x^n}^{\delta, \Delta}\right] $, asymptotically (as the number of observations increases) achieves
that benchmark performance.
This is achieved by bounding the deviation of the cumulative loss
incurred by $g_{\text{opt}}\left[ \hat{P}_{x^n}^{\delta,
\Delta}\right] $ from the minimum possible symbol-by-symbol loss
in Theorem \ref{thm:loss_denoisability_dev} for any block length,
$n$. Hence we show that,
$g_{\text{opt}}\left[ \hat{P}_{x^n}^{\delta, \Delta}\right] $
performs essentially as well as the best possible symbol-by-symbol
denoiser, $D_0 \left( x^n \right)$.

In preparation for Theorem \ref{thm:loss_denoisability_dev} let $\mathcal{F}^{\delta, \Delta}$, defined in (\ref{eqn:twoprongquant}), denote the set of probabilities with components in [0,1] that are integer multiples of $\delta$ (defined under Q2. in section \ref{sec:denoiser_construct}). Note that $\hat{P}^{\delta,
\Delta}_{x^n} \in \mathcal{F}^{\delta, \Delta}$, where $\hat{P}^{\delta, \Delta}_{x^n}$ was defined in (\ref{eqn:F_P_delta_map}). Also, let $\mathcal{G}_{\delta, \Delta} = \{g_{\text{opt}}[P]\}_{P
\in \mathcal{F}^{\delta, \Delta}}$ denote the set of all possible denoisers that can be constructed from the members of the set $\mathcal{F}^{\delta, \Delta}$ using (\ref{eqn:discrete_denoiser}). Define $G(\epsilon, B) = \frac{2\epsilon^2}{B^2}$, 
\begin{equation}
\alpha_n \left( \varepsilon, \delta, \Delta, \rho, \gamma \right) = \left[ \frac{1}{\delta}+1\right]^{\Delta} \left[2e^{-G(\epsilon+\delta \Lambda_{\max}, \Lambda_{\max})n} +
e^{-(1-\rho)\frac{n\gamma^2}{2}} \right]+e^{-(1-\rho)\frac{n\gamma^2}{2}}
\end{equation}
\begin{equation}
\nu \left(  \varepsilon, \delta, \Delta, \Lambda, \mathcal{C} \right) = 3\epsilon + 5\delta \Lambda_{\max} + 4\xi_{\Delta}\Lambda_{\max} + 4\lambda(\Delta)(1+\xi_{\Delta})
\label{eqn:nu}
\end{equation}
\begin{equation}
1-\frac{\rho(\epsilon, \delta)}{2} =  \left(1- \frac{6\epsilon}{\delta} \right)^2 
\label{eqn:rho}
\end{equation}
where
\begin{equation}
\xi_{\Delta} = \sup_{x \in [a,b]} \sup_{\hat{x} \in [a,b] \atop
\left| x- \hat{x} \right| \le \Delta } \int \left| f_{Y|x}(y) -
f_{Y| \hat{x}}(y) \right|dy \label{delta_Delta}
\end{equation}
and $\lambda(\Delta)$ is the moduli of continuity defined in (\ref{loss_modcont1}). The Lipschitz norm, $\parallel \Xi \parallel_L$ of $\xi_{\Delta}$ is given by
\begin{equation}
\parallel \Xi \parallel_L = \sup_{0 < \Delta < (b-a)}
\frac{\xi_{\Delta}}{\Delta} 
\label{eqn:delta_L}
\end{equation}
$D_0 \left( x^n \right)$ is the symbol-by-symbol minimum loss of $x^n$ defined in (\ref{1symb_denoisability}). \\
\begin{theorem}
For all $\epsilon >0$,  $\delta
>0$, $\rho=\rho (\epsilon, \delta)$, $\Delta > 0$ and $x^n \in [a,b]^n$ let,
\begin{equation*}
\gamma = \frac{\epsilon}{\left( \parallel \Lambda \parallel_L +
\Lambda_{\max}
\parallel \Xi \parallel_L+ (b-a)\parallel \Lambda \parallel_L \parallel \Xi
 \parallel_L+ \Lambda_{\max} \right)}
\end{equation*}
then, we have
\begin{eqnarray}
\mathrm{Pr} \left( \left| L_{\tilde{X}^{n,\delta, \Delta}}(x^n,Y^n) -
D_0(x^n) \right| >  \nu \left(  \varepsilon, \delta, \Delta, \Lambda, \mathcal{C} \right)  \right) 
\le \alpha_n \left( \varepsilon, \delta, \Delta, \rho, \gamma \right)
\qquad \forall \;n \mbox{ s.t. } nh_n > n_0 \left(\mathcal{C}, \rho, \delta,  K
\right) 
\end{eqnarray}
\label{thm:loss_denoisability_dev}
\end{theorem}
where, $\parallel \Xi \parallel_L$ is defined in (\ref{eqn:delta_L}) and the form of $n_0$ in (\ref{eqn:n_0}).
Note that the tightness condition on the probability measures
associated with the family of the conditional densities of the
channel, $\mathcal{C}$, guarantees that $n_0 \left(\mathcal{C}, \rho, \delta,  K
\right) <\infty$, $\forall \rho \in(0,1)$. 
Theorem \ref{thm:loss_denoisability_dev} formalizes the fact that the probability of deviation of the cumulative symbol-by-symbol loss, $L_{\tilde{X}^{n,\delta, \Delta}}(x^n,Y^n)$ from the minimum possible loss, $D_0(x^n)$ is exponentially small with the block length $n$.

\subsection*{Intuition behind the proof of Theorem \ref{thm:loss_denoisability_dev}}
The benchmark for assessing the performance of the proposed
denoiser is the minimum possible symbol-by-symbol cumulative loss,
$D_0\left(x^n\right)$. It has been shown in (\ref{eqn:sym_sym_denoisability}), that this is the minimum over all measurable mappings, $g: \mathbb{R} \rightarrow 
[a,b]$, of the expected loss under the marginal density induced by
the true distribution of the underlying clean sequence. This has
been further shown in (\ref{eqn:bayes_denoiser}) to be equal to
the expected value of the Bayes envelope under the true
conditional empirical distribution of the underlying clean signal
given the noisy observation. This true conditional empirical
distribution of the underlying clean signal is the quantity that
is unknown to us. However, if we have an estimate of this
conditional empirical distribution that is in some sense ``close''
to the true conditional empirical distribution and asymptotically
is essentially ``it'', we are on the right track. Since this is derived as a function of the marginal empirical distribution of the
underlying clean signal, all that is needed is, ``closeness'' of
the estimate of the marginal distribution of the underlying clean
signal to the true marginal empirical  distribution. The almost sure
convergence of the marginal density at the output of the
memoryless channel gives us, through the mapping in
(\ref{chan_inv_defn}), an estimate of the input empirical
distribution that weakly converges, as shown in Theorem
\ref{cont_map}, to the true empirical distribution of the
underlying clean signal. This then subsequently lends itself to
the convergence of the expected loss under the corresponding induced
densities at the output of the memoryless channel. From
(\ref{eqn:bayes_denoiser}) and (\ref{eqn:sym_sym_denoisability}), the fact that we have well-behaved
(satisfying conditions C1-C7) channel conditional densities,$\{
f_{Y|x} \}_{x \in [a,b]}$, and loss function, $\Lambda$ (satisfying conditions L1-L2), we can
bound the deviation of the expected value of $\mathcal{U} \left(
\left[ F \otimes \mathcal{C} \right]_{X|Y}\right)$ under the two
corresponding induced densities.

The goal, eventually, is to bound the deviation of the cumulative
loss, $L_{\tilde{X}^{n,\delta, \Delta}}$, incurred by the proposed
denoiser in (\ref{eqn:sym_sym_denoiser}) from $D_0 \left( x^n \right)$ as a function of the block
length, $n$. This is done by using Lemmas \ref{lem:loss_continuity}, \ref{lem:los_cont_w_quant} which formalize the deviation bounds of the expected loss under densities induced by weakly converging distributions.
Finally, Lemma \ref{loss_dev_true_estdist} is used to bound the
deviation of the empirical expected loss from the true expected loss. These Lemmas are analogous (in spirit) to the corresponding ones, i.e., Lemmas 1, 2, 3  (for context length, $k=0$) in the discrete-input, general valued output setting in \cite{Dembo}. There are, however, subtle differences in the bounds and the requirements on the channel, loss functions (C1-7, L1-2) that make it possible in this continuous valued setting. The combination of these results is used to bound
the deviation of $L_{\tilde{X}^{n,\delta, \Delta}}$ from $D_0 \left( x^n
\right)$ in the proof of Theorem \ref{thm:loss_denoisability_dev}. 
Take now, $\delta = \delta_n, \Delta = \Delta_n$ such that
$\delta_n \downarrow 0, \Delta_n \downarrow 0$ for all $\epsilon
>0 $ and  
\begin{equation}
\sum_{n=1}^{\infty} \alpha_n \left(\varepsilon, \delta_n, \Delta_n, \rho, \gamma \right) < \infty 
\label{eqn:sum_alpha}
\end{equation}
For example, $\delta_n, \Delta_n = \frac{1}{\log n}$ would satisfy the above requirements of summability and growth for any
$\varepsilon > 0$. With the growth rates that satisfy the summability condition in (\ref{eqn:sum_alpha}) for $\alpha_n \left(\varepsilon, \delta_n, \Delta_n, \rho, \gamma \right)$ let,
\begin{equation}
\hat{X}_{\text{ssuniv}}^n = \tilde{X}^{n,\delta_n, \Delta_n}
\label{eqn:ssuniv}
\end{equation}
where the subscript `ssuniv' stands for symbol-by-symbol universal denoiser.
A direct consequence of Theorem
\ref{thm:loss_denoisability_dev} and the Borel-Cantelli lemma gives us the following main theorem that establishes universal asymptotic optimality of our proposed symbol-by-symbol denoiser for any unknown individual underlying clean sequence, $\mathbf{x}$ .\\
\begin{theorem}
For all $\mathbf{x} \in \mathbb{R}^{\infty}$,
\begin{equation}
\lim_{n \rightarrow \infty} \left[
L_{\hat{X}^n_{\text{ssuniv}}}(x^n,Y^n) - D_0(x^n)\right] = 0 \qquad
a.s.
\end{equation}
\label{thm:mainresultsymsym}
\end{theorem}
\section{Construction of the Universal Denoiser and its performance guarantees}
\label{sec:2kplus1} 

 In this section, we propose an extension of
the symbol-by-symbol denoiser discussed in previous sections to a $2k+1$-length sliding window
denoising scheme, one that competes with sliding window schemes. The performance guarantees made in the
symbol-by-symbol case also hold in the proposed extension. The first result of this section is presented in Theorem 6, which assess the performance of our proposed scheme by showing that it does well relative to that of the best sliding window scheme of order $2k+1$, as would be chosen by a ``genie'' that knows the underlying clean sequence $x^n$. The main result of this section is Theorem \ref{thm:mainresult2kplus1}, which establishes the strong universality of our proposed sliding window denoiser, showing that it does essentially as well as any sliding window scheme, of any order, as the length of the data increases, regardless of what the underlying clean sequence may be. Theorem \ref{thm:mainresult2kplus1} will be shown to be a direct consequence of Theorem  \ref{thm:large_deviation_whole_seq}, analogously as Theorem \ref{thm:mainresultsymsym} of the previous section followed from Theorem  \ref{thm:loss_denoisability_dev}. 
\subsection{Extension to competition with $2k+1$-order sliding window denoisers}
The scheme we propose is pictorially depicted in Fig.~\ref{fig:2Kplus1}
below.
\begin{figure}[h]
\begin{center}
\includegraphics[width=6.5in]{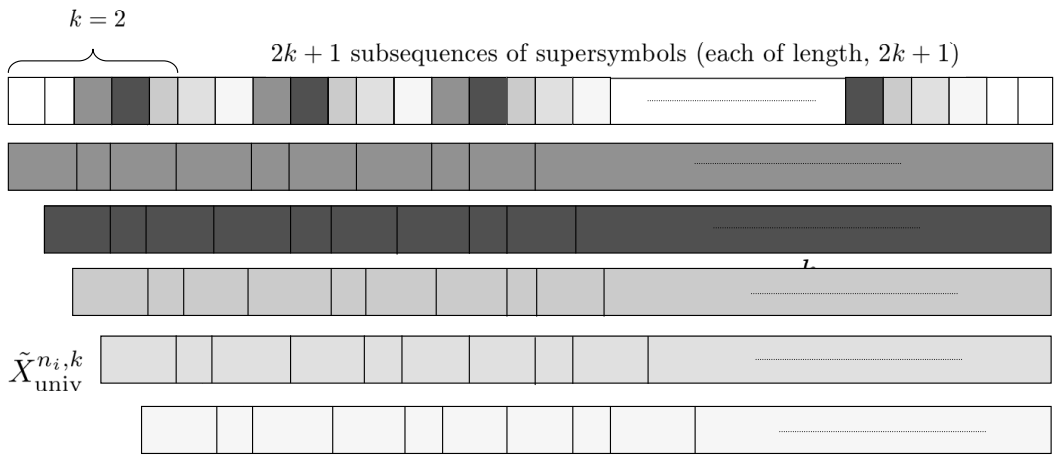}
\caption{Schematic representation of the $2k+1$-length sliding window denoiser}
\label{fig:2Kplus1}
\end{center}
\end{figure}
The necessity for independence of the symbols in the density
estimation procedure discussed in section \ref{sec:density_est}
coupled with the memoryless nature of the channel is the
motivation for partitioning the problem into subsequences that are
processed similarly, but separately. A $2k+1$-tuple super-symbol
is formed by jumping a length of $2k+1$ to achieve the
independence condition between the successive super-symbols. Note that there are $2k+1$ such subsequences and each
subsequence, $i$ (counting in the order of symbols in the
sequence), consists of $\lceil \frac{n-2k-i-1}{2k+1} \rceil$, $2k+1$-tuple
super symbols. We label the subsequences as $x^{n_i}$, for $1 \le
i \le 2k+1$. For a fixed $n$, each subsequence $x^{n_i}$ has the
following super symbols,
\begin{equation*}
x^{n_i} = \left\{ x_i^{2k+i}, x_{2k+1+i}^{4k+1+i}, \cdots, x_{\left(
\lceil \frac{n-2k-1-i}{2k+1}\rceil -1 \right)\left(2k+1 \right)
+i}^{\left( \lceil \frac{n-2k-1-i}{2k+1}\rceil -1 \right)
\left(2k+1\right)+i+2k} \right\}
\end{equation*}
This facilitates the extension of the ideas from the symbols of
the symbol-by-symbol denoiser to the super-symbol of the $2k+1$
sliding window denoiser. Some definitions are in order before we
set to investigate the optimality results of the scheme. As in the
symbol-by-symbol scheme, let $f_Y^{n,k}$ denote the $k^{\text{th}}$ order
density estimate of the noisy sequence of symbols and is computed
exactly as in (\ref{density_est}) except $y, Y_i \in
\mathbb{R}^{2k+1}$. Denote $\mathcal{F}^{[a,b],k}$ to be the set
of all probability distribution functions with support contained
in the hypercube $[a,b]^{2k+1}$. Let $D_k(x^n)$ denote the
$k^{\text{th}}$-order sliding window minimum loss and is defined
as
\begin{equation}
D_k(x^n) = \min_g E \left[ \frac{1}{n-2k} \sum_{i=k+1}^{n-k}
\Lambda(x_i, g(Y_{i-k}^{i+k})) \right]
\end{equation}
Note the similar definition of symbol-by-symbol denoisability in
(\ref{1symb_denoisability}). As before, $D_k(x^n)$ can be
expressed as
\begin{equation}
D_k(x^n) = \min_{g} E_{F^k_{x^n} \otimes \mathcal{C}}
\Lambda(X,g(Y_{-k}^k))
\label{2k_1_denoisability}
\end{equation}
where $F^k_{x^n}$ is the $k^{\text{th}}$ order empirical distribution of
the source. Define further the sliding window denoisability of
the individual sequence $\mathbf{x}=\left(x_1,x_2,x_3, \cdots \right)$ by
\begin{equation}
D \left(\mathbf{x} \right) = \lim_{k \rightarrow \infty}
\limsup_{n \rightarrow \infty} D_k(x^n)
\end{equation}
where the limit exists by monotonicity.  In words, $D(\mathbf{x})$ is the loss of a genie who knows the underlying clean sequence and can choose to denoise with the best sliding window scheme, of arbitrary order.  Extending the definition of $k^{\text{th}}$-order minimum loss to
a subsequence, $x^{n_i}$ as
\begin{equation}
D_k(x^{n_i}) = \min_{g} E_{F^k_{x^{n_i}} \otimes \mathcal{C}}
\Lambda(X,g(Y_{-k}^k))
\end{equation}
The mapping to the corresponding $k^{\text{th}}$ order input empirical
distribution is given by
\begin{equation}
\hat{F}_{x^n}^k[Y^n]=\arg\min_{F \in \mathcal{F}^{[a,b]^k}_n} d
\left(f_Y^{n,k}, \underbrace{\int \prod_{i=-k}^k
f_{Y|x_i}dF(x_{-k}^k)}_{[F \otimes \mathcal{C}]_Y^k} \right)
\label{eqn:chanel_inv_2kplus1}
\end{equation}
where $\mathcal{F}_n^{[a,b],k} \subseteq \mathcal{F}^{[a,b],k}$
denotes the set of $k^{\text{th}}$ order ($1 \le k \le \lfloor \frac{n}{2}
\rfloor$) empirical distributions induced by $n$-tuples with
$[a,b]^{2k+1}$-valued components. $\hat{P}_{x^n}^{\delta,
\Delta, k}$ denotes the $k$-th order estimate of the input
empirical distribution of the source analogously defined as in the
symbol-by-symbol case. The $2k+1$-length sliding window denoiser for each of the subsequences,
$i$, is given by
\begin{equation}
\tilde{X}^{n_i, \delta, \Delta, k}[y^n](j) =
g_{\text{opt}}\left[\hat{P}^{\delta, \Delta,
k}_{x^{n_i}}[y^{n_i}]\right]\left(y_{j-k}^{j+k}\right),  \quad j
\in \left\{k+i, 3k+1+i, \cdots \lceil \frac{n-2k-i-1}{2k+1} \rceil \right\}
\end{equation}
where the $k^{\text{th}}$ order equivalent of the denoiser in  (\ref{eqn:discrete_denoiser}) is given by
\begin{multline}
g_{\text{opt}}[P]\left(y_{-k}^k \right) = \arg \min_{\hat{x} \in
\mathcal{A}}
\Lambda(\cdot,\hat{x})^T[P \otimes \mathcal{C}]_{U|y_{-k}^k} \\
 = \arg \min_{\hat{x} \in \mathcal{A}} \sum_{\hat{a} \in
 \mathcal{A}} \Lambda \left( a, \hat{x} \right) \cdot \left\{ \sum_{u_{-k}^k \in \mathcal{A}^{2k+1}:u_0 =
 a}\left[ \prod_{i=-k}^k f_{Y|x=u_i}(y_i) P \left( U_{-k}^k = u_{-k}^k
 \right)\right]\right\}
 \label{eqn:discrete_denoiser_2kplus1}
\end{multline}
Let, $\mathcal{F}^k_{\delta, \Delta}$ denote the set of $2k+1$-
dimensional vectors with components in [0,1] that are integers
multiples of $\delta$. Note that, $\hat{P}_{x^{n_i}}^{\delta,
\Delta}[z^{n_i}] \in \mathcal{F}_{\delta, \Delta}^k$ for all
$z^n$. Finally, let $\mathcal{G}_{\delta, \Delta}^k =
\{g_{opt}[P]\}_{P \in \mathcal{F}_{\delta, \Delta}^k}$ and
\begin{equation}
\tilde{X}^{n,\delta, \Delta,k} =
\{\tilde{X}^{n_i,\delta, \Delta,k}\}_{1\le i \le 2k+1}
\end{equation}
be our candidate for the $n$-block $2k+1$-length sliding window denoiser. It is the sequence of $2k+1$ denoisers that operate individually on each of the
subsequences. The cumulative loss incurred by this sequence of
denoisers is defined as
\begin{equation}
L_{\tilde{X}^{n,\delta,\Delta,k}} =
\frac{1}{2k+1}\sum_{i=1}^{2k+1}
L_{\tilde{X}^{n_i,\delta,\Delta,k}} \label{eqn:cum_loss_2kplus1}
\end{equation}
where, $L_{\tilde{X}^{n_i,\delta,\Delta,k}}$ is the cumulative
loss incurred by the proposed denoiser for the $i^{\text{th}}$-
subsequence. The following Lemma illustrates a rather intuitive
fact, the average minimum $k^{\text{th}}$ order sliding window
loss incurred by operating on each of the subsequences is at most
the minimum $k^{\text{th}}$ order sliding window loss for the
entire sequence.\\
\begin{lemma}
For all $n \ge 1$, $k \le \lfloor \frac{n}{2} \rfloor$,
\begin{equation}
\frac{1}{2k+2} D_k(x^{n_i}) \le D_k(x^n)
\end{equation}
\label{lem:D_k_seq_subseq}
\end{lemma}

\subsection{Performance guarantees}
In this section we present Theorem \ref{thm:mainresult2kplus1}, wherein we demonstrate that, provided certain growth constraints on the context length $k$, quantization step sizes $\delta$, $\Delta$ and width of the kernel density estimate $h$ are satisfied, the cumulative loss, $L_{\tilde{X}^{n,\delta,\Delta,k}}$, incurred by the proposed
denoiser asymptotically approaches the sliding window denoisability. The growth constraints are specified at the end of this section. They are dictated by an exponential bound on the deviation between 
the cumulative loss, $L_{\tilde{X}^{n,\delta,k,\Delta}}$ and $D_k$ which we now develop.\\
Let
\begin{multline}
\alpha_n \left(\epsilon,k,\delta, \Delta,\rho, \gamma \right) = \nonumber\\
\left[ \frac{1}{\delta}+1 \right]^{{\Delta}^{2k+1}} \cdot \left[
A\left(k,\epsilon+\delta\Lambda_{\max},\Lambda_{\max}\right)\exp\left(-(n+1)G
\left(k, \epsilon+\delta\Lambda_{\max},\Lambda_{\max} \right)
\right)+ \right. \nonumber\\
\left. A\left(k,\sqrt{1-\rho},\frac{2}{\gamma}
\right)\exp\left(-(n+1)G \left(k, \sqrt{1-\rho},\frac{2}{\gamma}
\right) \right) \right] +
e^{-(1-\rho)\frac{(n-2k)\gamma^2}{2(2k+1)}}
\end{multline}
where,
\begin{eqnarray}
A \left( k, \epsilon, B\right) = (2k+1)\exp \left(
\frac{2\epsilon^2}{B^2}\right)\\
G \left( k, \epsilon, B\right) = \frac{2\epsilon^2}{(2k+1)B^2}
\end{eqnarray}
and
\begin{equation}
\nu \left( \varepsilon, \delta, \Delta, \Lambda,\mathcal{C},k \right) = 3\epsilon + 5\delta \Lambda_{\max} + 4\xi_{\Delta}^{2k+1}\Lambda_{\max} + 4\lambda(\Delta)\left(1+\xi_{\Delta}^{2k+1}\right)
\end{equation}
We now state the analogue of Theorem
\ref{thm:loss_denoisability_dev} in the present setting, which bounds the deviation of the cumulative loss incurred by the proposed $2k+1$-length sliding window denoiser from the minimum possible $D_k \left(x^n
\right)$. Note that here, $x \in [a,b]^{2k+1}$ and $Y \in [a,b]^{2k+1}$
( $2k+1$-tuple super-symbols) is the continuous valued output of the memoryless channel.
\begin{theorem}
For all $n \ge 1$, $\epsilon >0$, $\delta
>0$, $\rho=\rho(\epsilon, \delta)$ defined in (\ref{eqn:rho}), $\Delta > 0$, $1 \le k \le \lfloor \frac{n}{2} \rfloor$ and $x^n \in [a,b]^n$
\begin{multline}
\Pr \left(  L_{\tilde{X}^{n,\delta,\Delta,k}}(x^n,Y^n) - D_k(x^n)
 > \nu \left( \varepsilon, \delta, \Delta, \Lambda,\mathcal{C},k \right)  \right) \le \alpha_n \left(\epsilon,k,\delta, \Delta,\rho, \gamma_k \right) \; \forall \; n \mbox{ s.t } nh_n^k > n_k
\left(\mathcal{C}, \rho, \delta,K \right) 
\end{multline}
where,
\begin{equation}
\gamma_k = \frac{\epsilon}{\left( \parallel \Lambda \parallel_L +
\Lambda_{\max}
\parallel \Xi
 \parallel_L^k+ (b-a)\parallel \Lambda \parallel_L \parallel \Xi
 \parallel_L^k+ \Lambda_{\max} \right)}
\end{equation}
$\|\Xi \|_L^k$ (the $k^{\text{th}}$ order equivalent of $\|\Xi \|_L$ in (\ref{eqn:delta_L})) and $n_k
\left(\mathcal{C}, \rho, \delta,K \right)$ are defined in (\ref{eqn:delta_L^k}) and (\ref{eqn:n_0d}) respectively.
\label{thm:large_deviation_whole_seq}
\end{theorem}

Take now, $k = k_n$, $\delta=\delta_n$ and $\Delta= \Delta_n$ such
that $k_n \rightarrow \infty$, $\delta_n \downarrow 0$, $\Delta_n
\downarrow 0$, \[\sum_{n=1}^{\infty} \alpha_n \left(  \epsilon, k_n,
\delta_n, \Delta_n, \rho, \gamma_{k_n} \right) < \infty \] and $n_k
\left(\mathcal{C}, \rho, \delta,K \right) <\infty$. With growth rates that satisfy these conditions let,
\begin{equation}
\hat{X}_{\text{univ}}^n = \tilde{X}^{n,\delta_n,\Delta_n,k_n}
\label{eqn:kthoderuniv denoiser}
\end{equation}
For example, it can be verified that unbounded increasing $k_n=\log
\left( \log(n) \right)$, $h_n = \frac{1}{\log(n)}$, $\delta_n k_n
\rightarrow 0$, $\left(\delta_n, \Delta_n= \frac{1}{\log(n)}
\right)$ satisfies the requirements for a family, $\mathcal{C}$,
that has $\delta_{\Delta_n}^{2k_n+1} \rightarrow 0$ and loss
functions that have $\lambda \left( \Delta_n \right)
\delta_{\Delta_n}^{2k_n+1} \rightarrow 0$. Particularly for additive Gaussian noise channels of finite variance, squared and absolute loss functions
with the aforementioned growth rates of $k_n$, $\Delta_n$, $\delta_n$ satisfy the conditions of $\lambda \left( \Delta_n \right) \delta_{\Delta_n}^{2k_n+1} \rightarrow 0$ and $\delta_{\Delta_n}^{2k_n+1} \rightarrow 0$. \\
We now have the following result as a direct consequence of Theorem \ref{thm:large_deviation_whole_seq} and the Borel-Cantelli Lemma.
\begin{theorem}
For all $\mathbf{x} \in [a,b]^{\infty}$
\begin{equation}
\lim_{n \rightarrow \infty} \left[
L_{\hat{X}^n_{\text{univ}}}(x^n,Y^n) - D_{k_n}(x^n)\right] = 0
\qquad a.s.
\end{equation} \\
\label{thm:mainresult2kplus1}
\end{theorem}
In fact, we can go a step further and show that the $\limsup$ of
the cumulative loss incurred by the proposed denoiser is bounded
by the sliding window denoisability. Specifically,
\begin{corollary}
For all $\mathbf{x} \in [a,b]^{\infty}$
\begin{equation}
\limsup_{n \rightarrow \infty} \left[
L_{\hat{X}^n_{\text{univ}}}(x^n,Y^n) - D(\mathbf{x})\right] \le 0
\qquad a.s.
\end{equation}
\label{cor:2kplus1}
\end{corollary}
which is a corollary of Theorem \ref{thm:mainresult2kplus1}, proved similarly as corollary 1 in \cite{Dembo}.\\

\subsection{Computation complexity of the proposed denoiser}
\label{sec:contcomplexity}
Let us summarize the computational complexity of the proposed denoisers: the ``symbol-by-symbol''  and the $k^{\text{th}}$ order extensions. For the symbol by-symbol denoiser, we have already covered the analysis in Sections \ref{subsec:Fast Kernel Density Est}, \ref{subsec:Channel Inversion}. For $X_{\text{univ}}^n$ defined in (\ref{eqn:kthoderuniv denoiser}), we have:
\paragraph{Symbol-by-symbol scheme}
\begin{enumerate}
\item Fast Kernel Density Estimation, $O(n)$\\
Using the techniques of fast kernel density estimation in \cite{raykar_SDM_2006}, \cite{raykar_TR_4774}, \cite{Graykdtree}, \cite{GrayandMoore} it was shown that the complexity can be reduced from $O(n^2)$ to $O(n)$.
\item Channel Inversion, $O \left(n^3 \right)$\\
The polynomial complexity of the simplex approach in linear programming problems is discussed in detail in \cite{Tsitsiklis}.
\end{enumerate}               
\paragraph{$k^{\text{th}}$ order sliding window scheme}
\begin{enumerate}
\item Fast Kernel Density Estimation, $O \left(n \right)$\\
As before, the complexity of the denoiser continues to be linear in the length of the data, $n$ and the context length, $k$, i.e., $O \left( nk^{\gamma} \right) \; \gamma >0$ \cite{GrayandMoore}.
\item Channel Inversion, $O \left( n^{6k}\right)$ \\
From the fact that the dimensionality of the contexts is length $2k$, the channel inversion now increases in complexity exponentially and is given by $O \left( n^{6k}\right)$. Thus, our schemes are practical for small values of $k$, but become unrealistic to implement as $k$ grows. 

This lead to our follow up work in \cite{Kamakshi_5} that uses quantized contexts in conjunction with the (low complexity) symbol-by-symbol denoiser that asymptotically (with increasing levels of quantization of the contexts) achieves the performance of the sequence of denoisers proposed here.
\end{enumerate}
\section{Universality in the stochastic setting}
\label{sec:univcontstoch}
Our results also imply optimality for the stochastic setting when
the source (clean signal) is a stationary stochastic process
with distribution $F_{\mathbf{X}}$. For the pair $\left( F_{\mathbf{X}}, \mathcal{C} \right)$, define the denoisability, $\mathbb{D}(F_{\mathbf{X}},\mathcal{C})$, as 
\begin{equation}
\mathbb{D}(F_{\mathbf{X}},\mathcal{C}) = \lim_{n \rightarrow \infty}
\min_{\hat{X}^n} E L_{\hat{X}^n} \left(X^n,Y^n \right),
\label{denoisability}
\end{equation}
where the expectation is assuming $X^n$ are the first $n$ symbols
emitted by a source with distribution $F_{\mathbf{X}}$ and $Y^n$ is, as
before, the $n$-tuple of output noisy symbols from the channel
$\mathcal{C}$ that corrupts $X^n$. This is achieved by a ``genie'' that has access to the true distribution, $F_{\mathbf{X}}$, of the underlying clean signal, $\mathbf{X}$. It has been shown in \cite{Tsachy, Dembo} that the
limit in (\ref{denoisability}) exists and hence the
denoisability, $\mathbb{D}(F_{\mathbf{X}},\mathcal{C})$, is well-defined for every stationary $F_{\mathbf{X}}$.

We now state the main result for the stochastic setting wherein we establish that for any stationary underlying clean sequence $\mathbf{X} \sim F_{\mathbf{X}}$, the expected cumulative loss incurred by our proposed scheme asymptotically achieves the denoisability, $\mathbb{D} \left( F_{\mathbf{X}}, \mathcal{C}\right)$.
\begin{theorem}
For all stationary $\mathbf{X}$
\begin{equation}
\lim_{n \rightarrow \infty} E L_{\hat{X}_{\text{univ}}^n} \left(
X^n,Y^n\right) = \mathbb{D} \left(F_{\mathbf{X}},\mathcal{C}\right)
\label{eqn:stoch_mainresult}
\end{equation}
If $\mathbf{X}$ is also ergodic then
\begin{equation}
\limsup_{n \rightarrow \infty} L_{\hat{X}^n_{\text{univ}}}
\left( X^n, Y^n \right) = \mathbb{D} \left( F_{\mathbf{X}},
\mathcal{C}\right) \; a.s. \label{eqn:stoch_mainresult_a}
\end{equation}
\label{thm:stoch_main_result}
\end{theorem}

Given the results established for the semi-stochastic setting, the proof is analogous to that of  Theorem 3 in
\cite{Dembo} except for some subtle differences in our setting due
to the continuous input and output alphabets. We, however, do provide
the proof of the above statement for completeness and for 
accommodating these differences in Appendix \ref{ap:Proof of Stoch Setting}. 

We conclude this section by comparing 
the proposed sequence of denoisers to the DUDE-like
schemes in \cite{Dembo} for the case of finite input (or underlying clean data) and continuous
valued output (noisy data) . By a
minor modification, the proposed denoiser collapses to that in \cite{Dembo} when, as in the setting onf \cite{Dembo}, the channel input alphabet is finite. This is illustrated by
comparing the first pass of the DUDE-like denoiser with a modified
version of the proposed scheme through the schematic
representation in Fig. \ref{fig:dembo_scheme_equiv}. The theoretical details of the equivalence of the modification shown in Fig. \ref{fig:dembo_scheme_equiv} below to the denoiser in \cite{Dembo} are elaborated in Appendix \ref{Proof of Theorem in Dembo}. 

\begin{figure}[h]
\begin{center}
\includegraphics[width=6in,height=4.5in]{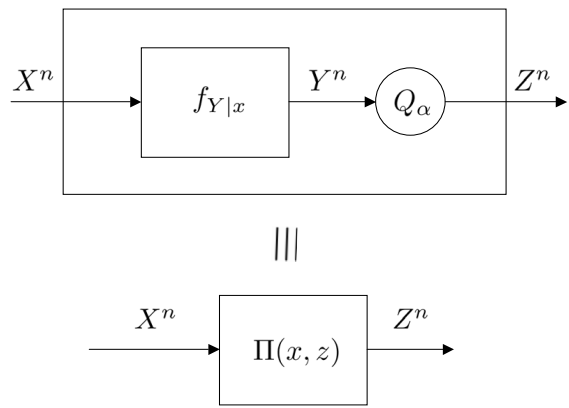}
\caption{Modification to our proposed scheme that is equivalent to that in \cite{Dembo}}
\label{fig:dembo_scheme_equiv}
\end{center}
\end{figure}
\section{Experimental Results}
\label{sec:results}
In this section, we discuss experimental results of applying the
proposed scheme to denoising 256-level gray scale images. We
demonstrate efficacy of the scheme with results of its application
to cases of additive and multiplicative Gaussian noise. In
addition, we consider a highly nonlinear, non-conventional
noise distribution: a locally varying Rayleigh noise
whose variance is a function of the gray level of the underlying
clean image. The first pass of the denoiser is
performed using a Fast Kernel Density Estimation approach proposed
in \cite{gray_moore} and a channel inversion procedure. This
channel inversion is performed using a convex optimization linear
programming technique that maps the output $k^{\text{th}}$-order
density estimate to the corresponding input $k^{\text{th}}$-order
input empirical distribution in accordance with
(\ref{eqn:chanel_inv_2kplus1}). The experimental results presented
in this section have been obtained by implementing the scheme of the previous sections, with no heuristic modifications that are likely to boost the performance. 
The practical implementation aspects are discussed
in greater detail and depth in \cite{Kamakshi_3},
\cite{Kamakshi_5}.

The first example we consider is, denoising of the boats image
that is corrupted by an additive white noise channel (AWGN) with,
$\sigma = 20$. The loss function, $\Lambda$, to be minimized in
this case is the squared error between the true clean image and
our denoised estimate. The denoiser in this case is a mapping from
$\mathbb{R} \rightarrow \mathcal{A}= \{0, \cdots 255 \}$ and
reduces to that in (\ref{eqn:discrete_denoiser_2kplus1}). Results
of the proposed denoising scheme are shown in the Fig.
~\ref{denoise_perform_1} below with context length, $k$, ranging
from 1 to 6. The context (for $k>1$) around any location, $i$, in
the block of noisy data are 2D neighborhoods. The 2D contexts for
various values of $k$ are shown in Fig. \ref{fig:context_len}
below. As is evident from both, the reported Root Mean Squared
Error (RMSE) figures and the perceptual quality, we are able to
achieve improved denoising performance with increasing context
lengths. Finally, we compare the results of the proposed scheme to
that achieved by wavelet-based thresholding scheme \cite{Donoho_1}
and Bayesian Least Squares Gaussian Scaled Mixture (BLS-GSM)
denoiser in \cite{Portilla}. Increasing context lengths, $k$,
translates to accruing increasing $k^{\text{th}}$-order statistics
from the finite block length data. This is the classic trade-off
between increasing context lengths and reliability of the
associated higher order statistics is seen in Fig.
~\ref{fig:denoise_perform_2} where we see only marginal gains in
the RMSE between, $k=4$ and $k=6$. The results for the AWGN case
are primarily aimed at demonstrating the practicality of the
proposed scheme fully acknowledging the performance lead of
schemes like the BLS-GSM that are particularly catered to the
problem of denoising in the case of AWGN channels. The benefits of
the proposed approach are in fact highlighted in unconventional
cases like nonlinear noise channels which will be discussed next.
\begin{figure}[h]
\begin{center}
\includegraphics[width=2.25in]{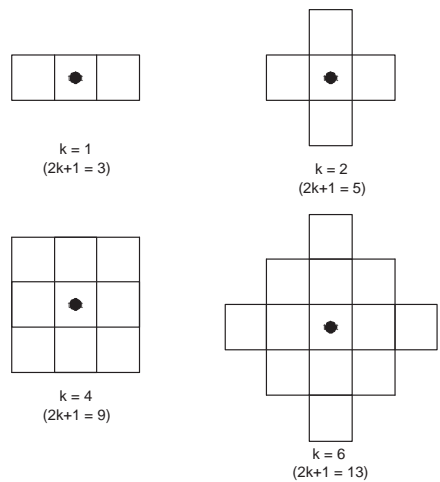}
\caption{2D Contexts for context length, $k$}
\label{fig:context_len}
\end{center}
\end{figure}

\begin{figure}[h]
\begin{center}
\begin{tabular*}{0.75\textwidth}{@{\hspace{0.5cm}}c@{\hspace{0.5cm}}c}
\includegraphics[width=2.25in,height=1.75in]{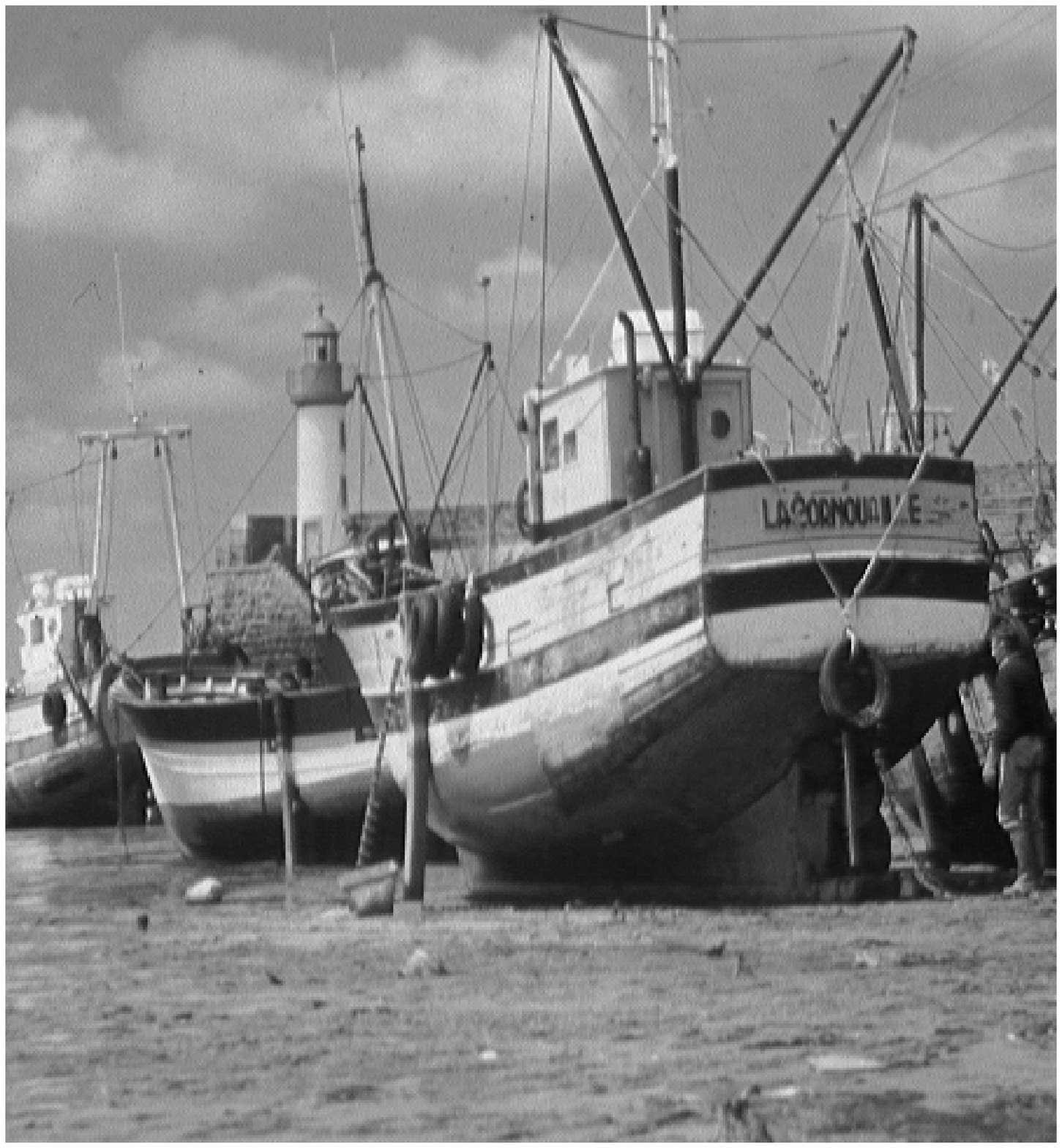} &
\includegraphics[width=2.25in]{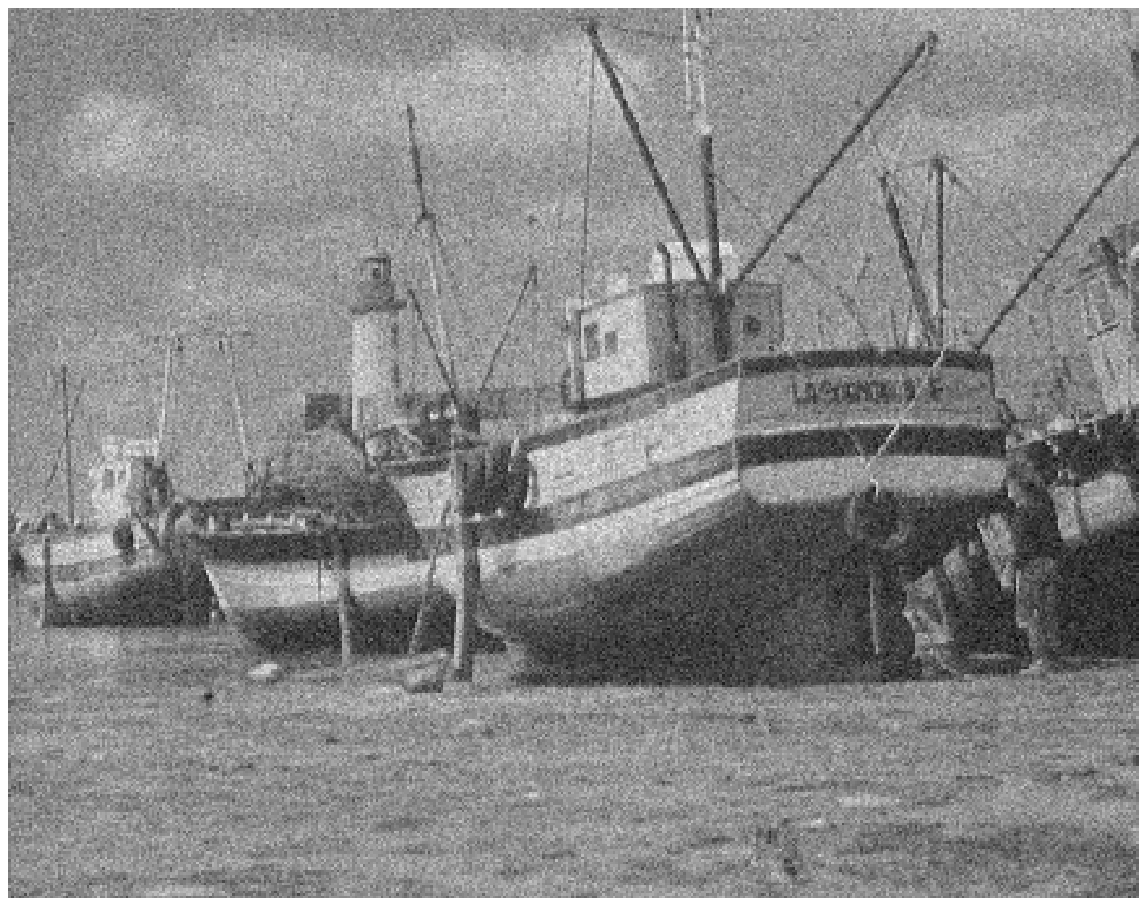} \vspace{0.5cm}\\
\includegraphics[width=2.25in]{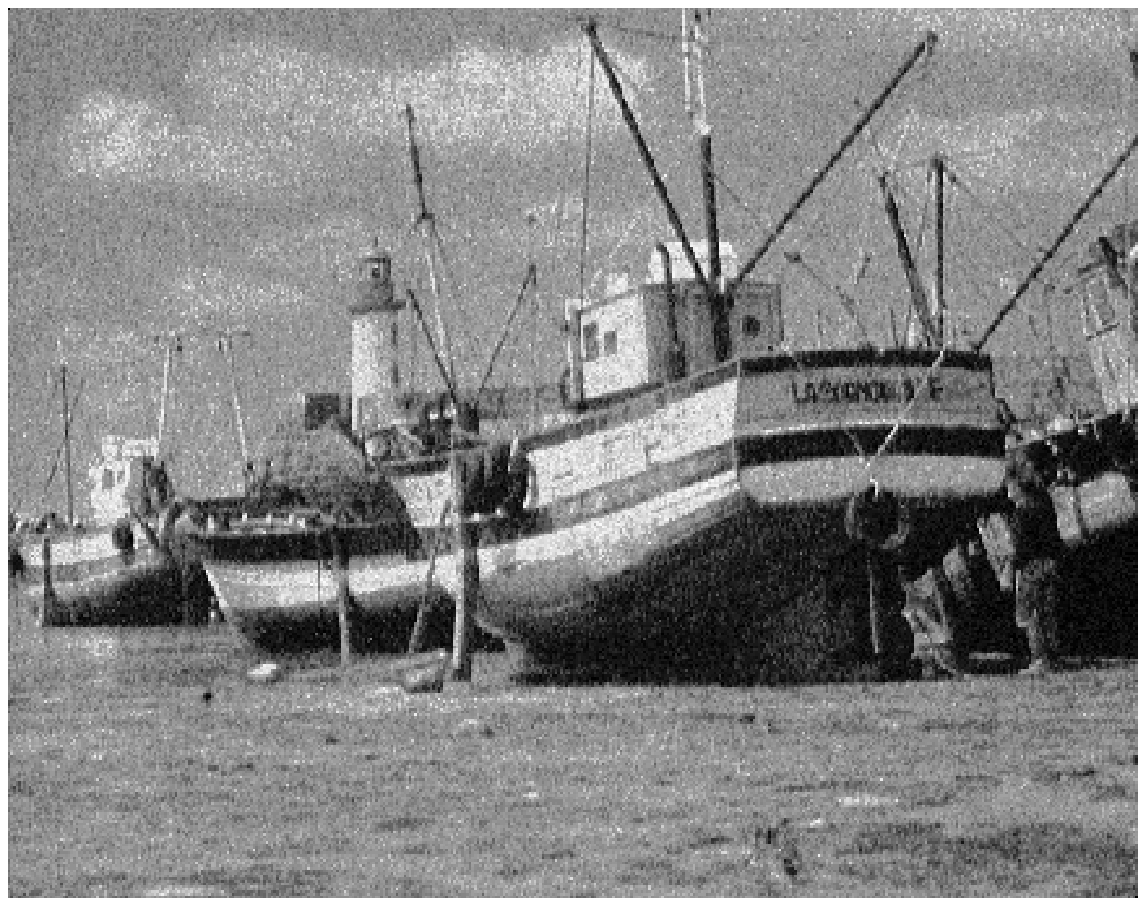} &
\includegraphics[width=2.25in]{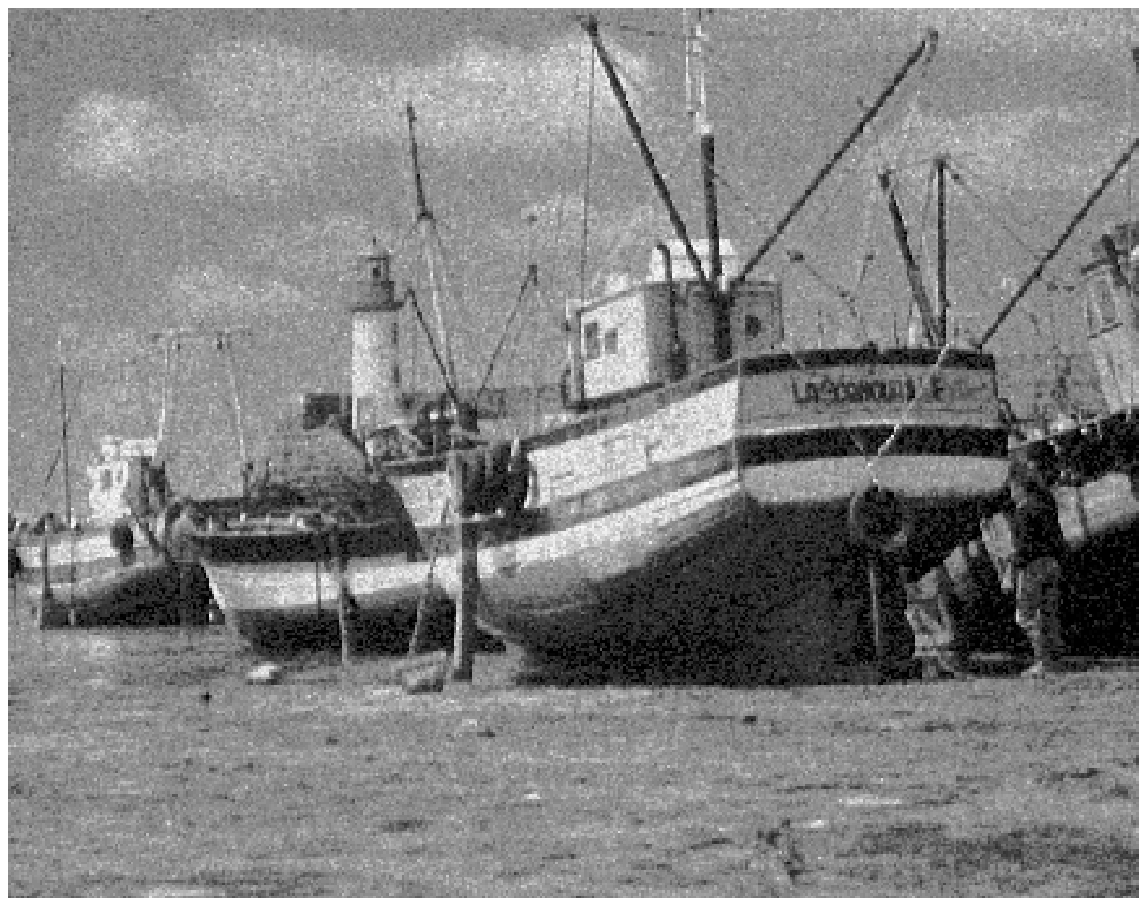} \\
RMSE = 14.7354 & RMSE = 13.0945 \vspace{0.5cm}\\
\includegraphics[width=2.25in]{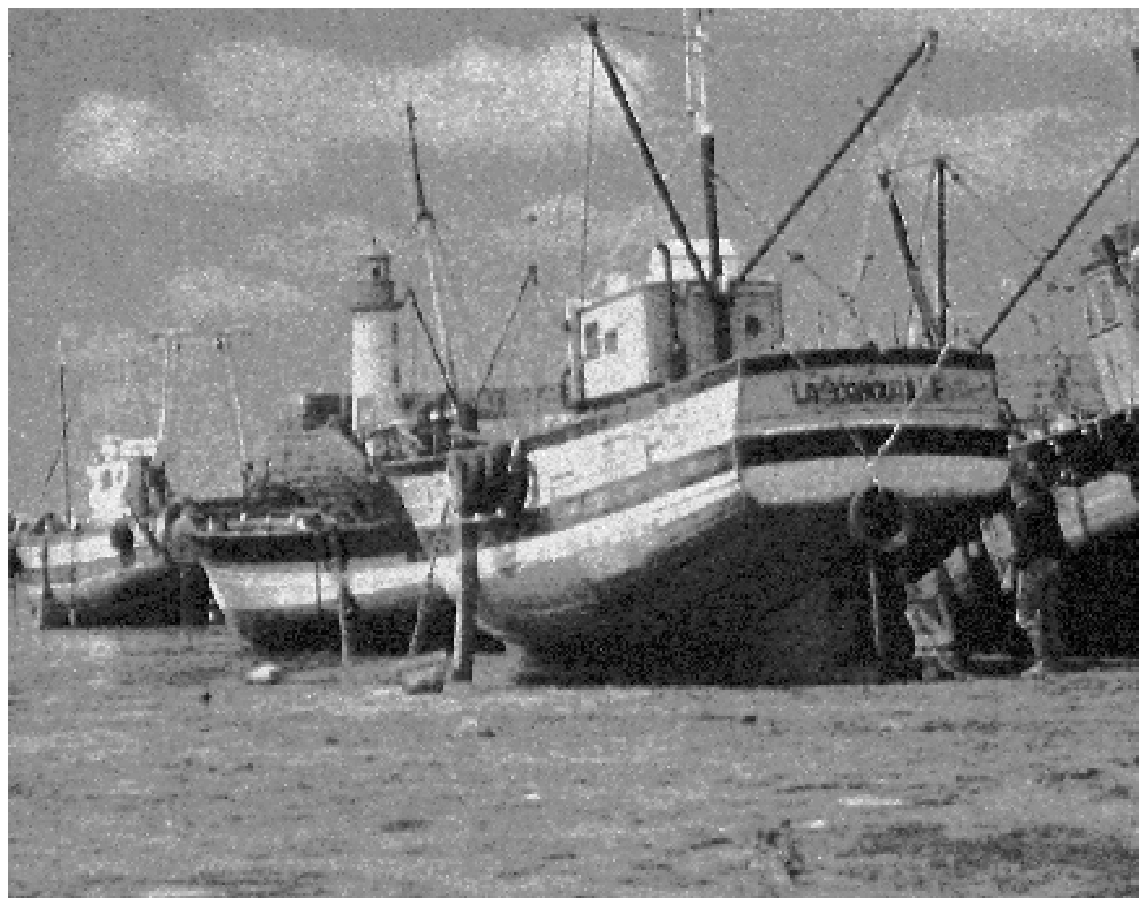} &
\includegraphics[width=2.25in]{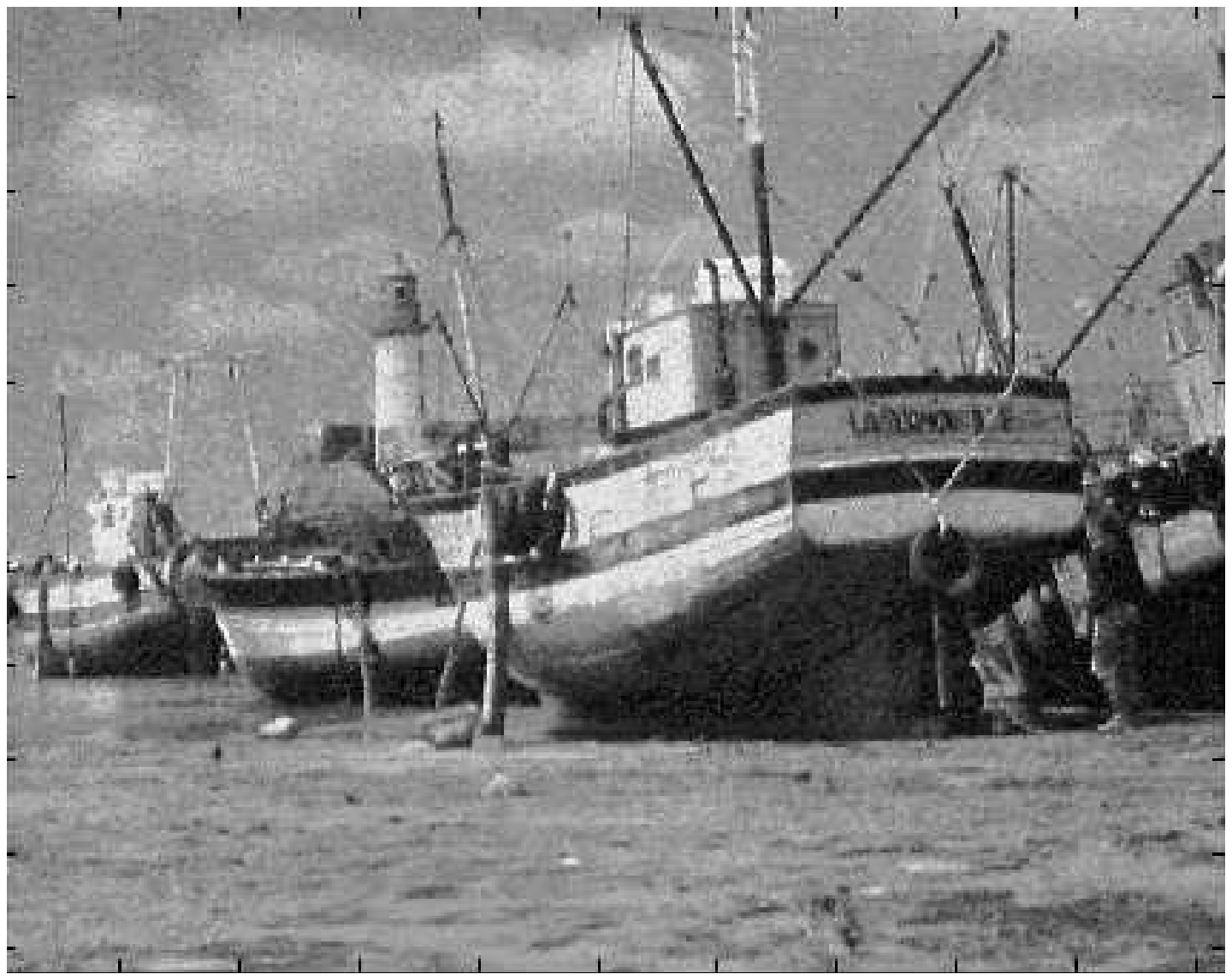} \\
RMSE = 11.2899 & RMSE = 11.2610 \vspace{0.5cm}\\
\includegraphics[width=2.25in,height=1.8in]{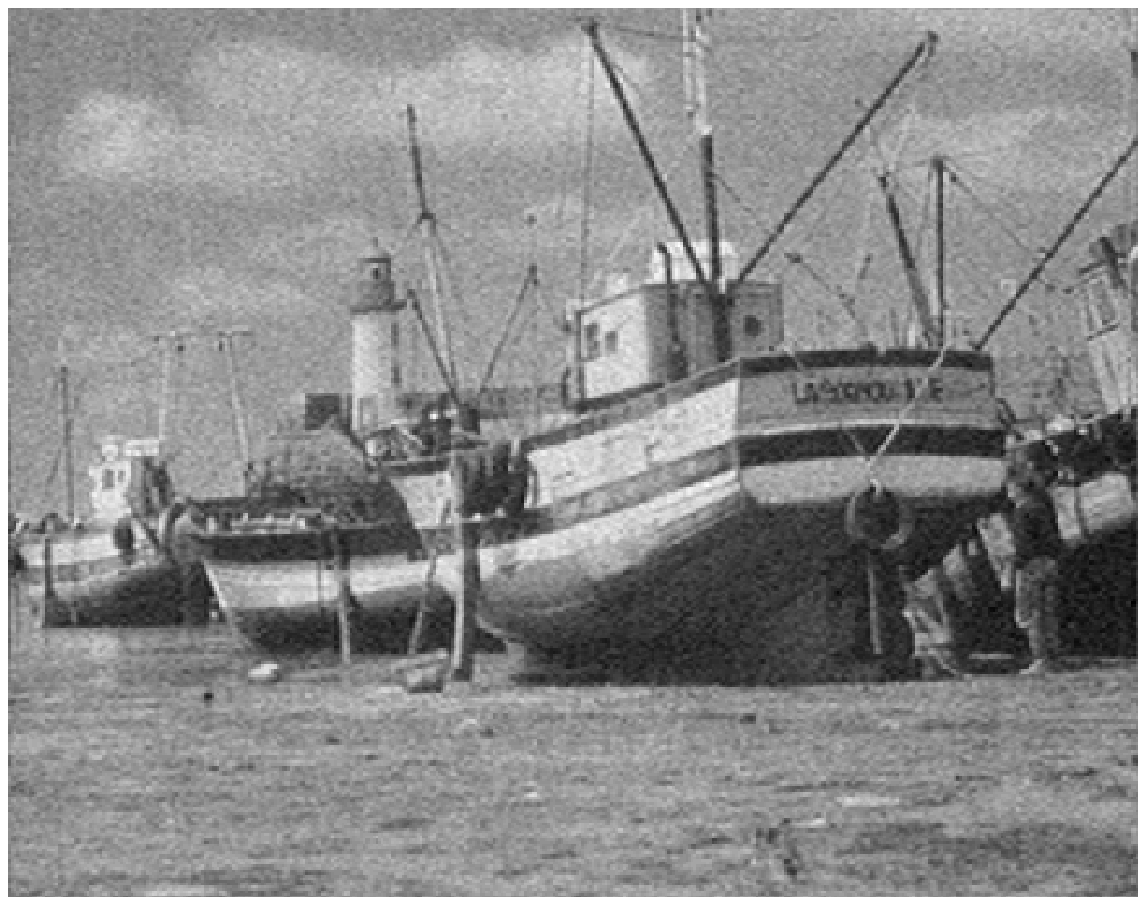} &
\includegraphics[width=2.25in,height=1.8in]{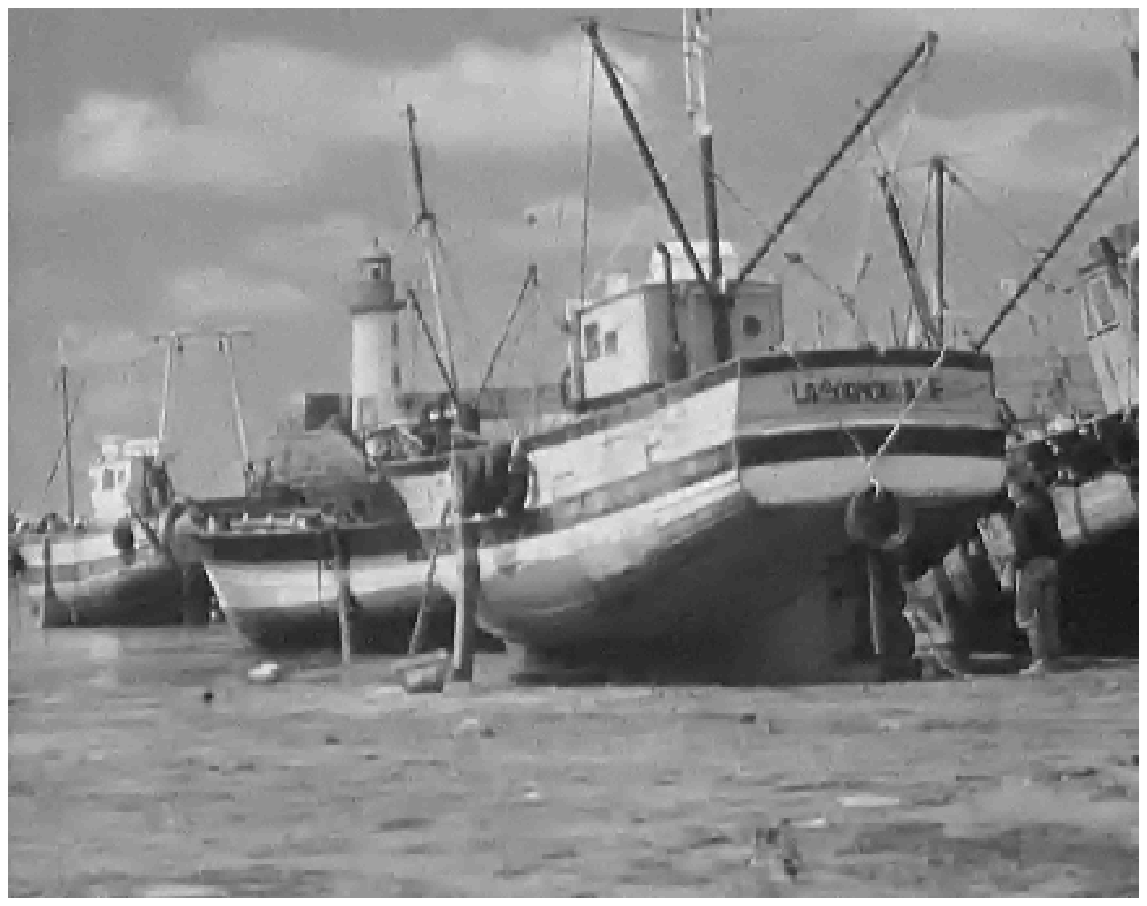}\\
RMSE = 11.1782 & RMSE = 7.842
\end{tabular*}
\caption{Row 1- left: Original image, right: Noisy image, $\sigma
= 20$; Denoised Images using, Row 2- left: $k=1$ right: $k=2$; Row
3- left: $k=4$, right: $k=6$; Row 4- left: the scheme in
\cite{Donoho_1}, right: the scheme in \cite{Portilla}}
\label{denoise_perform_1}
\end{center}
\end{figure}

\begin{figure}[h]
\begin{center}
\includegraphics[width=2.25in]{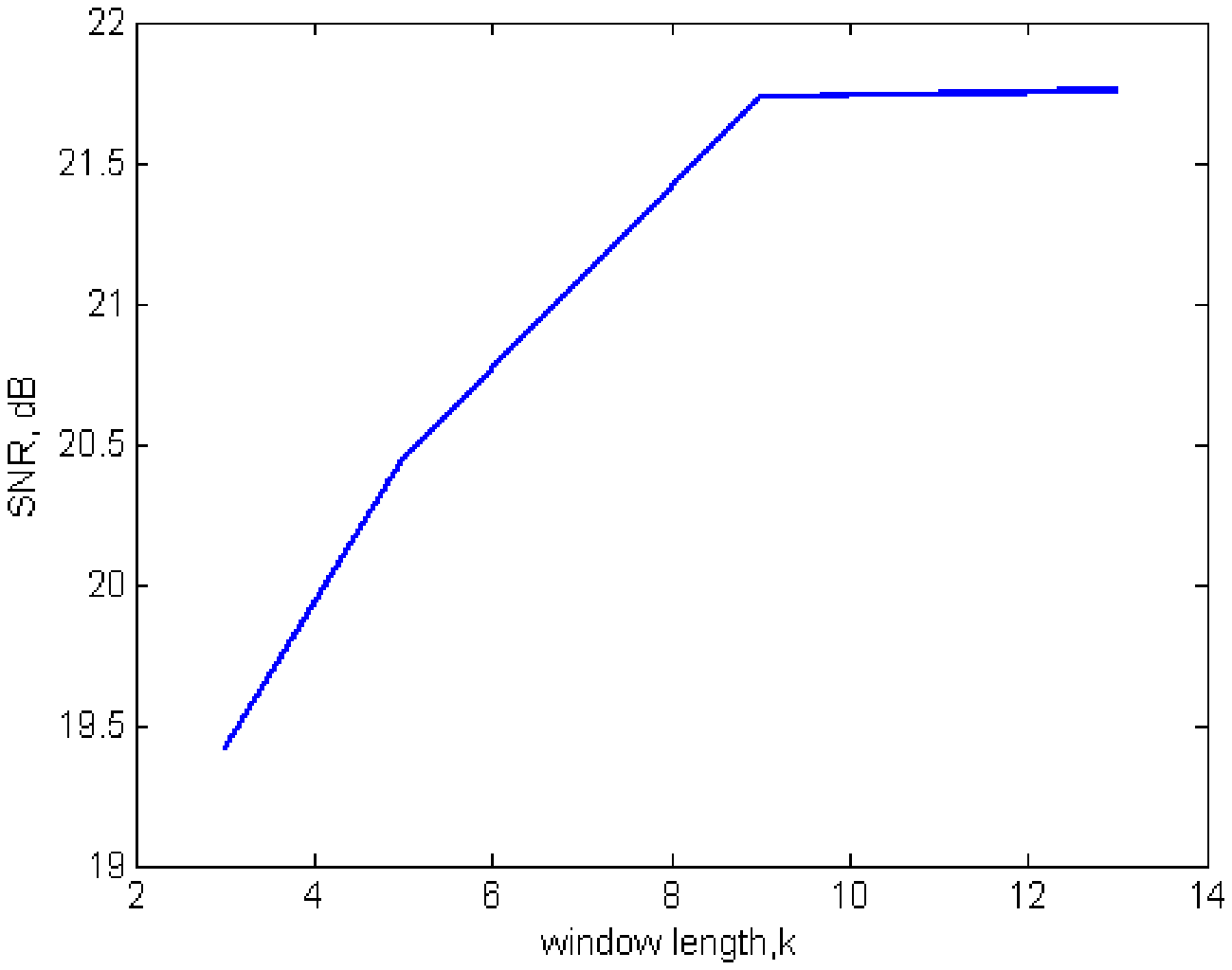}
\caption{Comparison of RMSE of the denoised image for various
context lengths, $k$} \label{fig:denoise_perform_2}
\end{center}
\end{figure}
Another example of the application of the proposed scheme is in
denoising an image corrupted with an unconventional distribution
as discussed earlier in this section. More specifically, we
simulate the noisy image by using a gray-level dependent Rayleigh
distribution (with probability density function, $f(x) =
\frac{x}{b^2} e^{\frac{-x^2}{2b^2}})$ whose variance parameter,
$B$, is chosen as a function of clean image's gray level at that
location. In this particular example, we generate a matrix of
256x256 Rayleigh distributed random variables whose parameters $B$
are chosen according to the following rule, $B(i,j) = I
(i,j)*35/256$, where $I(i,j)$ is the true value of the clean image
at location $(i,j)$. We will discuss the denoising performance
only in the symbol-by-symbol case in this setting in favor of
succinctness to convey the point of efficacy of the proposed
scheme. More detailed results and discussions on this problem
setting can be found in \cite{Kamakshi_3}. We compare, in
Fig.~\ref{fig:pdf_comp}, the empirical distribution estimate,
$\hat{F}_{x^n}$, of the underlying clean image with the histogram
generated from access to the ``true'' clean image. We also compare
these results to the smoothed histogram estimate of the true clean
image that was produced using the Kernel Density estimation
approach in \cite{gray_moore}. From a visual inspection of the
figure, it is evident that we are able to reasonably recover the
true marginal empirical distribution of the underlying clean image
and correspondingly the estimate of the true image.

\begin{figure}
\begin{center}
\begin{tabular*}{0.75\textwidth}{@{\hspace{0.25cm}}c@{\hspace{0.25cm}}c}
\includegraphics[width=2.25in]{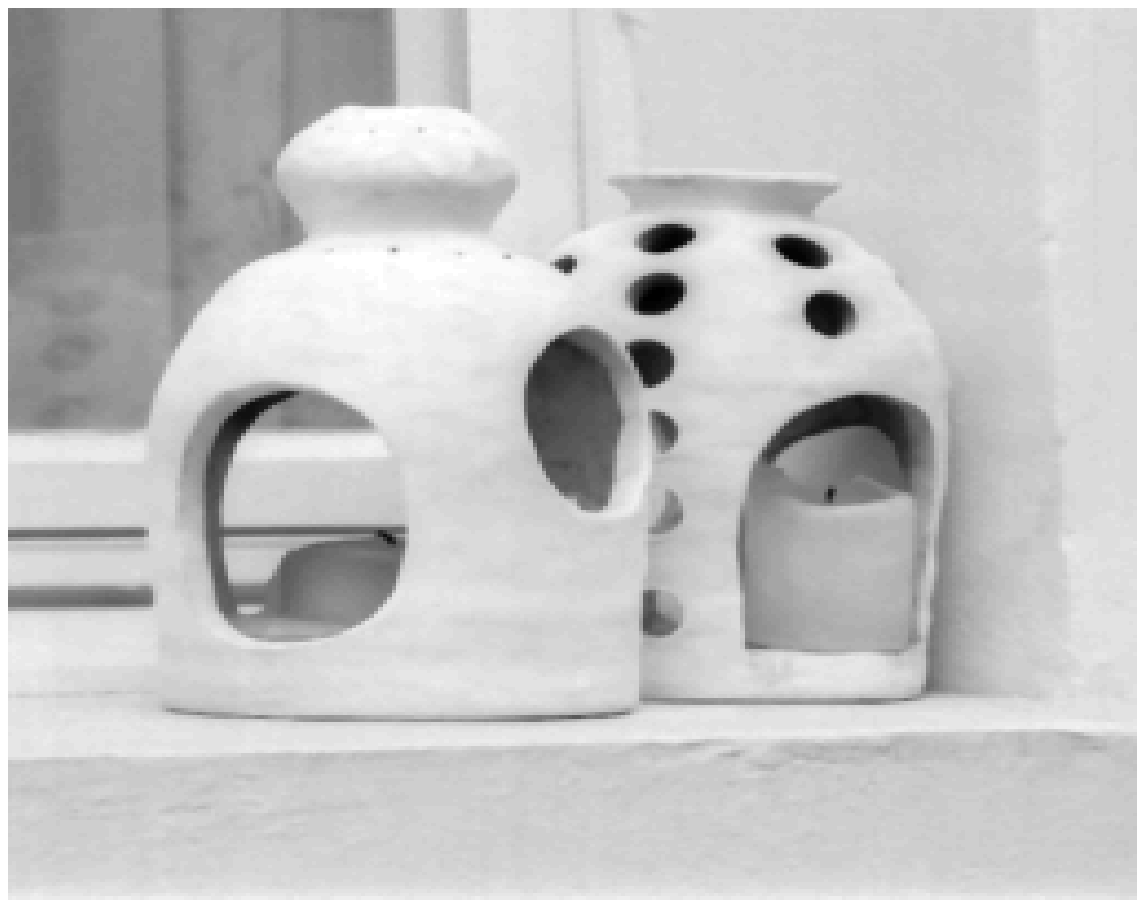} &
\includegraphics[width=2.25in]{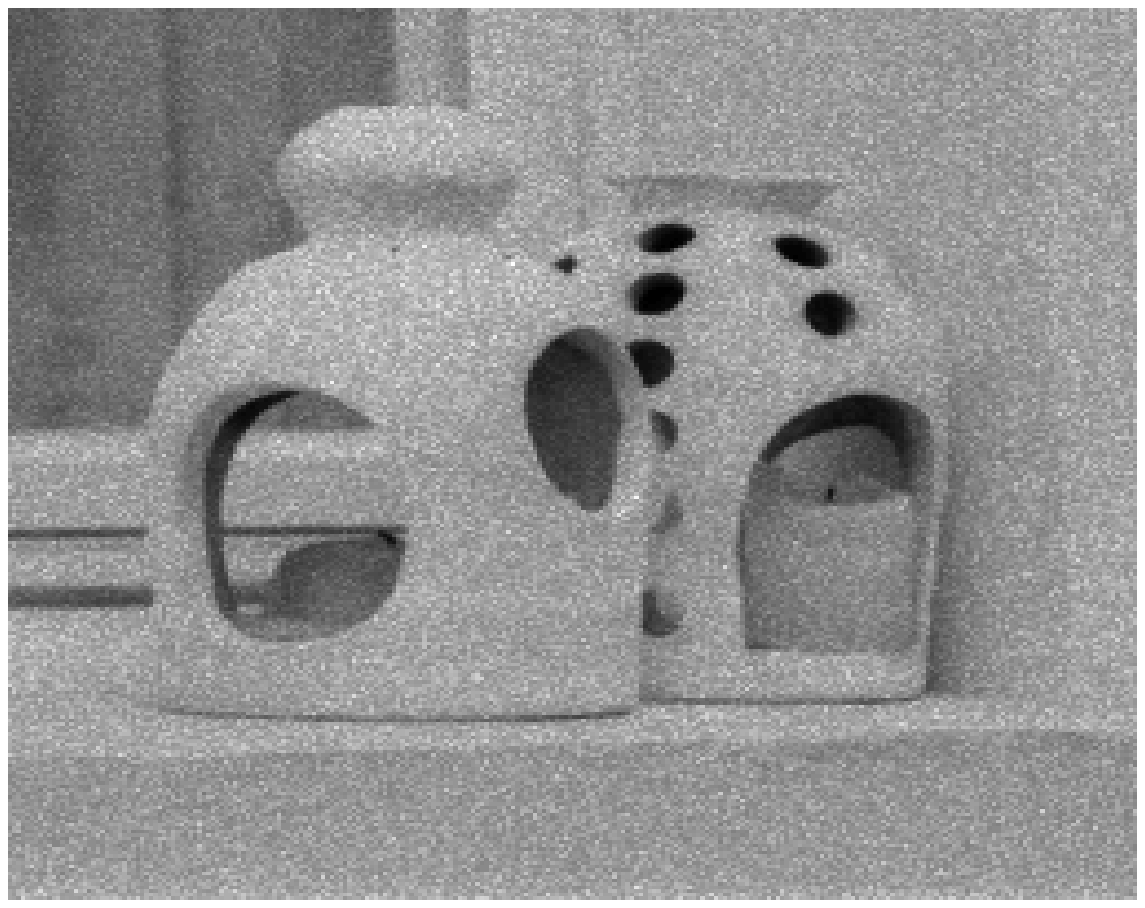} \\
& $\text{RMSE}_{\text{noise}}$ = 38.802\\
\includegraphics[width=2.25in]{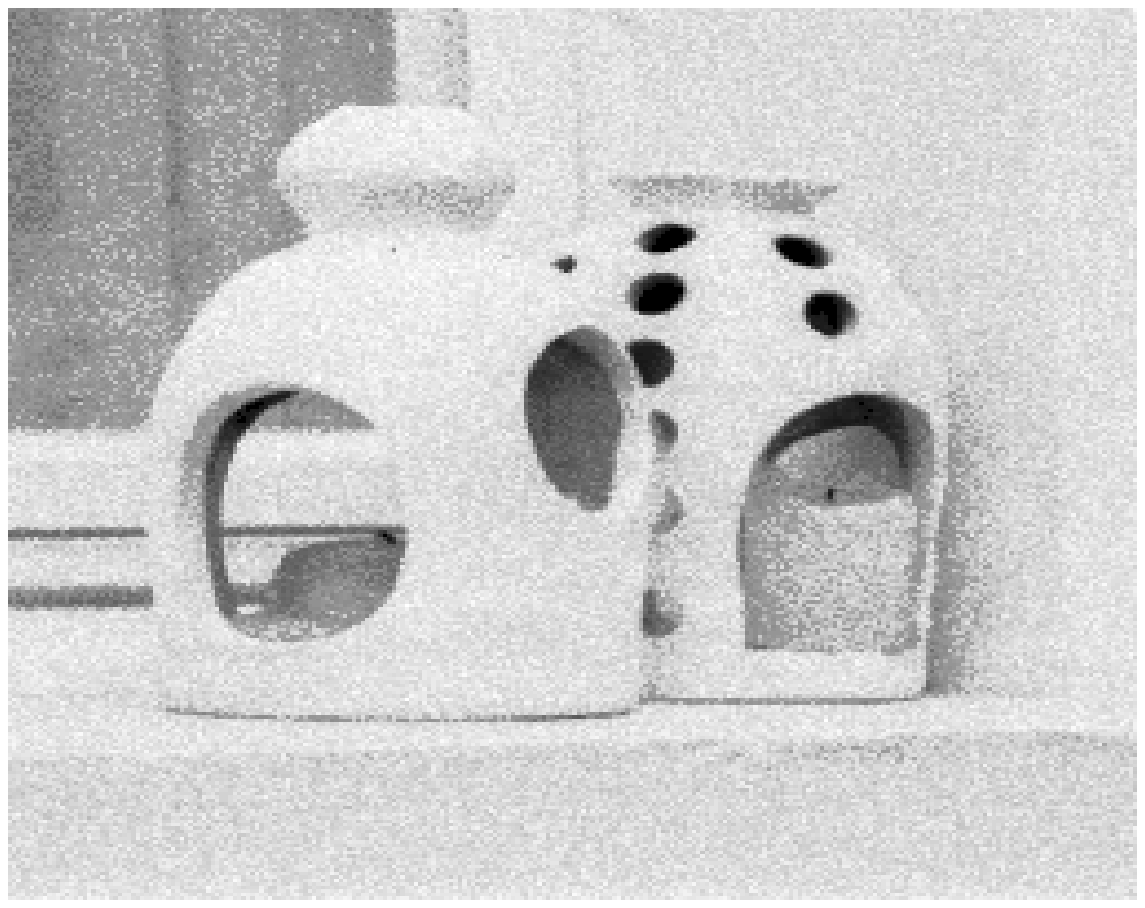} &
\includegraphics[width=2in]{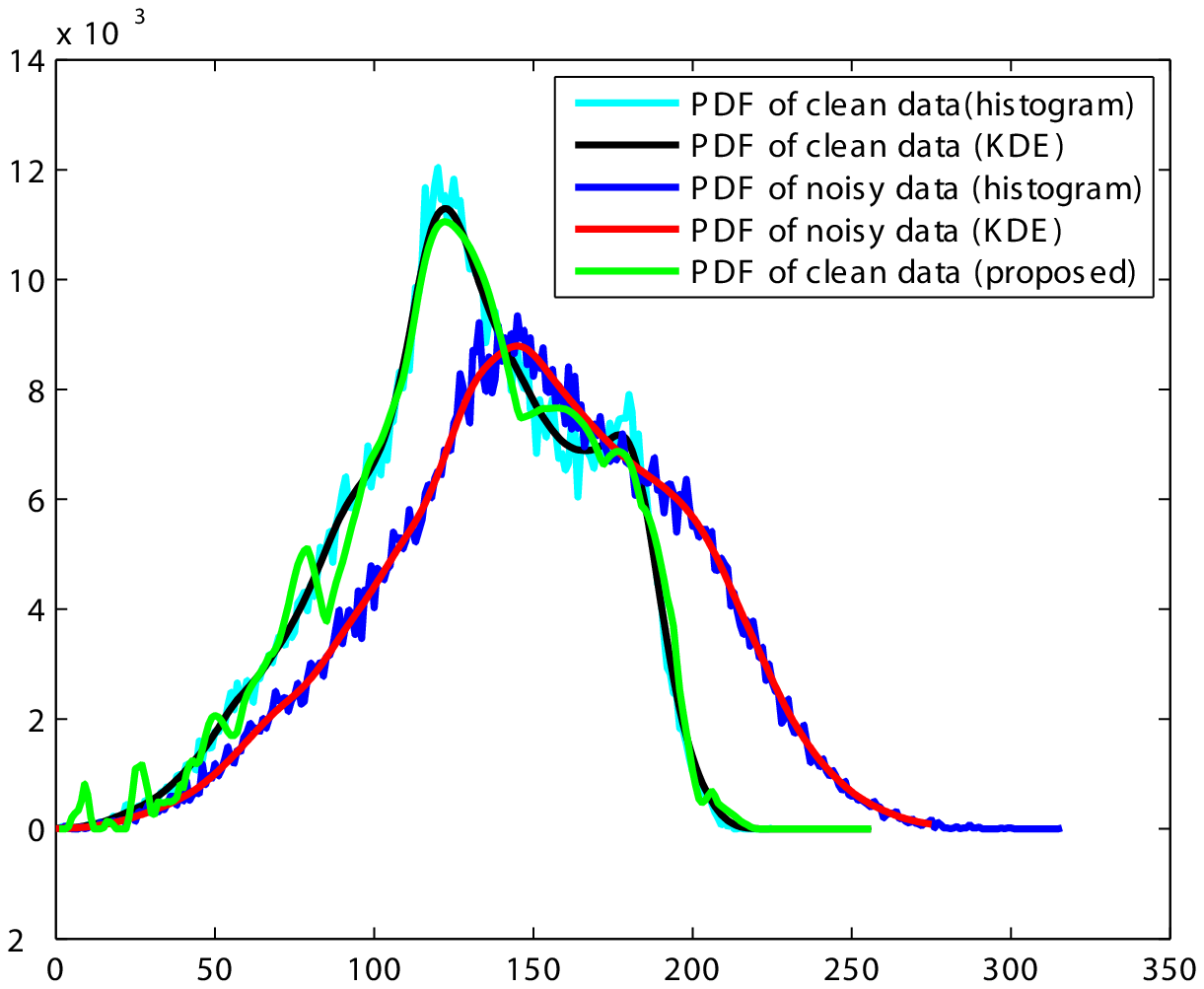} \\
RMSE = 11.2839 &
\end{tabular*}
\caption{Row 1- left: Original image, right: Noisy image; Denoised
images using Row 2- left: symbol-symbol scheme, right: Comparison
of Distribution estimates for the symbol-by-symbol denoiser}
\label{fig:pdf_comp}
\end{center}
\end{figure}

Finally, we present the results of denoising the boats image that
is corrupted by a multiplicative Gaussian noise with a
distribution, $\mathcal{N}(1,0.2)$ in Fig. \ref{fig:mult_noise}.
The noise in this case literally multiplies this case literally
multiplies the original clean image to corrupt it and as such, the
effects are relatively more catastrophic. We compare,
qualitatively, the results from the proposed denoiser with that of
\cite{Portilla} to validate its efficacy.

\begin{figure}
\begin{center}
\begin{tabular*}{0.75\textwidth}{@{\hspace{0.25cm}}c@{\hspace{0.25cm}}c}
\includegraphics[width=2.25in,height=1.75in]{boat} &
\includegraphics[width=2.25in]{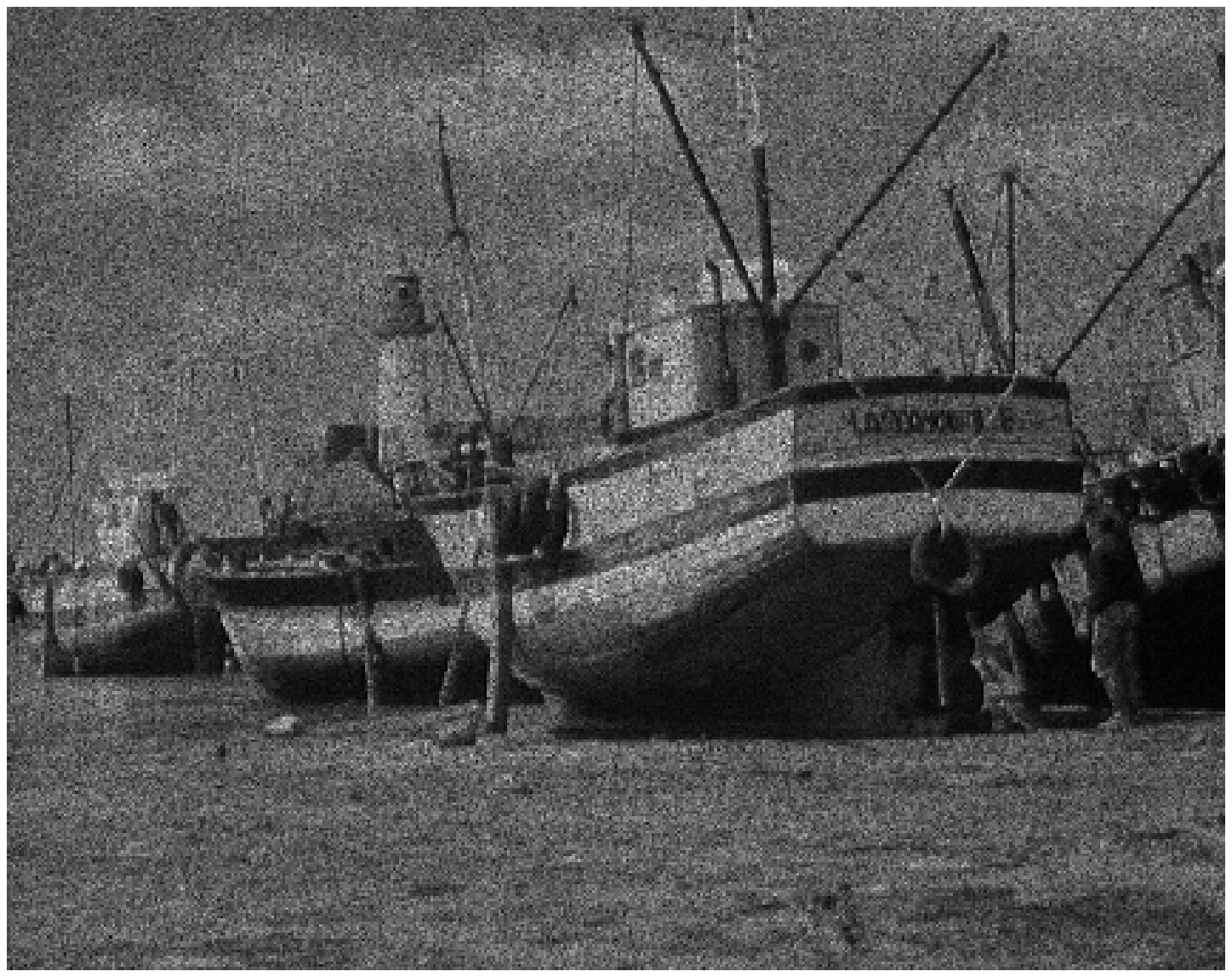} \\
\includegraphics[width=2.25in]{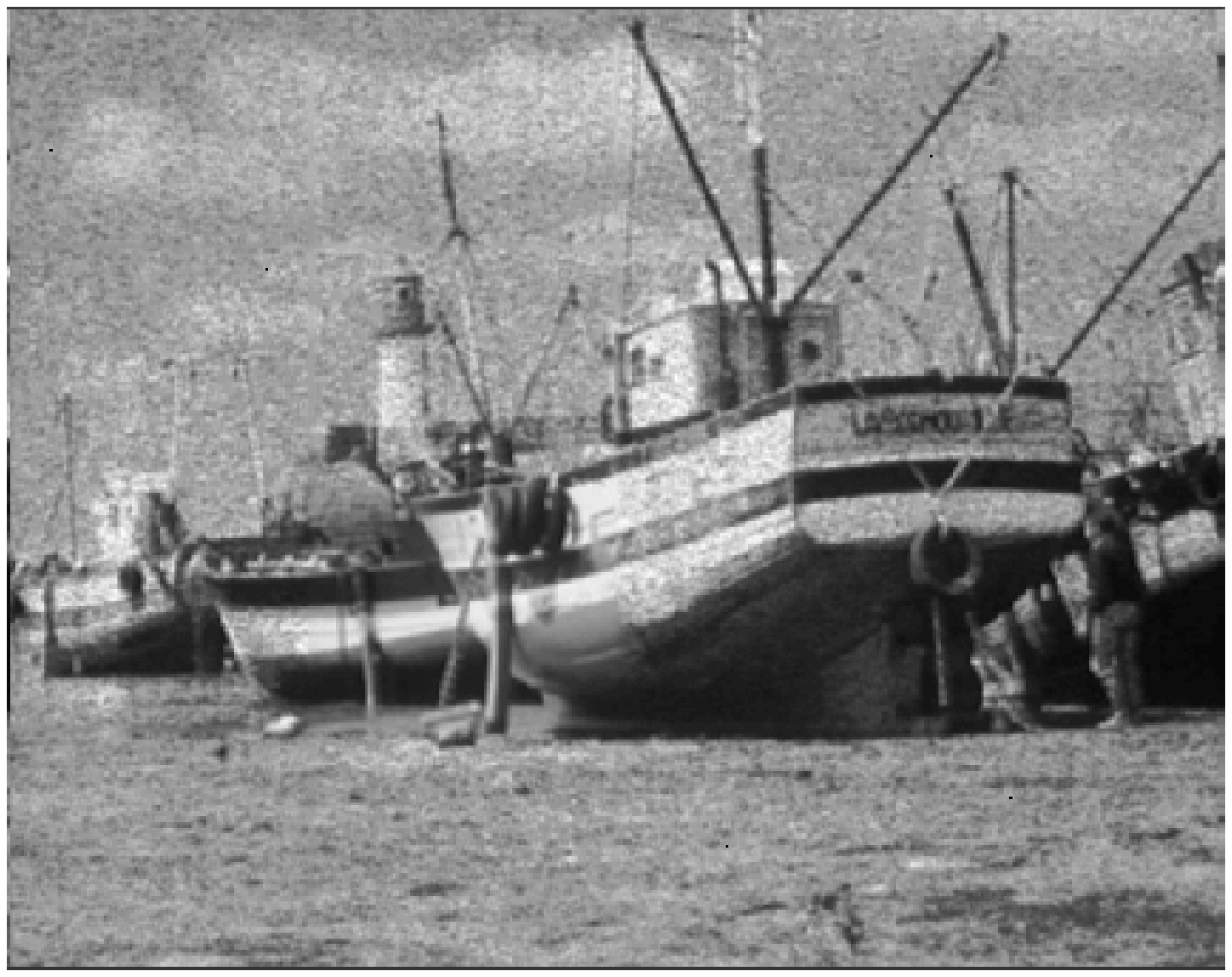} &
\includegraphics[width=2.25in]{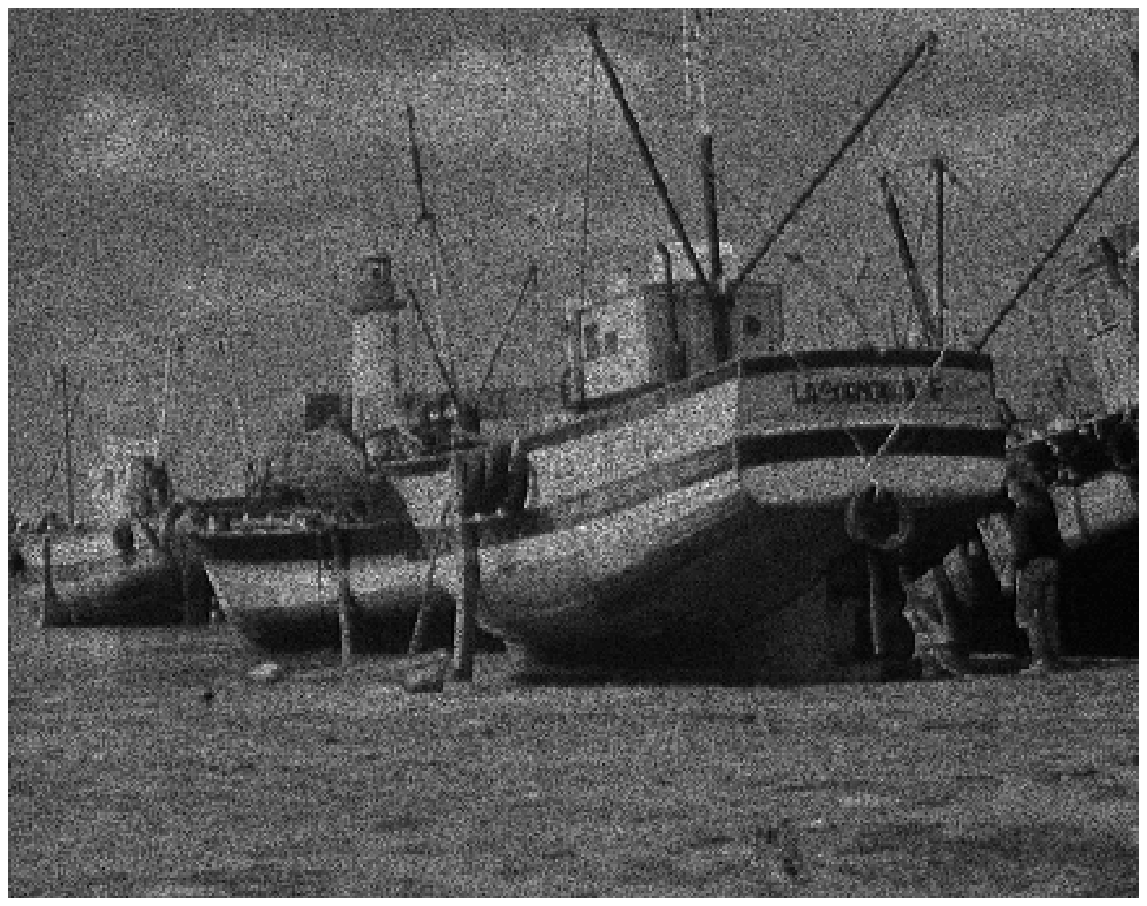} \\
\end{tabular*}
\caption{Row 1- left: Original image, right: Noisy image; Denoised
images using Row 2- left: proposed scheme, right: BLS-GSM
\cite{Portilla}} \label{fig:mult_noise}
\end{center}
\end{figure}

\section{Conclusion and Future Directions}
\label{sec:conclusion}
We have presented a family of schemes for denoising continuous amplitude signals that is universally optimal. A salient feature of our setting and results is the wide generality of channels and loss functions for which they apply. The techniques presented in this paper draw from the ``DUDE framework'' in \cite{Tsachy}. A weighted `context aggregation' was suggested in \cite{Tsachy} as an approach to enhance the performance of the DUDE in the first pass of the statistics collection. The proposed technique provides a natural context aggregation mechanism whereby neighboring contexts in addition to the observed are weighted by the kernel in the density estimation step. The denoiser proposed in \cite{Dembo} was shown to be asymptotically universal and extended the domain of applicability of DUDE-like schemes to cases where the noise is continuous valued. This approach, even though elegant theoretically, suffers from some of the same issues as the DUDE in terms of sparseness of statistics for large alphabet sizes. Our technique addresses this problem for the problem setting considered in \cite{Dembo} by natural context aggregation induced by the kernel density estimation. In the setting where the underlying clean signal is discrete-valued, taking values in a finite alphabet space, a slight modification of our scheme has been shown to reduce to the scheme in \cite{Dembo}. We also simultaneously provide a framework to address the case of continuous valued alphabets, where there is need to learn distribution functions instead of individual mass points as in the discrete-valued case. Finally, the proposed scheme is practical and tractable in its computational requirements as demonstrated by the experimental results.

The experimental results in this paper seem promising enough to motivate further exploration of practical aspects of the proposed scheme. This is an interesting future direction that is currently under investigation. Additional directions of research include studying the applicability of recursive density estimation techniques discussed in \cite{Laslo} in designing recursive denoisers as an alternative to the scheme presented in this paper. This would be particularly useful in multidimensional data applications like denoising noise corrupted video. It could also be of theoretical interest to understand the implications of a recursive structure to the denoiser and its associated optimality results.
\appendices
\section{Conditions on the channel}
\label{app:channel conds}
In addition to conditions C1-C4 in section \ref{sec:problem_setting}, the following conditions on the channel (noise distribution) round up the necessary assumptions for the performance guarantees made in this work.
\begin{itemize}
\item[C5.] The channel satisfies the uniform Lipschitz continuity
condition,
\begin{equation}
\sup_{y \in \mathbb{R}}
\| f_{Y|x}(y) \|_{BL} < \infty 
\end{equation}
where
\begin{eqnarray}{ }
\| f_{Y|x}(y) \|_{BL}  &=& \| f_{Y|x}(y) \|_L  + \| f_{Y|x}(y) \|_{\infty} \\
\| f_{Y|x}(y) \|_L &=& \sup_{x \neq z \atop x,z \in
[a,b]}\frac{\left| f_{Y|x}(y) -
f_{Y|z}(y)\right|}{|x-z|} < \infty , \; \forall y \in \mathbb{R}\\
\| f_{Y|x}(y) \|_{\infty} &=& \sup_{x \in [a,b]} f_{Y|x}(y)
\end{eqnarray}

\item[C6.] The conditional densities, additionally, satisfy the following Lipschitz continuity condition,
\begin{equation}
\parallel \Xi \parallel_L = \sup_{0 < \Delta < (b-a)}
\frac{\xi_{\Delta}}{\Delta} < \infty
\end{equation}
where,
$\xi_{\Delta}$ is defined in (\ref{delta_Delta}).

\item[C7a.] The family of conditional densities, $\mathcal{C}$, have
uniformly bounded second order universal derivatives, i.e.,
$\exists$ a $\mathcal{B}_{\mathcal{C}}$ s.t. $0 < \mathcal{B}_{\mathcal{C}}<\infty$ and $D_2^{\ast}
\left( f_{Y|x} \right) < \mathcal{B}_{\mathcal{C}}, \forall x \in
[a,b] $, where the second order universal derivative is defined as (refer \cite{Luc_2}  for further details)
\begin{equation}
D_2^{\ast} \left( f_{Y|x} \right) = \liminf_{h \downarrow 0 } \int
\left| \left(f_{Y|x} \ast \phi_{h} \right)^{(2)}\right|dy
\end{equation}
$\phi_{h}(x) = \frac{1}{h} \phi \left( \frac{x}{h}\right)$, $\phi \in C^{\infty}$, $C^{\infty}$ is a set of functions that have infinitely
many continuous derivatives with compact support and $f^{(s)}$ denotes the $s$-th derivative of $f$. This is a mild technical condition that enables the proof of the convergence of marginal density estimates  at the output of the memoryless channel to the true marginal density. Note that we are
not imposing the differentiability of the conditional densities of the channel themselves. We are, instead, proposing a milder constraint that the smoothed version of the channel conditional densities is ``differentiable enough''. This condition is trivially satisfied if we have a family of conditional densities
that have a uniformly absolutely continuous derivative.
\item[C7b.] An alternative to the previous condition on the family of conditional
densities of the channel is, $\lim_{|t| \rightarrow 0 }
\Omega_{\mathcal{C}}(t) = 0$, where
\begin{eqnarray}{ }
\Omega_{\mathcal{C}}(t) = \sup_{x \in [a,b]} \omega_x(t)
\end{eqnarray}
and
\begin{eqnarray}{ }
\omega_{x}(t) = \int \left| f_{Y|x}(y-t) - f_{Y|x}(y) \right|dy
\end{eqnarray}
From the fact \cite{Zygmund} that, for any $f \in L_1(\mathbb{R})$, the
corresponding, $L_1$-modulus of continuity, 
\begin{equation*}
\omega(t) = \int
\left| f(x-t) -f(x)\right|dx \rightarrow 0, \; \text{as} \; |t| \rightarrow
0
\end{equation*}
and
\begin{equation*}
\| \omega \|_{\infty}\le 2\|f\|_1< \infty
\end{equation*}
it follows that the global $L_1$-modulus of continuity, $\Omega_{\mathcal{C}}(t)$, is
well-defined for all $t$ and families of conditional
densities, $\mathcal{C}$. In other words, this condition demands
uniform convergence of the $L_1$-moduli of continuity of the
individual members comprising the family of conditional densities.
\end{itemize}

\section{Proof of Lemma \ref{lem:lem_4_pf_dens_const}}
A theorem necessary for the proof of Lemma
\ref{lem:lem_4_pf_dens_const} is as follows

\begin{theorem}
\textit{Every kernel $K$ with $\int K =1, K \ge 0$ is an
approximate identity, i.e for $\lim_{n \rightarrow \infty} h_n = 0
$ and every $f_i \in L_1$, s.t. $D_2^{\ast}\left(f_i \right)<
\infty$ are uniformly bounded we have}
\begin{equation*}
\lim_{n \rightarrow \infty} \int \left|\left(\frac{1}{n}\sum_{i=1}^n
f_i\right) \ast K_{h_n} -\left(\frac{1}{n} \sum_{i=1}^n f_i
\right)\right| =0
\end{equation*}
\label{thm:mix_density_approx_1}
\end{theorem}
An alternate formulation of the approximation identity is the
following,
\begin{theorem}
\textit{Every kernel $K$ with $\int K =1, K \ge 0$ is an
approximate identity, i.e for $\lim_{n \rightarrow \infty} h_n = 0
$ and every $f_i \in L_1$, s.t. $\lim_{|t| \rightarrow 0}
\Omega_{\mathcal{C}}(t) = 0$}
\begin{equation*}
\lim_{n \rightarrow \infty} \int \left|\left(\frac{1}{n}\sum_{i=1}^n
f_i\right) \ast K_{h_n} -\left(\frac{1}{n} \sum_{i=1}^n f_i
\right)\right| =0
\end{equation*}
\label{thm:mix_density_approx_2}
\end{theorem}

A definition regarding the notion of an \textit{associated
kernel}, $L$, with the kernel, $K$ that is necessary for the
subsequent proof is,

\begin{definition}
The function $L$ defined by
\begin{eqnarray*}
L(x)&=&(-1)^s \int_x^{\infty} \frac{(y-x)^{s-1}}{(s-1)!} K(y)dy
\qquad
(x >0) \\
L(-x)&=&(-1)^s L(x) \qquad (x<0)
\end{eqnarray*}
is the kernel associated with kernel $K$. The function $L$ is
sometimes said to have a parameter $s$ since it figures in the
definition of $L$. When $K$ is symmetric, $L$ is symmetric.
\end{definition} Furthermore,

\begin{equation}
\int |L| \le \frac{1}{s!} \int |x|^s |K(x)|dx
\end{equation}
for all nonnegative integers $s$. For $s=0$, we define $L=K$. For
$K \ge 0$, we have the equality

\begin{equation}
\int |L| = \frac{1}{s!} \int |x|^s |K(x)|dx
\end{equation}
Finally,

\begin{eqnarray}
\int L &=& \int\frac{x^s}{s!}  K(x)dx \nonumber \\&=&\left\{\begin{array}
{r@{\quad:\quad}l} 0& s \mbox{   odd}\\
0& s \mbox{   even, and the order of  } K \mbox{  is} > s
\end{array}\right.
\end{eqnarray}

\begin{proof}[Proof of Theorem \ref{thm:mix_density_approx_1}]

Let us start with the case that $f_i$ has $s-1$ absolutely
continuous derivatives. Then, by Taylor's series expansion with
remainder,

\begin{equation*}
f_i(x+y) -f_i(x) = \sum_{j=1}^{s-1} \frac{y^j}{j!} f_i^{(j)}(x)+
\int_{x}^{x+y}\frac{(x+y-u)^{s-1}}{(s-1)!}f_i^{(s)}(u)du
\end{equation*}
so that, for class $s$ kernels $K$,
\begin{eqnarray}
& & \left(\frac{1}{n}\sum_{i=1}^n f_i\right)\ast
K_{h_n}-\left(\frac{1}{n}\sum_{i=1}^n f_i\right) \nonumber \\
&=& \frac{1}{n} \int \left(\sum_{i=1}^n f_i(x+y)-\sum_{i=1}^n
f_i(x)\right)K_{h_n}(y)dy \qquad (\mbox{recall that}
\int K = 1 ) \nonumber \\
 &=& \frac{1}{n} \sum_{i=1}^n\left[\sum_{j=1}^{s-1}0 + \int
 \int_x^{x+y}\frac{(x+y-u)^{s-1}}{(s-1)!}f_i^{(s)}(u)du \;
 K_{h_n}(y)dy\right] \nonumber \\
 &=& \frac{1}{n} \sum_{i=1}^n \left[ \int_x^{\infty} f_i^{(s)}(u)\int_{u-x}^{\infty} \frac{(x+y-u)^{s-1}}{(s-1)!} K_{h_n}(y)dy \; du \right. \nonumber\\
 & & \left. -\int_{-\infty}^{x} f_i^{(s)}(u)\int_{-\infty}^{u-x} \frac{(x+y-u)^{s-1}}{(s-1)!}
 K_{h_n}(y)dy \; du \right] \nonumber \\
 &=& \frac{1}{n} \sum_{i=1}^n \left[ \int_x^{\infty} f_i^{(s)}(u) (-1)^s (L)_{h_n}(u-x)du
 \right.\nonumber \\
 & & \left. - \int_{-\infty}^{x} f_i^{(s)}(u) (-1)(-1)^s (-1)^s (L)_{h_n}(x-u)du
 \right] \nonumber \\
 &=& \frac{1}{n} \sum_{i=1}^n \left[\int_{-\infty}^{\infty} f_i^{(s)}(u) \;
 (L)_{h_n}(x-u)du\right] \nonumber  \\
 &=& \frac{1}{n}\sum_{i=1}^n  h^s f_i^{(s)}\ast L_{h_n}
 \label{f_i_Lhn_relation}
\end{eqnarray}
where $(L)_{h_n}$ is the kernel associated with $K_{h_n}$ and $L$
is the kernel associated with $K$. Therefore, by Young's
inequality \cite{Rudin},
\begin{eqnarray}
\int \left|\left(\frac{1}{n} \sum_{i=1}^n f_i \right) \ast K_{h_n} -
\left( \frac{1}{n}  \sum_{i=1}^n f_i \right) \right| = 
\int \left|\frac{1}{n}
\sum_{i=1}^n  {h_n}^s f_i^{(s)}\ast L_{h_n} \right| \nonumber \\
 \le \frac{h_n^s}{n} \int
\left|\sum_{i=1}^n f_i^{(s)}\right|  \int |L|  \nonumber \\
 \le \frac{h_n^s}{n}
\left(\sum_{i=1}^n \int \left|f_i^{(s)}\right| \right) \int |L|
\label{approx_ident}
\end{eqnarray}

Since $f_i$'s have $(s-1)$ absolutely continuous derivatives,
$\int|f_i^{(s)}| < \infty$, and further if $\int|f_i^{(s)}| < M <
\infty$, $\forall i$ (uniformly bounded) the inequality in
(\ref{approx_ident}) simplifies to
\begin{eqnarray}
\int \left|\left(\frac{1}{n} \sum_{i=1}^n f_i \right) \ast K_{h_n} -
\left( \frac{1}{n}  \sum_{i=1}^n f_i \right) \right|  \le h_n^s M \int
|L| \label{approx_ident_1}
\end{eqnarray}

Since,
\begin{equation}
\int |L| \le \frac{1}{s!} \int |x|^s |K(x)|dx = B_K< \infty
\end{equation}
for $K$ being an order s kernel, inequality in equation
(\ref{approx_ident_1}) becomes
\begin{eqnarray}
\int \left|\left(\frac{1}{n} \sum_{i=1}^n f_i \right) \ast K_{h_n} -
\left( \frac{1}{n}  \sum_{i=1}^n f_i \right) \right|  \le h_n^s M B_K
\label{approx_ident_2}
\end{eqnarray}
Taking limit $ n \rightarrow \infty$ on either sides, we get
\begin{eqnarray}
0 \le \lim_{n \rightarrow \infty}\int \left|\left(\frac{1}{n}
\sum_{i=1}^n f_i \right) \ast K_{h_n} - \left( \frac{1}{n}
\sum_{i=1}^n f_i \right) \right|  \le \lim_{n \rightarrow \infty}h_n^s M
B_K = 0 \label{approx_ident_3}
\end{eqnarray}

This can be extended to the general $f_i$'s using the universal
derivative defined earlier. As a reminder,
\begin{equation}
D_s^{\ast} \left( f_i \right) \triangleq \liminf_{h \downarrow 0}
\int \left| \left( f_i \ast \phi_h \right) ^{(s)} \right|
\end{equation}
where, $\phi$ is a mollifier.

Mollifiers are class 0 kernels, nonnegative and zero outside
$[-1,1]$. They also have infinitely many continuous derivatives
and is called a \textit{mollifier} because of its exceptional
smoothing properties. An example of a mollifier is
\begin{equation}
K(x)=Ce^{-\frac{1}{1-x^2}}, \; |x| \le 1
\end{equation}

For a class $s$ kernel, $K$, and a family of density functions $\{
f_i \}_{i \in \mathbb{N}}$ with associated universal derivatives
that are uniformly bounded, i.e., $D_2^{\ast} \left( f_i \right) <
\mathcal{B}_{\mathcal{C}} < \infty $, $\forall i \in \mathbb{N}$,
it can then be shown that,
\begin{eqnarray}
\int \left| \left(\frac{1}{n}\sum_{i=1}^n f_i\right)\ast
K_{h_n}-\left(\frac{1}{n}\sum_{i=1}^n f_i\right) \right| &\le&
\frac{1}{n}\sum_{i=1}^n \int \left| f_i \ast K_{h_n}-f_i \right| \nonumber\\
 &\le& \frac{1}{n}\sum_{i=1}^n h_n^s D_s^{\ast} \left( f_i \right) \int \left|L
 \right| \nonumber \\
&\le& \frac{1}{n}\sum_{i=1}^n h_n^s \mathcal{B}_{\mathcal{C}} \int \left|L \right| \nonumber \\
 &=& h_n^s \mathcal{B}_{\mathcal{C}} \int \left|L \right|
\end{eqnarray}
Taking limits on both sides we get,
\begin{eqnarray}
\lim_{n \rightarrow \infty} \int \left|
\left(\frac{1}{n}\sum_{i=1}^n f_i\right)\ast
K_{h_n}-\left(\frac{1}{n}\sum_{i=1}^n f_i\right) \right| = 0
\end{eqnarray}
\end{proof}

\begin{proof}[Proof of Theorem \ref{thm:mix_density_approx_2}]
\begin{eqnarray}
f_i(x) = f_i(x) \int K_h(t)dt = \int f_i(x) K_h(t)dt, \quad
\forall i
\end{eqnarray}
Therefore,
\begin{eqnarray}
\left| \left( \frac{1}{n} \sum_{i=1}^n f_i \ast K_h \right) (x) -
\frac{1}{n} \sum_{i=1}^n f_i(x) \right|  = \left| \int \left[
\frac{1}{n} \sum_{i=1}^n  f_i(x-t) - \frac{1}{n} \sum_{i=1}^n f_i(x) \right]K_h(t) dt\right| \nonumber \\
 \le \int \left| \frac{1}{n} \sum_{i=1}^n f_i(x-t) - \frac{1}{n} \sum_{i=1}^n  f_i (x)\right| \left| K_h(t) \right|^{\frac{1}{p}} \left| K_h(t) \right|^{\frac{1}{p'}} dt
\end{eqnarray}
where $\frac{1}{p}+ \frac{1}{p'} = 1, \left( \frac{1}{p'}=0 \text{
if } p =1 \right)$. Applying Holder's inequality with exponents
$p$ and $p'$, and then raising both sides to the $p^{\text{th}}$
power and integrating with respect to x, we obtain
\begin{eqnarray}
\int \left| \left(\frac{1}{n} \sum_{i=1}^n f_i \ast K_h \right)
(x) - \frac{1}{n} \sum_{i=1}^n f_i(x) \right|^p dx \nonumber \\
\le \int \left[ \int \left| \frac{1}{n} \sum_{i=1}^n f_i(x-t)
-\frac{1}{n} \sum_{i=1}^n  f_i(x)\right|^p \left| K_h(t) \right|
dt \right] \left[ \int \left| K_h(t) \right| dt
\right]^{\frac{p}{p'}}dx \nonumber \\
 = \| K \|_1^{\frac{p}{p'}} \int \left[ \int \left| \frac{1}{n} \sum_{i=1}^n f_i(x-t) - \frac{1}{n} \sum_{i=1}^n f_i(x) \right|^p\left| K_h (t) \right|dt
 \right]dx \nonumber\\
 \le \| K \|_1^{\frac{p}{p'}} \int \left[ \frac{1}{n} \sum_{i=1}^n \int \left| f_i(x-t) - f_i(x) \right|^p\left| K_h (t) \right|dt
 \right]dx
\end{eqnarray}
Changing the order of integration in the last expression (which is
justified since the integrand is nonnegative), we obtain
\begin{eqnarray}
\| \left(\frac{1}{n} \sum_{i=1}^n f_i \right) \ast K_h -
\frac{1}{n} \sum_{i=1}^n f_i\|_p^p &\le& \| K \|_1^{\frac{p}{p'}}
\int \left| K_h(t) \right| \frac{1}{n} \sum_{i=1}^n \omega_i(t) dt \nonumber \\
 &\le& \| K \|_1^{\frac{p}{p'}} \int \left| K_h(t) \right| \Omega(t)dt
\end{eqnarray}
For $\delta >0$,
\begin{eqnarray}
I_{h} = \int \left| K_h(t) \right| \Omega(t)dt = \int_{|t|<
\delta} + \int_{|t|\ge \delta} = A_{h,\delta} + B_{h,\delta}
\end{eqnarray}
Since, we have $\Omega(t) \rightarrow 0$ as $|t| \rightarrow 0$,
for $\eta>0$, we can choose $\delta$ so small that $\Omega(t) <
\eta$ if $|t|< \delta$. Then
\begin{eqnarray}
A_{h,\delta} \le \eta \int_{|t|<\delta}\left| K_h(t) \right| dt
\le \eta \| K \|_1, \quad \forall h>0
\end{eqnarray}
Also, $\Omega$ is a bounded function by Minkowski's inequality
[note that $\|\Omega \|_{\infty} \le \sup_{i \in \mathbb{N}} \|
\omega_i \|_{\infty} \le \sup_{i \in \mathbb{N}} \left( 2\| f_i
\|_p \right)^p$, which for $p=1$, becomes $\|\Omega \|_{\infty}
\le 2$], so that $B_{h,\delta}$ is less than a constant multiple
of $\int_{|t|\ge \delta} \left|K_h(t)\right|dt$, which tends to
zero with $h$. This proves that $I_h \rightarrow 0 $ as $h
\rightarrow 0 $ and the theorem follows.
\end{proof}

Another lemma necessary for the proof of Lemma
\ref{lem:lem_4_pf_dens_const} is the following.

\begin{lemma}{(A Multinomial distribution inequality)}

    Let $N_1, \cdots, N_k$ be a multinomial random vector with
parameters $n,p_1,\cdots, p_k$. Then
\begin{equation}
P\left(\sum_{i=1}^k \left|\frac{N_i}{n}-p_i\right| \ge \epsilon
\right) \le 2^{k+1}e^{\frac{-n\epsilon^2}{2}}
\end{equation}
\end{lemma}
\textbf{Proof}

    By Scheffe's theorem,
\begin{equation}
\sum_{i=1}^k\left|\frac{N_i}{n}-p_i\right| =
2\sup_{A}\left|\frac{N(\mathbb{A})}{n}-P(\mathbb{A})\right|
\end{equation}

where, $\mathbb{A}$ = \{all $2^k$ possible sets of integers from
$1, \cdots, k$\} and $N(\mathbb{A})$ is the cardinality of
$\mathbb{A}$. By Bonferroni's inequality and Hoeffding's
inequality,

\begin{equation}
P\left(\sup_{\mathbb{A}}\left|\frac{N(\mathbb{A})}{n}-P(\mathbb{A})\right|
\ge \frac{\epsilon}{2}\right) \le  2^k
2e^{-2n\left(\frac{\epsilon}{2}\right)^2}
\end{equation}

%

The expected value of $f^n(x)$ is denoted by,
\begin{equation}
g_h(x) =  E(f^n(x)) =  \frac{1}{nh^d} \sum_{i=1}^n \int K\left(
\frac{x-y}{h} \right)f_i(y) dy \label{g_h}
\end{equation}

\begin{proof}[Proof of Lemma \ref{lem:lem_4_pf_dens_const}]

Let $g_h$ be defined as in (\ref{g_h}). By Theorem 1, it
is enough to show that $\int |f^n(x)-g_h(x)|dx \rightarrow 0$
exponentially. Let $\mu_n$ be the empirical probability measure
for $X_1, X_2, \cdots, X_n$ and note that
\begin{eqnarray}
f^n(x)=\frac{1}{h^d}\int K\left(\frac{x-y}{h}\right)\mu_n(dy)\\
\end{eqnarray}

For given $\epsilon > 0$, find finite constants $M,L,N,a_1,
\cdots, a_N$ and disjoint finite rectangles $A_1, \cdots, A_N$ in
$\mathbb{R}^d$ such that the function
\begin{equation}
K^{\ast}(x) = \sum_{i=1}^N a_i I_{A_i}(x) \\
\end{equation}

satisfies: $|K^{\ast}| \le M, K^{\ast} = 0$ outside $[-L,L]^d,$
and $\int |K(x)-K^{\ast}(x)|dx < \epsilon$. Define $g_h^{\ast}$
and $f^{n\ast}$ as $g_h$ and $f^n$ with $K^{\ast}$ instead of $K$.
Then
\begin{eqnarray}
\begin{split}
\int |f^n(x)-g_h(x)|dx &\le  \int |f^n(x)- f^{n\ast}(x)|dx +\int
|f^{n\ast}(x)-g_h^{\ast}(x)|dx +\int |g_h^{\ast}(x)-g_h(x)|dx \nonumber\\
  &\le  \int \frac{1}{h^d}\int \left|K^{\ast}\left(\frac{x-y}{h}\right) - K\left(\frac{x-y}{h}\right)\right|\mu_n(dy)dx   \nonumber \\
  &+ \int \frac{1}{nh^d}\sum_{i=1}^n\int \left|K^{\ast}\left(\frac{x-y}{h}\right) - K\left(\frac{x-y}{h}\right)\right|f_i(y)dydx  \nonumber \\ 
  &+ \int \left|f^{n\ast}(x)-g_h^{\ast}(x) \right|dx \nonumber\\
  &\le  2\epsilon + \int \left|f^{n\ast}(x)-g_h^{\ast}(x)
 \right|dx
\end{split}
 \label{first_ineq}
\end{eqnarray}\\
by a double change of integral. But, if $\mu$ is the probability
measure for $f$,
\begin{eqnarray}
 \int \left|f^{n\ast}(x)-g_h^{\ast}(x)
 \right|dx  \le  \sum_{i=1}^N
|a_i|\int\left|\frac{1}{nh^d}\sum_{j=1}^n\int_{x-hA_i}f_j(y)dy-\frac{1}{h^d}\int_{x-hA_i}\mu_n(dy)\right|dx
 \nonumber \\
 \le  \frac{1}{h^d}\sum_{i=1}^N |a_i|\int \left| \frac{1}{n}\sum_{j=1}^n\mu_j(x-hA_i) -
 \mu_n(x-hA_i)\right|dx \\
 \nonumber
\end{eqnarray}

Lemma \ref{lem:lem_4_pf_dens_const} follows if we can show that
for all finite rectangles $\mathbb{A}$ of $\mathbb{R}^d$
\begin{equation*}
\frac{1}{h^d}\sum_{i=1}^N \int \left|
\frac{1}{n}\sum_{j=1}^n\mu_j(x-hA_i) -
 \mu_n(x-hA_i)\right|dx  \rightarrow 0 \mbox{ exponentially as } n
 \rightarrow \infty
\end{equation*}

Choose an $\mathbb{A}$, and let $\epsilon >0$ be arbitrary.
Consider the partition of $\mathbb{R}^d$ into sets $B$ that are
d-fold products of intervals of the form $\left[\frac{(i-1)h}{N},
\frac{ih}{N}\right)$, where $i$ is and integer, and $N$ is a new
constant to be chosen later. Call the partition $\Pi$. Let
\begin{equation*}
\mathbb{A} = \prod_{i=1}^d \left[x_i,x_i+a_i\right), \min_{i} a_i
\ge \frac{2}{N} \\
\end{equation*}
and
\begin{equation*}
\mathbb{A}^{\ast}=\prod_{i=1}^d\left[x_i+\frac{1}{N},
x_i+a_i-\frac{1}{N}\right)\\
\end{equation*}

Define
\begin{equation*}
C_x = \left(x-h\mathbb{A}-\bigcup_{\stackrel{B \in \Pi}{B
\subseteq x-h\mathbb{A}}} B\right) \subseteq
x+h(\mathbb{A}-\mathbb{A}^{\ast})= C_x^{\ast}
\\
\end{equation*}
Clearly, for any $n$

\begin{eqnarray}
\int \left|
\frac{1}{n}\sum_{j=1}^n\mu_j(x-h\mathbb{A})-\mu_n(x-h\mathbb{A})\right|dx
&\le& \int \sum_{\stackrel{B \in \Pi}{B \subseteq
x-h\mathbb{A}}}|\frac{1}{n}\sum_{j=1}^n\mu_j(B)-\mu_n(B)|dx
\nonumber \\ & & \qquad \qquad +\int
\left(\frac{1}{n}\sum_{j=1}^n\mu_j+\mu_n\right)(C_x^{\ast}) \nonumber \\
\label{diff_meas}
\end{eqnarray}

The last term in (\ref{diff_meas}) equals
\begin{eqnarray}
2\lambda(h(\mathbb{A}-\mathbb{A}^{\ast}))&=&2h^d\lambda(\mathbb{A}-\mathbb{A}^{\ast})
\\ &=& 2h^d\left(\prod_{i=1}^d a_i - \prod_{i=1}^d\left(a_i
-\frac{2}{N} \right)\right) \label{L_meas}
\end{eqnarray}
where $\lambda$ is the Lebesgue measure. Now, putting
(\ref{L_meas}), (\ref{diff_meas}) and (\ref{first_ineq}) together,
we get
\begin{eqnarray}
\int |f^n(x)-g_h(x)|dx  \le 2\epsilon + \int
\left|f^{n\ast}(x)-g_h^{\ast}(x)\right| \nonumber \\
 \le 2\epsilon + \sum_{i=1}^N|a_i|\frac{1}{h^d}\int
\sum_{\stackrel{B \in \Pi}{B \subseteq
x-hA_i}}|\frac{1}{n}\sum_{j=1}^n\mu_j(B)-\mu_n(B)|dx 
+\sum_{i=1}^N|a_i|\frac{2}{h^d}h^d\lambda({A_i}-{A_i}^{\ast})  \nonumber \\
 \le 2\epsilon + \frac{1}{h^d}\sum_{i=1}^N|a_i|\sum_{B \in
 \Pi}\left|\frac{1}{n}\sum_{j=1}^n\mu_j(B)-\mu_n(B)\right|\int_{B \subseteq x-hA_i}dx 
 +\sum_{i=1}^N|a_i|\frac{2}{h^d}h^d\lambda({A_i}-{A_i}^{\ast}) \nonumber \\
 \le 2\epsilon + \frac{1}{h^d}\sum_{i=1}^N|a_i|\sum_{B \in
 \Pi}|\frac{1}{n}\sum_{j=1}^n\mu_j(B)-\mu_n(B)|h^d\lambda(A_i) 
 +\sum_{i=1}^N|a_i|\frac{2}{h^d}h^d\lambda(A_i-A_i^{\ast}) \nonumber \\
 \le 2\epsilon + \left(\sum_{i=1}^N |a_i|\lambda(A_i)\right)\sum_{B
\in \Pi}|\frac{1}{n}\sum_{j=1}^n\mu_j(B)-\mu_n(B)|
+2\sum_{i=1}^N|a_i|\lambda(A_i-A_i^{\ast})) \nonumber \\
\end{eqnarray}
The third term on the right hand side can be made smaller than
$\epsilon$ by choosing $N$ large enough ($A_i^{\ast}\rightarrow A_i$, $\forall i$
as $N \rightarrow \infty)$. The coefficient of the first term on
the right hand side is equal to $\int \left|K^{\ast}\right|\le 1+
\epsilon$. Thus, we have shown that for every $\epsilon > 0$, we
can find $N$ large enough such that
\begin{eqnarray}
\int |f^n(x)-g_h(x)|dx &\le& 3\epsilon + (1+\epsilon)\sum_{B \in
\Pi}|\frac{1}{n}\sum_{j=1}^n\mu_j(B)-\mu_n(B)|\nonumber\\
 &\le& 5\epsilon + \sum_{B \in \Pi}|\frac{1}{n}\sum_{j=1}^n\mu_j(B)-\mu_n(B)|
\end{eqnarray}
We are almost in a position to use the multinomial inequality were
it not for the fact that the partition $\Pi$ is infinite. Thus, it
is necessary to "cut-off" the tails of the distribution. Consider
a finite partition, $\Pi_r$, consisting of sets of $\Pi$ that has
a non-empty intersection with $[-r,r]^d$ where $r>0$ is to be
picked later. Let $\Pi_r^{\ast}$ be $\Pi_r \bigcup [-r,r]^{d^c}$.
The cardinality of $\Pi_r$ is at most
\begin{equation*}
\left(\frac{2rN}{h}+2\right)^d = O(n)
\end{equation*}
To take care of the tails we argue as follows: let $T $ stand for
the tail set, i.e., the complement of $[-r,r]^{d}$. then
\begin{eqnarray}
\sum_{B \in
\Pi}\left|\frac{1}{n}\sum_{j=1}^n\mu_j(B)-\mu_n(B)\right|  \le
\sum_{B \in
\Pi_r}\left|\frac{1}{n}\sum_{j=1}^n\mu_j(B)-\mu_n(B)\right| +
\frac{1}{n}\sum_{j=1}^n\mu_j(T) + \mu_n(T) \nonumber \\
 \le  \sum_{B \in
\Pi_r}\left|\frac{1}{n}\sum_{j=1}^n\mu_j(B)-\mu_n(B)\right| +
2\frac{1}{n}\sum_{j=1}^n\mu_j(T) +
\left|\frac{1}{n}\sum_{j=1}^n\mu_j(T)-\mu_n(T)\right| \nonumber \\
\le  \sum_{B \in
\Pi_{r^{\ast}}}\left|\frac{1}{n}\sum_{j=1}^n\mu_j(B)-\mu_n(B)\right|
+ 2\frac{1}{n}\sum_{j=1}^n\mu_j(T) \nonumber \\
\le \sum_{B \in
\Pi_{r^{\ast}}}\left|\frac{1}{n}\sum_{j=1}^n\mu_j(B)-\mu_n(B)\right|+2\sup_{i
\in \mathcal{I}} \mu_i (T) \label{meas_ineq}
\end{eqnarray}
Now, $2\sup_{i \in \mathcal{I}}\mu_i(T)$ can be made smaller than
$\epsilon$ by choice of $r$. This gives,
\begin{eqnarray} \int |f^n(x)-g_h(x)|dx &\le& 6\epsilon + \sum_{B
\pi_{r^{\ast}}}\left|\frac{1}{n}\sum_{j=1}^n\mu_j(B)-\mu_n(B)\right|
\end{eqnarray}

where $r$ depends on $\epsilon, \Upsilon$, and $N$ depends on
$\epsilon, K$.


By Lemma 1, for $\delta > 6\epsilon$ and $\rho \in (0,1)$,
\begin{eqnarray}
P\left(\int \left| f^n -g_h\right| > \delta \right) &\le& 
P\left(\sum_{B
\pi_{r^{\ast}}}\left|\frac{1}{n}\sum_{j=1}^n\mu_j(B)-\mu_n(B)\right|
> \delta -6\epsilon \right) \nonumber \\
 &\le&
 2^{2+\left(2+\frac{2rN}{h}\right)^d}e^{-\frac{1}{2}n(\delta-6\epsilon)^2}\\
  &\le& e^{-(1-\rho)\frac{n\delta^2}{2}}, n \ge n_0(\rho,\delta, K, \Upsilon, {h})
\end{eqnarray}

This concludes that the proof 5 $\Rightarrow$ 4 for nonnegative
$K$. Note that the inequality can be forced for all $n,h$ with
\begin{eqnarray}
n > \frac{16+4^{d+1}}{\rho \delta^2}\\
nh^d > n_0^d \left(\mathcal{C}, \rho, \delta,  K,d \right) = \frac{4 2^d (2r(\mathcal{C},K)N)^d}{\rho \delta^2}
\label{eqn:n_0d}
\end{eqnarray}
if we pick
\begin{equation*}
\epsilon = \frac{\delta}{6}\left(1-\sqrt{1-\frac{\rho}{2}}\right)
\end{equation*}
For the symbol-by-symbol case, $d=1$ and (\ref{eqn:n_0d}) becomes
\begin{eqnarray}
n > \frac{16+4^{d+1}}{\rho \delta^2} \\
nh^d > n_0 \left(\mathcal{C}, \rho, \delta, K \right) = \frac{16r(\mathcal{C},K)N}{\rho \delta^2}
\label{eqn:n_0}
\end{eqnarray}
\end{proof}

\section{Proof of Theorem \ref{cont_map}}
\label{app:proofofthmcont_map}
\begin{definition}[Prohorov metric]
For any two laws $P$ and $Q$ on the set $[a,b] \subset
\mathbb{R}$, the Prohorov metric, $\rho$ is defined as
\begin{equation*}
\rho \left(P,Q\right) := \inf \{ \varepsilon>0: P^{\Delta}(B) \le
P(B^{\varepsilon})+\varepsilon, B \in \mathcal{B}^{[a,b]} \}
\end{equation*}
where $B^{\varepsilon} = \{\tilde{x}: \left| x-\tilde{x} \right| <
\varepsilon, x \in B \}$.\\
\end{definition}

\begin{proof}[Proof of Theorem \ref{cont_map}]
Let $P_n$ and $Q_n$ denote the laws associated with the
distribution functions, $F_{x^n}$ and $\hat{F}_{x^n}$. From
\cite[Theorem 11.7.1]{Dudley}, $\rho \left( P_n, Q_n \right)
\rightarrow 0 \Rightarrow \beta \left(P_n, Q_n \right)$ then by
definition of the $\beta$-metric, we have
\begin{equation}
\lim_{n \rightarrow \infty} \left| \int f d \left(P_n - Q_n
\right)\right| = 0 \qquad  \forall \|f\|_{BL} \le 1
\end{equation}
By a mere scaling, the above statement is also true for a
uniformly bounded Lipschitz class of functions,
$\mathcal{S}^{[a,b]}_M= \{f: \|f\|_{BL} < M, f :[a,b] \rightarrow
\mathbb{R} \}$ for some $M < \infty$. It is also true that
\begin{equation}
\lim_{n \rightarrow \infty} \left| \int f(x,y) d \left(P_n - Q_n
\right)\right| = 0 \qquad  \forall y \text{ and } f \in
\mathcal{S}^{[a,b]\times \mathbb{R}}
\end{equation}
where $\mathcal{S}^{[a,b]\times \mathbb{R}}_M:= \{ f :[a,b] \times
\mathbb{R} \rightarrow \mathbb{R},
\parallel f(y)
\parallel_{BL} < M \; \forall y\}$ for some $M < \infty$ and
\begin{eqnarray}
\parallel f(y) \parallel_L := \sup_{x \neq z} \frac{\left|
f(x,y)-f(z,y)\right|}{|x-z|}\\
\parallel f(y) \parallel_{\infty} := \sup_{x} f(y,x)\\
\parallel f(y)\parallel_{BL}:=\parallel f(y)
\parallel_L+\parallel f(y) \parallel_{\infty}
\end{eqnarray}
Hence, for a channel with conditional densities, $\{f_{Y|x}\}_{x
\in [a,b]} \in \mathcal{S}_M^{[a,b]\times \mathbb{R}}$, we have
\begin{equation}
\left| \int f_{Y|x}dF_{x^n} - \int f_{Y|x}d\hat{F}_{x^n} \right|
\rightarrow 0 \quad \forall y \in \mathbb{R}
\end{equation}
and by dominated convergence theorem,
\begin{eqnarray}
\int \left| \int f_{Y|x}dF_{x^n} - \int f_{Y|x}d\hat{F}_{x^n}
\right|dy \rightarrow 0
\end{eqnarray}
and hence, $d \left( \left[F_{x^n} \otimes \mathcal{C}\right]_Y ,
\left[\hat{F}_{x^n} \otimes \mathcal{C}\right]_Y \right)
\rightarrow 0$.

Hence, the mapping of input empirical distributions to output
densities induced by the channel,
\begin{eqnarray}
f_{Y^n}(y) = \left[F_{x^n} \otimes \mathcal{C} \right]_Y = \int
f_{Y|x}dF_{x^n}(x)
\end{eqnarray}
is  continuous with respect to the $\beta$ metric on the input
distributions and the total variation metric on the output
densities. We also have the fact that $\left( \mathcal{F}^{[a,b]},
\beta \right)$ is a compact \cite[Theorem 11.5.4 , Corollary
11.5.5 ]{Dudley} metric space. Since, we have a continuous 1-1
(bijection) mapping between the compact metric space of input
distributions with the $\beta$ metric, $\left(
\mathcal{F}^{[a,b]}, \beta \right)$, and the space of output
densities, with the total variation metric, $\left( \left[
\mathcal{F}^{[a,b]}\otimes \mathcal{C}\right], d \right)$,we can
apply the continuous mapping theorem \cite{Rudin} to get
continuity in the inverse mapping too. This gives the desired
result that as $d(\left[F_{x^n} \otimes \mathcal{C}
\right]_Y,\left[\hat{F}_{x^n} \otimes \mathcal{C} \right]_Y)
\rightarrow 0$, we have $\beta \left( P_n, Q_n\right) \rightarrow
0$ and $\rho \left( P_n, Q_n\right) \rightarrow 0$. Finally using
the fact \cite{Dudley}, $\lambda \le \rho$, $\lambda \left(
F_{x^n}, \hat{F}_{x^n} \right)\rightarrow 0$.
\end{proof}

\section{Proof of Lemma \ref{lem:dist_approx}}
\begin{proof}
Consider $f\in \mathcal{C}_b([a,b])$, where $\mathcal{C}_b$
denotes the set of all continuous bounded functions, $f:[a,b]
\rightarrow \mathbb{R}$. For any $F \in \mathcal{F}^{[a,b]}$ and $P^{\Delta}$ that is constructed using (\ref{P_delta})
\begin{eqnarray}
\left| \int f dF(x) \right. &-& \left. \int f P^{\Delta}(dx) \right| \nonumber \\ 
&=& \left|\int f \left(dF(x) - P^{\Delta}(dx) \right) \right| \nonumber \\
    &=& \left| \int f dF(x)- \sum_{i=1}^N f(a_i)P \left(a_i \right) \right| \nonumber \\
    &\le& \left| \sum_{i=0}^{N-1} \int_{a_i}^{a_{i+1}} \left(f(a_i)+ \omega_f(\Delta) \right)dF(x)    - \sum_{i=1}^N f(a_i)P \left(a_i \right) \right| \nonumber \\
    &=& \left| \sum_{i=0}^{N-1} \left(f(a_i)+ \omega_f(\Delta) \right)
    P \left( a_i \right) - \sum_{i=1}^N f(a_i)P \left(a_i \right) \right| \nonumber \\
    &=& \left| \omega_f(\Delta) \sum_{i=1}^{N} P \left( a_i \right) \right| \nonumber \\
    &=& \omega_f(\Delta)
\end{eqnarray}
where $\omega_f \left(\Delta \right) = \max_{y \in [a,b]} \left|
f(y+\Delta)- f(y) \right|$ and $N$ is the number of quantization
levels as defined previously. Hence, 
\begin{eqnarray}
\lim_{\Delta \rightarrow 0 }
\left| P^{\Delta}f -Pf \right| &=& \left| \lim_{\Delta \rightarrow 0
} \int f \left(dF(x) - P^{\Delta}(dx) \right) \right| \\ &=&
\lim_{\Delta \rightarrow 0} \omega_f(\Delta) \\ &=& 0, \qquad \forall f
\in \mathcal{C}_b([a,b]) 
\end{eqnarray}
This implies weak convergence of
$P^{\Delta} \Rightarrow P$. Hence, the statement of the theorem
follows from the Prohorov metric that metrizes weak convergence.
\end{proof}

\section{Proof of Theorem \ref{thm:loss_denoisability_dev}} Using the definition
of the Lipschitz norm of the loss function, $ \Lambda$, and the
channel continuity function, $\xi_{\Delta}$, we bound the
deviation of the expected value of the loss function under two
marginal densities induced at the output of the memoryless channel
by the corresponding empirical distributions of the underlying
clean signal at the input of the memoryless channel.
\begin{lemma}
For any $F, \hat{F} \in \mathcal{F}^{[a,b]}$, measurable $g:
\mathbb{R} \rightarrow [a,b]$ and a bounded Lipschitz loss
function with $E_{f_{Y|u}} \Lambda(u,g(Y)) <\infty $, $\forall u$,
\begin{multline}
\left|E_{F \otimes C} \Lambda(U_0,g(Y))-E_{\hat{F} \otimes C}
\Lambda(U_0,g(Y)) \right| \\ \le  \left( \parallel \Lambda
\parallel_L + \Lambda_{\max} \parallel \Xi
 \parallel_L+ (b-a)\parallel \Lambda \parallel_L \parallel \Xi
 \parallel_L+ \Lambda_{\max} \right) \beta \left( P, \hat{P}\right)
\end{multline}
where $P$ and $\hat{P}$ are the laws associated with $F$ and
$\hat{F}$, $\beta \left( P, \hat{P}\right)$ is the $\beta$
metric between the corresponding laws.\\
\label{lem:loss_continuity}
\end{lemma}

Similarly, we bound the deviation of the expected loss function
under the marginal density induced by any empirical distribution
at the input of the memoryless channel from that of the expected
loss under the marginal density induced by the corresponding
probability mass function (under the mapping discussed in section
\ref{sec:dist_approx}), in the following Lemma
\begin{lemma}
For any $\Delta > 0$, $F \in \mathcal{F}^{[a,b]}$ with the
associated law $P$, $P^{\Delta} \in \mathcal{F}^{\Delta}$,
measurable $g: \mathbb{R} \rightarrow [a,b]$ and a continuous
bounded loss function with $E_{f_{Y|u}} \Lambda(u,g(Y)) <\infty $,
$\forall$ $u$ ,

\begin{equation*} \left|
E_{P^{\Delta} \otimes C} \Lambda(U_0,g(Y))-E_{F \otimes C}
\Lambda(U_0,g(Y)) \right| \le \xi_{\Delta} \Lambda_{\max} +
\lambda(\Delta)\left(1 + \xi_{\Delta} \right)
\end{equation*}
where $\lambda(\Delta)$ is the global modulus of continuity of the
loss function $\Lambda$ as defined in equation
(\ref{loss_modcont}) and $\xi_{\Delta}$ is as defined in
(\ref{delta_Delta}). \\
\label{lem:los_cont_w_quant}
\end{lemma}
The proofs for Lemmas \ref{lem:loss_continuity} and
\ref{lem:los_cont_w_quant} are discussed in the following section, Appendix \ref{app:Proofs of Lemmas 5 and 6}

\begin{lemma}
For every $n \ge 1$, $x^n \in [a,b]^n$, measurable $g: \mathbb{R}
\rightarrow [a,b]$, and $\varepsilon
>0$,

\begin{equation}
Pr \left( \left| \frac{1}{n} \sum_{i=1}^n \Lambda(x_i, g(Y_i)) -
E_{F_{x^n} \otimes \mathcal{C}} \Lambda(U, g(Y)) \right| >
\epsilon \right) \le 2\exp(-G(\epsilon, \Lambda_{\max})n)
\label{hoeff_cum_loss}
\end{equation}
\label{loss_dev_true_estdist}
\end{lemma}

\begin{proof}
By linearity of expectation, $\frac{1}{n} \sum_{i=1}^n E
\Lambda(x_i,g(Y_i)) = E_{F_{x^n} \otimes \mathcal{C}}
\Lambda(U,g(Y))$. Thus, the expression inside the absolute value
brackets in (\ref{hoeff_cum_loss}) is a sum of zero mean
random variables, bounded in magnitude by $\Lambda_{\max}$.
Furthermore, $\Lambda(x_i,g(Y_i))$ and $\Lambda(x_j,g(Y_j))$ are
independent whenever $i \neq j$. This allows the use of Hoeffding
inequality \cite{devroye:gyorfi:lugosi:1996} as in \cite{Dembo} leading to (\ref{hoeff_cum_loss}).
\end{proof}

In preparation of the proof of Theorem \ref{thm:loss_denoisability_dev}, we need also the following two Lemmas
\begin{lemma}
$d(f_Y^n,\left[\hat{F}_{x^n} \otimes \mathcal{C} \right]_Y)
\rightarrow 0 $ a.s. \label{lem:density_consist_cont_map_1}
\end{lemma}

\begin{proof}
By definition,
\begin{eqnarray*}
0 \le d(f_Y^n,\left[\hat{F}_{\hat{x}^n} \otimes \mathcal{C}
\right]_Y) \le d(f_Y^n,\left[F_{x^n} \otimes \mathcal{C}
\right]_Y), \; \forall n \label{lemma_5_1}
\end{eqnarray*}
Taking limit $n \rightarrow \infty$ in the inequality of
(\ref{lemma_5_1}), we get
\begin{eqnarray*}
0 \le \lim_{n \rightarrow \infty} d(f_Y^n,\left[\hat{F}_{x^n}
\otimes \mathcal{C} \right]_Y) \le \lim_{n \rightarrow \infty}
d(f_Y^n,\left[F_{x^n} \otimes \mathcal{C} \right]_Y) &=&0 \quad
a.s. \label{lemma_5_2}
\end{eqnarray*}
where the second part of the inequality in (\ref{lemma_5_2})
follows from Theorem \ref{thm:density_consist}.\\
\end{proof}

\begin{lemma}
$d(\left[F_{x^n} \otimes \mathcal{C} \right]_Y,\left[\hat{F}_{x^n}
\otimes
\mathcal{C} \right]_Y) \rightarrow 0 $ \textit{a.s.}\\
\label{lem:density_consist_cont_map_2}
\end{lemma}

\begin{proof}
\begin{eqnarray*}
0 \le d(\left[F_{x^n} \otimes \mathcal{C}
\right]_Y,\left[\hat{F}_{x^n} \otimes \mathcal{C} \right]_Y)  \le
d(\left[F_{x^n} \otimes \mathcal{C} \right]_Y,f_Y^n) +
d(f_Y^n,\left[\hat{F}_{x^n} \otimes \mathcal{C} \right]_Y)
\end{eqnarray*}

We have already seen $d(\left[F_{x^n} \otimes \mathcal{C}
\right]_Y,f_Y^n) \rightarrow \; a.s$ and by Lemma
\ref{lem:density_consist_cont_map_1}, 
\begin{equation*}
d(f_Y^n,\left[\hat{F}_{x^n}
\otimes \mathcal{C} \right]_Y) \rightarrow 0 \quad a.s.
\end{equation*}
Whence,
\begin{equation*}
d(\left[F_{x^n} \otimes \mathcal{C} \right]_Y,\left[\hat{F}_{x^n}
\otimes \mathcal{C} \right]_Y) \rightarrow 0 \; a.s.
\end{equation*}
\end{proof}

We are now ready for the proof of Theorem \ref{thm:loss_denoisability_dev},
\begin{proof}[Proof of Theorem \ref{thm:loss_denoisability_dev}]
We fix $n \ge 1$, $x^n \in [a,b]^n$,
\begin{multline}
\left| E_{\hat{P}^{\delta, \Delta}_{x^n}[Y^n] \otimes
\mathcal{C}} \Lambda(U, g(Y)) - E_{F_{x^n} \otimes \mathcal{C}}
\Lambda(U, g(Y)) \right| \le \\\left|E_{\hat{P}^{\delta,
\Delta}_{x^n}[Y^n] \otimes \mathcal{C}} \Lambda(U, g(Y)) -
E_{\hat{F}_{x^n}[Y^n] \otimes \mathcal{C}} \Lambda(U, g(Y))
\right| + \\
\left|E_{\hat{F}_{x^n}[Y^n] \otimes \mathcal{C}} \Lambda(U, g(Y))
- E_{F_{x^n} \otimes \mathcal{C}} \Lambda(U, g(Y)) \right|
\label{thm_9_1}
\end{multline}
Hence,
\begin{align}
Pr \left( \sup_{g: \mathbb{R} \rightarrow [a,b]} \left|
E_{\hat{P}^{\delta, \Delta}_{x^n}[Y^n] \otimes \mathcal{C}}
\Lambda(U, g(Y)) - E_{F_{x^n} \otimes \mathcal{C}} \Lambda(U,
g(Y)) \right| >  \epsilon + \delta \Lambda_{\max} +
\xi_{\Delta}\Lambda_{\max} + \right. \nonumber \\ \left. \lambda(\Delta)
(1+\xi_{\Delta}) \right) 
\le Pr \left( \left| E_{\hat{F}_{x^n}[Y^n] \otimes \mathcal{C}}
\Lambda(U, g(Y)) - E_{F_{x^n} \otimes \mathcal{C}} \Lambda(U,
g(Y)) \right| > \epsilon \right) +  \\
Pr \left(\left| E_{\hat{F}{x^n}[Y^n] \otimes \mathcal{C}}
\Lambda(U, g(Y)) - E_{\hat{P}^{\delta, \Delta}_{x^n}[Y^n]
\otimes \mathcal{C}} \Lambda(U, g(Y)) \right| > \delta
\Lambda_{\max} + \xi_{\Delta}\Lambda_{\max} + \lambda(\Delta)
(1+\xi_{\Delta})\right) \label{thm_9_1a}
\end{align}
Now,
\begin{align}
Pr \left( \left| E_{\hat{F}{x^n}[Y^n] \otimes \mathcal{C}}
\Lambda(U, g(Y)) - E_{F_{x^n} \otimes \mathcal{C}} \Lambda(U,
g(Y)) \right| > \epsilon \right) \le \nonumber\\ 
Pr \left( \left( \parallel
\Lambda \parallel_L + \Lambda_{\max}
\parallel \Xi
 \parallel_L+ (b-a)\parallel \Lambda \parallel_L \parallel \Xi
 \parallel_L+ \Lambda_{\max} \right) \beta
\left(P_{x^n}, \hat{P}_{x^n} \right)
> \epsilon \right) \label{thm_9_2}\\
\le Pr\left( \left( \parallel \Lambda \parallel_L + \Lambda_{\max}
\parallel \Xi
 \parallel_L+ (b-a)\parallel \Lambda \parallel_L \parallel \Xi
 \parallel_L+ \Lambda_{\max} \right)d \left( F_{x^n} \otimes \mathcal{C},
\hat{F}_{x^n}
\otimes \mathcal{C} \right)> \epsilon \right) \nonumber \\
\le e^{-(1-\rho)\frac{n\gamma^2}{2}}, \nonumber \\ \mbox{for all } nh_n >
n_0(\mathcal{C}, \rho, \delta, K) \label{thm_9_3}
\end{align}
where $\mathcal{C}$ is the family of channel densities
$\{f_{Y|x}\}$. The inequality in (\ref{thm_9_2}) is due to Lemma
\ref{lem:loss_continuity}, while the first inequality in
(\ref{thm_9_3}) is by application of Theorem \ref{cont_map} and
the second inequality is due to Lemma
\ref{lem:density_consist_cont_map_2} and Theorem
\ref{thm:density_consist}. Finally, application of Lemma
\ref{lem:los_cont_w_quant} to (\ref{thm_9_1a}) yields
\begin{align}
Pr \left( \sup_{g: \mathbb{R} \rightarrow [a,b]} \left|
E_{\hat{P}^{\delta, \Delta}_{x^n}[Y^n] \otimes \mathcal{C}}
\Lambda(U, g(Y)) - E_{F_{x^n} \otimes \mathcal{C}} \Lambda(U,
g(Y)) \right| > \epsilon + \delta \Lambda_{\max} +
\xi_{\Delta}\Lambda_{\max} + \right. \nonumber \\ \left. \lambda(\Delta)
(1+\xi_{\Delta}) \right) 
\le e^{-(1-\rho)\frac{n\gamma^2}{2}}, \quad \mbox{for all } n >
n_0(\mathcal{C}, \rho, \delta, K) \label{thm_9_4}
\end{align}
Combining (\ref{thm_9_4}) with Lemma \ref{loss_dev_true_estdist}
gives
\begin{multline}
Pr \left( \left| \frac{1}{n} \sum_{i=1}^n \Lambda(x_i,g(Y_i)) -
E_{\hat{P}^{\delta, \Delta}_{x^n} \otimes \mathcal{C}
}\Lambda(U, g(Y) ) \right| > 2\epsilon + 2\delta \Lambda_{\max} +
\xi_{\Delta}\Lambda_{\max} + \lambda(\Delta)
(1+\xi_{\Delta}) \right) \\
\le 2e^{-G(\epsilon+\delta \Lambda_{\max},
\Lambda_{\max})n} + e^{-(1-\rho)\frac{n\gamma^2}{2}}, \quad 
\mbox{for all } nh_n > n_0(\mathcal{C}, \rho, \delta, K)
\label{loss_dev1}
\end{multline}

By the union bound, (\ref{loss_dev1}) guarantees that for any
class $\mathcal{G}$
\begin{multline}
Pr \left( \max_{g \in \mathcal{G}} \left| \frac{1}{n} \sum_{i=1}^n
\Lambda(x_i,g(Y_i)) - E_{\hat{P}^{\delta, \Delta}_{x^n} \otimes
\mathcal{C} \Lambda(U, g(Y) )} \right| > 2\epsilon + 2\delta
\Lambda_{\max}
+C_{\Delta}\Lambda_{\max} \right. \\ 
\left. + \lambda(\Delta)(1+\xi_{\Delta})\right)
\le |\mathcal{G}| \left[2e^{-G(\epsilon+\delta \Lambda_{\max},
\Lambda_{\max})n} + e^{-(1-\rho)\frac{n\gamma^2}{2}} \right]
\label{mx_loss_dev1}
\end{multline}

Consequently,
\begin{multline}
Pr \left( \left| L_{\tilde{X}^{n, \delta, \Delta}}(x^n, Y^n) -
\min_{g \in \mathcal{G}_{\delta, \Delta}} E_{\hat{P}^{\delta,
\Delta}_{x^n} \otimes \mathcal{C}} \Lambda(U, g(Y)) \right| > 
2\epsilon + 2\delta \Lambda_{\max}
+C_{\Delta}\Lambda_{\max} \right. \\
\left. +\lambda(\Delta)(1+\xi_{\Delta})\right) = Pr \left( \left| \frac{1}{n} \sum_{i=1}^n
\Lambda(x_i,g_{opt}[\hat{P}^{\delta, \Delta}_{x^n}[Y^n]](Y_i)) -
E_{\hat{P}^{\delta, \Delta}_{x^n} \otimes \mathcal{C}}
\Lambda(U, g_{opt}[\hat{P}^{\delta, \Delta}_{x^n}[Y^n]](Y) )
\right| \right.  \\ > \left. 2\epsilon + 2\delta \Lambda_{\max}
+C_{\Delta}\Lambda_{\max}+\lambda(\Delta)(1+\xi_{\Delta})\right)\\
\le Pr \left( \max_{g \in \mathcal{G}_{\delta, \Delta}} \left|
\frac{1}{n} \sum_{i=1}^n \Lambda(x_i,g(Y_i)) -
E_{\hat{P}^{\delta, \Delta}_{x^n} \otimes \mathcal{C}}
\Lambda(U, g(Y)) \right|
>  2\epsilon + 2\delta \Lambda_{\max}
+C_{\Delta}\Lambda_{\max} \right. \\ +\left. \lambda(\Delta)(1+\xi_{\Delta})\right)
\le |\mathcal{G}_{\delta, \Delta}| \left[2e^{-G(\epsilon+\delta
\Lambda_{\max}, \Lambda_{\max})n} +
e^{-(1-\rho)\frac{n\gamma^2}{2}} \right] \label{max_loss_dev2}
\end{multline}
where the first equality follows from the definition of
$\tilde{X}^{n, \delta, \Delta}$ and the fact that for any  $P \in
\mathcal{F}_{\delta, \Delta}$,
\begin{equation*}
\min_{g \in \mathcal{G}_{\delta,
\Delta}} E_{P\otimes \mathcal{C}} \Lambda(U, g(Y))  = E_{P \otimes
\mathcal{C}} \Lambda(U, g_{opt}[P](Y))
\end{equation*}
The first inequality follows by the fact that $\hat{P}^{\delta, \Delta}_{x^n}[Y^n]
\in \mathcal{F}_{\delta, \Delta}$ and therefore
$g_{opt}[\hat{P}^{\delta, \Delta}_{x^n}[Y^n]] \in
\mathcal{G_{\delta, \Delta}}$ and finally the last inequality
follows from (\ref{mx_loss_dev1}). It also follows, from
(\ref{thm_9_4}), that

\begin{multline}
Pr \left( \left| \min_{g \in \mathcal{G}_{\delta, \Delta}}
E_{\hat{P}^{\delta, \Delta}_{x^n} \otimes \mathcal{C}}
\Lambda(U, g(Y))  - \min_{g \in \mathcal{G}_{\delta, \Delta}}
E_{F_{x^n} \otimes \mathcal{C}} \Lambda(U,G(Y)) \right| > \right. \\ \left. \epsilon
+ \delta \Lambda_{\max} + \xi_{\Delta}\Lambda_{\max} +
\lambda(\Delta)
(1+\xi_{\Delta}) \right) \le  e^{-(1-\rho)\frac{n\gamma^2}{2}} \label{loss_dev2}
\end{multline}

Combining (\ref{max_loss_dev2}) and (\ref{loss_dev2}) gives
\begin{multline}
Pr \left( \left| L_{\tilde{X}^{n,\delta, \Delta}}(x^n,Y^n) -
\min_{g \in \mathcal{G}_{\delta, \Delta}} E_{F_{x^n} \otimes
\mathcal{C}} \Lambda(U,g(Y)) \right|  >  3\epsilon + 3\delta
\Lambda_{\max} + 2\xi_{\Delta}\Lambda_{\max} + \right. \\
\left. 2\lambda(\Delta)(1+\xi_{\Delta})\right) 
\le |\mathcal{G}_{\delta, \Delta}| \left[2e^{-G(\epsilon+\delta
\Lambda_{\max}, \Lambda_{\max})n} +
e^{-(1-\rho)\frac{n\gamma^2}{2}}
\right]+e^{-(1-\rho)\frac{n\gamma^2}{2}} \label{loss_dev3}
\end{multline}

On the other hand, letting $\hat{P}_{x^n}^{\delta,\Delta}$
denote the element in $\mathcal{F}_{\delta,\Delta}$ closest (under
the Prohorov metric of the corresponding measures) to $F_{x^n}$,
\begin{eqnarray}
& & \left| D_0(x^n)-\min_{g \in \mathcal{G}_{\delta, \Delta}}
E_{F_{x^n} \otimes
\mathcal{C}} \Lambda(U,g(Y)) \right|  \nonumber \\
&=&\left| \min_{F \in \mathcal{F}^{[a,b]}_n} E_{F_{x^n} \otimes
\mathcal{C}} \Lambda(U,g_{opt}[F](Y)) - \min_{g \in
\mathcal{G}_{\delta,\Delta}} E_{F_{x^n} \otimes \mathcal{C}}
\Lambda(U,g(Y)) \right| \label{2}\\
&\le &\left| \min_{F \in \mathcal{F}^{[a,b]}_n}
E_{\tilde{F}_{x^n}^{\delta, \Delta} \otimes \mathcal{C}}
\Lambda(U,g_{opt}[F](Y)) - \min_{g \in
\mathcal{G}_{\delta,\Delta}} E_{F_{x^n} \otimes \mathcal{C}}
\Lambda(U,g(Y)) \right| + \nonumber \\
& &\qquad \qquad \Lambda_{\max}\delta+\xi_{\Delta}\Lambda_{\max}+\lambda(\Delta)(1+\xi_{\Delta}) \label{3}\\
&=&\left| \min_{P \in \mathcal{F}^{\delta,\Delta}}
E_{\hat{P}_{x^n}^{\delta, \Delta} \otimes \mathcal{C}}
\Lambda(U,g_{opt}[P](Y)) - \min_{g \in
\mathcal{G}_{\delta,\Delta}} E_{P_{x^n} \otimes \mathcal{C}}
\Lambda(U,g(Y)) \right| + \nonumber  \\
& & \qquad \qquad \Lambda_{\max}\delta+\xi_{\Delta}\Lambda_{\max}+\lambda(\Delta)(1+\xi_{\Delta})\label{4}\\
&=&\left| \min_{g \in \mathcal{G}_{\delta,\Delta}}
E_{\tilde{F}_{x^n}^{\delta, \Delta} \otimes \mathcal{C}}
\Lambda(U,g(Y)) - \min_{g \in \mathcal{G}_{\delta,\Delta}}
E_{F_{x^n} \otimes \mathcal{C}} \Lambda(U,g(Y)) \right| + \nonumber \\
& & \qquad \qquad \Lambda_{\max}\delta+\xi_{\Delta}\Lambda_{\max}+\lambda(\Delta)(1+\xi_{\Delta})\label{5}\\
&\le& 2 \left(
\Lambda_{\max}\delta+\xi_{\Delta}\Lambda_{\max}+\lambda(\Delta)(1+\xi_{\Delta})
\right)\label{6}
\end{eqnarray}
where (\ref{3}) and (\ref{6}) follow from Lemma \ref{lem:los_cont_w_quant}, and (\ref{4})
follows from the fact that the achiever of the minimum in the
first term of (\ref{3}) is $F_{x^n}^{\delta, \Delta}$ which, by
definition, is a member of $\mathcal{F}_{\delta, \Delta}$.
Finally, combining (\ref{loss_dev2}) with (\ref{6}) gives

\begin{eqnarray}
Pr \left( \left| L_{\tilde{X}^{n,\delta, \Delta}}(x^n,Y^n) -
D_0(x^n) \right| >  3\epsilon + 5\delta \Lambda_{\max} +
4\xi_{\Delta}\Lambda_{\max} +
4\lambda(\Delta)(1+\xi_{\Delta}) \right) \nonumber \\
\le |\mathcal{G}_{\delta, \Delta}| \left[e^{-G(\epsilon+\delta
\Lambda_{\max}, \Lambda_{\max})n} +
e^{-(1-\rho)\frac{n\gamma^2}{2}}
\right]+e^{-(1-\rho)\frac{n\gamma^2}{2}} \\
\text{for all } nh_n > n_0 \left(\mathcal{C}, \rho, \delta, K
\right) \nonumber
\end{eqnarray}
From the definition of $\mathcal{G}_{\delta,
\Delta}$, it is clear that $\left| \mathcal{G}_{\delta, \Delta}
\right| \le \left[ \frac{1}{\delta}+1\right]^{\Delta}$. Hence,
\begin{eqnarray}
Pr \left( \left| L_{\tilde{X}^{n,\delta, \Delta}}(x^n,Y^n) -
D_0(x^n) \right| >  3\epsilon + 5\delta \Lambda_{\max} +
4\xi_{\Delta}\Lambda_{\max} +
4\lambda(\Delta)(1+\xi_{\Delta}) \right) \nonumber \\
\le \left[ 1 + \frac{1}{\delta}\right]^{\Delta} \left[e^{-G(\epsilon+\delta
\Lambda_{\max}, \Lambda_{\max})n} +
e^{-(1-\rho)\frac{n\gamma^2}{2}}
\right]+e^{-(1-\rho)\frac{n\gamma^2}{2}} \\
\text{for all } nh_n > n_0 \left(\mathcal{C}, \rho, \delta, K
\right) \nonumber
\end{eqnarray}
\end{proof}

\section{Proof of Lemmas \ref{lem:loss_continuity} and
\ref{lem:los_cont_w_quant}}
\label{app:Proofs of Lemmas 5 and 6}
We need the following proposition for the proof of Lemma \ref{lem:loss_continuity}
\begin{proposition}
$A(x) = \int \Lambda \left(x,g(y) \right)f_{Y|x}(y)dy$ is a
bounded Lipschitz function for any measurable $g:\mathbb{R} \rightarrow [a,b]$.\\
\label{lem:A_cont_fn}
\end{proposition}

\begin{proof}
Let $\Delta = \left| x-x'\right|$,
\begin{equation*}
\begin{split}
A(x) - A(x') &= \int \Lambda \left(x,g(y) \right) f_{Y|x}(y)dy - \int \Lambda \left(x',g(y) \right)f_{Y|x'}(y)dy \\
 &\le \int \left( \Lambda \left(x',g(y) \right) + \lambda \left(
\Delta, x \right) \right) f_{Y|x}(y)dy - \int \left( \Lambda \left(x',g(y) \right)  \right)f_{Y|x'}(y)dy  \\
 &\le \int \left( \Lambda \left(x',g(y) \right) + \lambda \left(
\Delta, x \right) \right) \left( f_{Y|x'}(y) + \varepsilon_{\Delta}(y) \right)dy - \int \left( \Lambda \left(x',g(y) \right)  \right)f_{Y|x'}(y)dy \\
 &\le \lambda \left(\Delta, x \right) + \Lambda_{\max} \xi_{\Delta} + \lambda \left(\Delta, x \right) \xi_{\Delta}
\end{split}
\end{equation*}
Also,
\begin{equation*}
\begin{split}
A(x) - A(x') &= \int \Lambda \left(x,g(y) \right) f_{Y|x}(y)dy - \int \Lambda \left(x',g(y) \right)f_{Y|x'}(y)dy \\
 &\ge \int \left( \Lambda \left(x',g(y) \right) - \lambda \left(
\Delta, x \right) \right) f_{Y|x}(y)dy - \int \left( \Lambda \left(x',g(y) \right)  \right)f_{Y|x'}(y)dy  \\
 &\ge \int \left( \Lambda \left(x',g(y) \right) - \lambda \left(
\Delta, x \right) \right) \left( f_{Y|x'}(y) - \varepsilon_{\Delta}(y) \right)dy - \int \left( \Lambda \left(x',g(y) \right)  \right)f_{Y|x'}(y)dy \\
 &\ge -\lambda \left(\Delta, x \right) - \Lambda_{\max} \xi_{\Delta} + \lambda \left(\Delta, x \right) \xi_{\Delta}\\
 &\ge -\lambda \left(\Delta, x \right) - \Lambda_{\max} \xi_{\Delta} - \lambda \left(\Delta, x \right) \xi_{\Delta}
\end{split}
\end{equation*}
Hence, $\left| A(x) - A(x') \right| \le \lambda \left( \Delta
\right) + \Lambda_{\max} \xi_{\Delta}+ \lambda \left(\Delta \right)$. \\
The assumption of Lipschitz continuity (condition, C6) of the channel guarantees
$\lim_{\Delta \rightarrow 0} \xi_{\Delta} = 0$. With this and the fact that $\lim_{\Delta \rightarrow 0} \lambda \left(
\Delta \right) = 0$, we have \[\lim_{\left| x- x'\right| < \Delta
\atop \Delta \rightarrow 0} \left|A(x) - A(x') \right| = 0\]
\end{proof}

Moreover,
\begin{eqnarray}
\parallel A \parallel_L &=& \sup_{0< \Delta < \left( b-a \right)} \sup_{x \neq x' \atop
\left| x- x'\right| = \Delta} \frac{\left|A(x) - A(x')
\right|}{\left|x-x' \right|} \nonumber \\ 
&\le& \sup_{0 < \Delta <
\left(b-a\right)}\frac{\lambda \left( \Delta \right) +
\Lambda_{\max} \xi_{\Delta} + \lambda \left( \Delta \right)
\xi_{\Delta}}{\Delta} \nonumber \\
&\le& \parallel \Lambda \parallel_L + \Lambda_{\max} \parallel
 \Xi
 \parallel_L+ (b-a)\parallel \Lambda \parallel_L \parallel \Xi \parallel_L
\end{eqnarray}
Hence,
\begin{eqnarray}
\parallel A \parallel_{BL} &=& \parallel A \parallel_{L} + \parallel A
\parallel_{\infty} \nonumber \\
 &\le& \parallel \Lambda \parallel_L + \Lambda_{\max} \parallel
 \Xi
 \parallel_L+ (b-a)\parallel \Lambda \parallel_L \parallel \Xi
 \parallel_L+ \Lambda_{\max}
\end{eqnarray}

\begin{proof}[Proof of Lemma \ref{lem:loss_continuity}]
\begin{eqnarray}
\left| E_{F \otimes C} \Lambda(U_0,g(Y)) \right. &-& \left. E_{ \hat{F} \otimes C}
\Lambda(U_0,g(Y)) \right| \nonumber\\
 &=& \left| \int dF(x) \left( \int \Lambda\left( x,g(y)\right)f_{Y|x}(y)dy\right) \right. \nonumber \\ & & \qquad - \left. \int d \hat{F}(x) \left( \int \Lambda\left( x,g(y)\right)f_{Y|x}(y)dy\right) \right| \nonumber\\
 &=& \left| \int dF(x) A(x) - \int d \hat{F}(x) A(x)\right| \nonumber \\
 &=& \left| \int A(x)  d \left( F - \hat{F} \right)(x) \right| \nonumber \\
 &\le& \parallel A \parallel_{BL} \beta \left( P, \hat{P}\right) \label{eqn:lem_loss_cont_0} \\
 &\le& \left( \parallel \Lambda \parallel_L + \Lambda_{\max} \parallel
 \Xi \parallel_L+ (b-a)\parallel \Lambda \parallel_L \parallel \Xi
 \parallel_L+ \Lambda_{\max} \right) \beta \left( P, \hat{P}\right) \nonumber
\end{eqnarray}
where, (\ref{eqn:lem_loss_cont_0}) follows from the fact that
$A(x)$ is a bounded Lipschitz function as shown in Proposition
\ref{lem:A_cont_fn}. Hence, \\ as $\beta \left(P, \hat{P}\right)
\rightarrow 0$ we have $\left| E_{F \otimes C}
\Lambda(U_0,g(Y))-E_{ \hat{F} \otimes C} \Lambda(U_0,g(Y)) \right|
\rightarrow 0$. \\
\end{proof}

\begin{proof}[Proof of Lemma \ref{lem:los_cont_w_quant}]
\begin{equation}
\begin{split}
&\left| E_{P^{\Delta} \otimes C} \Lambda(U_0,g(Y))-E_{F \otimes C}
\Lambda(U_0,g(Y)) \right| \\ 
&=\left|\sum_{i=1}^{N(\Delta)} \int_{a_{i-1}}^{a_i} dF(u') \left( \int
\Lambda\left( u',g(y)\right)f_{Y|X=u'}(y)dy\right) -  \sum_{i=1}^{N(\Delta)} P^{\Delta}(a_i) \left( \int \Lambda\left(
 a_i,g(y)\right)f_{Y|X=a_i}(y)dy\right)\right| \\
&=\left| \sum_{i=1}^{N(\Delta)} \int dy \left( \int_{a_{i-1}}^{a_i} f_{Y|X=u'}(y) dF(u') \Lambda\left( u',g(y)\right)\right) -  \sum_{i=1}^{N(\Delta)} P^{\Delta}(a_i) \left( \int \Lambda\left(
 a_i,g(y)\right)f_{Y|X=a_i}(y)dy\right)\right| \label{loss_dev0_2}
\end{split}
\end{equation}

Equality in (\ref{loss_dev0_2}) is due to application of Fubini's
theorem.
Hence,
\begin{equation}
\begin{split}
&\left| E_{P^{\Delta} \otimes C} \Lambda(U_0,g(Y))-E_{F \otimes C}
\Lambda(U_0,g(Y)) \right|\\
& <\left| \sum_{i=1}^{N(\Delta)} \int dy \left( \int_{a_{i-1}}^{a_i} f_{Y|X=u'}(y) dF(u') \left(\Lambda\left( a_i,g(y)\right)+ \lambda(\Delta) \right) \right)  
 - \sum_{i=1}^{N(\Delta)} P^{\Delta}(a_i) \left( \int \Lambda\left( a_i,g(y)\right)f_{Y|X=a_i}(y)dy\right)\right| \\
&=\left| \sum_{i=1}^{N(\Delta)} \int dy \left(\Lambda\left( a_i,g(y)\right)+\lambda(\Delta) \right) \left( \int_{a_{i-1}}^{a_i}
 f_{Y|X=u'}(y)dF(u')\right)  - \sum_{i=1}^{N(\Delta)} P^{\Delta}(a_i) \left( \int
\Lambda\left(a_i,g(y)\right)f_{Y|X=a_i}(y)dy\right)\right|\\
\end{split}
\end{equation}
\begin{multline}
<\left| \sum_{i=1}^{N(\Delta)} \int dy \left(\Lambda\left( a_i,g(y)\right)+\lambda(\Delta) \right)\left( f_{Y|X=a_i}(y) + \varepsilon (y) \right) \left( \int_{a_{i-1}}^{a_{i}}
  dF(u')\right) - \right. \\ 
\left. \sum_{i=1}^{N(\Delta)} P^{\Delta}(a_i) \left(\int \Lambda\left(u_i,g(y)\right)f_{Y|X=a_i}(y)dy\right)\right|\\
<\left| \sum_{i=1}^{N(\Delta)} \left( \int_{a_{i-1}}^{a_i} dF(u')\right) \left[\int \Lambda\left( a_i,g(y)\right) f_{Y|X=a_i}(y)dy  
+ \int \varepsilon (y) \Lambda\left( a_i,g(y)\right)dy + \lambda(\Delta) \int f_{Y|X=a_i}(y)dy \right. \right. \\+ \lambda(\Delta) \int \varepsilon(y)dy
- \left. \left. \sum_{i=1}^{N(\Delta)} P^{\Delta}(a_i) \left(
 \int \Lambda\left(a_i,g(y)\right)f_{Y|X=a_i}(y)dy\right)\right] \right|
\end{multline}

\begin{multline}
 <\left| \sum_{i=1}^{N(\Delta)} \left( \int_{a_{i-1}}^{a_i} dF(u')\right) \left[\int \Lambda\left( a_i,g(y)\right) f_{Y|X=a_i}(y)dy  \right. \right. \\ \left. \left. +
 \right. \right. \int \varepsilon (y) \Lambda\left( a_i,g(y)\right)dy + \lambda(\Delta) \int  f_{Y|X=a_i}(y)dy + \lambda(\Delta) \int \varepsilon(y)dy\\
   - \left. \left. \sum_{i=1}^{N(\Delta)} P^{\Delta}(a_i) \left(
 \int \Lambda\left(a_i,g(y)\right)f_{Y|X=a_i}(y)dy\right)\right]
 \right|
\end{multline}
\begin{eqnarray*}
 &=& \left| \sum_{i=1}^{N(\Delta)} \left( \int_{a_{i-1}}^{a_{i}} dF(u')\right) \left[ \int \varepsilon (y) \Lambda\left( a_i,g(y)\right)dy + \lambda(\Delta) + \lambda(\Delta)
 \xi_{\Delta} \right] \right| \\
&=& \left| \sum_{i=1}^{N(\Delta)} \left( F(a_i) -
F(a_{i-1})\right) \left[ \int \varepsilon (y) \Lambda\left(
a_i,g(y)\right)dy + \lambda(\Delta) + \lambda(\Delta)
 \xi_{\Delta} \right] \right| \\
 &\le& \int \sum_{i=1}^{N(\Delta)} \varepsilon(y) \Lambda\left( a_i,g(y)\right) P^{\Delta}(u_i)dy+
 \left(\lambda(\Delta) + \lambda(\Delta) \xi_{\Delta} \right) \\
 &\le& \xi_{\Delta} \Lambda_{\max} + \left(\lambda(\Delta) +
 \lambda(\Delta) \xi_{\Delta} \right) \\
 &=& \xi_{\Delta} \Lambda_{\max} + \lambda(\Delta) \left(1 +
 \xi_{\Delta} \right)
\end{eqnarray*}

Hence,
\begin {equation}
\lim_{\Delta \rightarrow 0} \left| E_{P^{\Delta} \otimes C}
\Lambda(U_0,g(Y))-E_{F \otimes C} \Lambda(U_0,g(Y)) \right| = 0
\end{equation}
\end{proof}

\section{Proof of Theorem \ref{thm:large_deviation_whole_seq}}
In preparation of Theorem \ref{thm:large_deviation_whole_seq} we start by presenting the proof of Lemma \ref{lem:D_k_seq_subseq} and Theorem \ref{thm:loss_denoisability_dev_2kplus1}
\begin{proof}[Proof of Lemma \ref{lem:D_k_seq_subseq}]
\begin{eqnarray}
D_k \left(x^n \right) &=& \min_{g} E \left[ \frac{1}{n-2k}
\sum_{i=k+1}^{n-k} \Lambda \left( X_0, g\left(Y_{-k}^{k}
\right)\right) \right] \nonumber\\
 &=&\min_{g} \int \frac{1}{n-2k} \sum_{i=k+1}^{n-k}  \Lambda
 \left( x_i, g\left(y_{i-k}^{i+k} \right)\right) \prod_{l=i-k}^{i+k}f_{Y|X
 =x_l}(y_l)dy_l \\
 &=&\min_{g} \frac{1}{n-2k} \sum_{i=k+1}^{n-k}  \int \Lambda
 \left( x_i, g\left(y_{i-k}^{i+k} \right)\right) \prod_{l=i-k}^{i+k}f_{Y|X
 =x_l}(y_l)dy_l \\
 &=&\min_{g} \frac{1}{2k+1} \sum_{i=1}^{2k+1}  \int \frac{1}{\frac{n-2k}{2k+1}} \sum_{j=0}^{\lceil \frac{n-2k-i-1}{2k+1}\rceil-1 }  \Lambda
 \left( x_{j(2k+1)+k+1}, \right. \\
 & & \qquad \qquad \qquad \left. g\left(y_{j(2k+1)+i}^{j(2k+1)+i+2k} \right)\right) \prod_{l=j(2k+1)+i}^{j(2k+1)+i+2k}f_{Y|X
 =x_l}(y_l)dy_l \nonumber \\
   &\ge&\min_{g} \frac{1}{2k+1} \sum_{i=1}^{2k+1}  \int \frac{1}{\lceil \frac{n-2k}{2k+1} \rceil} \sum_{j=0}^{\lceil \frac{n-2k-i-1}{2k+1}\rceil -1}  \Lambda
 \left( x_{j(2k+1)+k+1}, \right. \\ 
 & & \left. \qquad \qquad \qquad g\left(y_{j(2k+1)+i}^{j(2k+1)+i+2k} \right)\right) \prod_{l=j(2k+1)+i}^{j(2k+1)+i+2k}f_{Y|X
 =x_l}(y_l)dy_l \nonumber
 \end{eqnarray}
 
\begin{eqnarray} 
 &\ge&\frac{1}{2k+1} \sum_{i=1}^{2k+1}  \min_{g_i} \int \frac{1}{\lceil \frac{n-2k}{2k+1} \rceil} \sum_{j=0}^{\lceil \frac{n-2k-i-1}{2k+1}\rceil \left( 2k+1 \right)+k+1+i}  \Lambda
 \left( x_i, \right. \\ 
& & \left. \qquad \qquad \qquad g_i\left(y_{j(2k+1)+i}^{j(2k+1)+i+2k} \right)\right) \prod_{l=j(2k+1)+i}^{j(2k+1)+i+2k}f_{Y|X
 =x_l}(y_l)dy_l \nonumber \\
 &=& \frac{1}{2k+1} \sum_{i=1}^{2k+1} D_k \left( x^{n_i}\right)
\end{eqnarray}
\end{proof}

Proposition \ref{lem:A_cont_fn}, Lemmas \ref{lem:loss_continuity} and
\ref{lem:los_cont_w_quant} are extendible to their $k^{\text{th}}$-order equivalents with the proofs carrying over directly from the symbol-by-symbol case. We hence merely state the Lemmas for the $k^{\text{th}}$-order case and proofs are left out in this discussion.\\

\begin{proposition}
$A(x) = \int \Lambda \left(x,g\left(y_{-k}^k \right)
\right)\prod_{i=-k}^k f_{Y|x_i}(y_i)dy_{-k}^k$ is a bounded
Lipschitz function for any measurable $g: [a,b]^{2k+1} \rightarrow
\mathbb{R}$.\\
\label{lem:A_cont_fn_k}
\end{proposition}

\begin{lemma}
For any $F, \hat{F} \in \mathcal{F}^{[a,b],k}$, measurable $g:
\mathbb{R}^{2k+1} \rightarrow [a,b]$ and a bounded Lipschitz loss
function with $E_{f_{Y|u}} \Lambda(u,g(Y_{-k}^k)) <\infty $,
$\forall u$,
\begin{equation*}
\begin{split}
\left|E_{F \otimes C} \Lambda(U_0, \right. & g(Y_{-k}^k)) -  \left. E_{\hat{F} \otimes
C} \Lambda(U_0,g(Y_{-k}^k)) \right|  \\ & \le  \left( \parallel \Lambda
\parallel_L + \Lambda_{\max} \parallel \Xi
 \parallel_L^k+ (b-a)\parallel \Lambda \parallel_L \parallel \Xi
 \parallel_L^k+ \Lambda_{\max} \right) \beta \left( P, \hat{P}\right)
\end{split}
\end{equation*}
where $P$ and $\hat{P}$ are the laws associated with $F$ and
$\hat{F}$ and $\beta$ is the usual $\beta$-metric\\
\label{lem:loss_continuity_k}
\end{lemma}
$\parallel \Xi \parallel_L^k$ is the $k^{th}$ order Lipschitz
norm of the channel.
\begin{equation}
\parallel \Xi
 \parallel_L^k = \sup_{0 < \Delta < (b-a)}
\frac{\xi_{\Delta}^{2k+1}}{\Delta}
\label{eqn:delta_L^k}
\end{equation}
and $\xi_{\Delta}$ is as defined in (\ref{delta_Delta}).\\
\begin{lemma}
For any $\Delta > 0$, $F \in \mathcal{F}^{[a,b],k}$ with the
associated measure $P$, $P^{\Delta,k} \in \mathcal{F}^{\Delta,k}$,
measurable $g: \mathbb{R}^{2k+1} \rightarrow [a,b]$ and a
continuous bounded loss function with $E_{f_{Y|u}}
\Lambda(u,g(Y_{-k}^k)) <\infty $, $\forall$ $u$ ,

\begin{equation*} \left|
E_{P^{\Delta,k} \otimes C} \Lambda(U_0,g(Y_{-k}^k))-E_{F \otimes
C} \Lambda(U_0,g(Y_{-k}^k)) \right| \le \xi_{\Delta}^{2k+1}
\Lambda_{\max} + \lambda(\Delta)\left(1 + \xi_{\Delta}^{2k+1}
\right)
\end{equation*} \\
\label{lem:los_cont_w_quant_k}
\end{lemma}
These are then used to bound the deviation of the cumulative loss
incurred by the proposed denoiser for each of the $2k+1$
subsequences from the minimum possible $k^{\text{th}}$-order
sliding window loss for that subsequence. We now, state the
$k^{\text{th}}$-order equivalent of Theorem
\ref{thm:loss_denoisability_dev} for each subsequence.\\

\begin{theorem}
For all $m \ge 1$, $k \ge 1$, $\epsilon >0$, $\rho \in (0,1)$,
$\delta
>0$, $\Delta > 0$, and $x^m \in [a,b]^{(2k+1)m}$
\begin{eqnarray}
Pr \left( \left| L_{\tilde{X}^{m,\delta, \Delta,k}}(x^m,Y^m) -
D_k(x^m) \right| >  3\epsilon + 5\delta \Lambda_{\max} +
4\xi_{\Delta}^{2k+1}\Lambda_{\max} +
4\lambda(\Delta)(1+\xi_{\Delta}^{2k+1}) \right) \nonumber \\
\le |\mathcal{G}_{\delta, \Delta}^k| \left[e^{-G(\epsilon+\delta
\Lambda_{\max}, \Lambda_{\max})m} +
e^{-(1-\rho)\frac{m\gamma_k^2}{2}}
\right]+e^{-(1-\rho)\frac{m\gamma_k^2}{2}} \\
\text{for all } mh_m^k > m_k \left(\mathcal{C}, \rho,\delta, K
\right) \nonumber
\end{eqnarray}
\label{thm:loss_denoisability_dev_2kplus1}
\end{theorem}
where,
\begin{equation*}
\gamma_k = \frac{\epsilon}{\left( \parallel \Lambda \parallel_L +
\Lambda_{\max}
\parallel \Xi
 \parallel_L^k+ (b-a)\parallel \Lambda \parallel_L \parallel \Xi
 \parallel_L^k+ \Lambda_{\max} \right)}
\end{equation*}
and $G$, $\mathcal{G}^{k}_{\delta,\Delta}$ are as defined in Theorem \ref{thm:large_deviation_whole_seq}.\\
\begin{proof}
The proof of this theorem carries over directly from the proof of
Theorem \ref{thm:loss_denoisability_dev} using Proposition \ref{lem:A_cont_fn_k}, Lemmas \ref{lem:loss_continuity_k},
\ref{lem:los_cont_w_quant_k} and \ref{loss_dev_true_estdist}.\\
\end{proof}

\begin{proof}[Proof of Theorem \ref{thm:large_deviation_whole_seq}]
\begin{multline}
L_{\tilde{X}^{n,\delta,\Delta,k}}(x^n,Y^n) - D_k(x^n) =\\
L_{\tilde{X}^{n,\delta, \Delta,k}}(x^n,Y^n) - \frac{1}{2k+1}
\sum_{i=1}^{2k+1} D_k(x^{n_i}) + \frac{1}{2k+1} \sum_{i=1}^{2k+1}
D_k(x^{n_i}) - D_k(x^n)
\end{multline}
From Lemma \ref{lem:D_k_seq_subseq}, we have
\begin{eqnarray}
L_{\tilde{X}^{n,\delta,\Delta,k}}(x^n,Y^n) - D_k(x^n)  &\le&
L_{\tilde{X}^{n,\delta,\Delta,k}}(x^n,Y^n) - \frac{1}{2k+1}
\sum_{i=1}^{2k+1} D_k(x^{n_i}) \nonumber \\
 &=& \frac{1}{2k+1} \sum_{i=1}^{2k+1}L_{\tilde{X}^{n_i
,\delta,\Delta,k}}(x^{n_i},Y^{n_i}) - \frac{1}{2k+1}
\sum_{i=1}^{2k+1} D_k(x^{n_i}) \nonumber \\
 &\le& \frac{1}{2k+1} \sum_{i=1}^{2k+1} \left[ \left|L_{\tilde{X}^{n_i
,\delta,\Delta,k}}(x^{n_i},Y^{n_i}) - D_k(x^{n_i}) \right|
\right]
\end{eqnarray}
Hence,
\begin{multline}
 Pr \left(  L_{\tilde{X}^{n,\delta,\Delta,k}}(x^n,Y^n) - D_k(x^n)
 > 3\epsilon + 5\delta \Lambda_{\max} +
4\xi_{\Delta}^{2k+1}\Lambda_{\max} +
4\lambda(\Delta)\left(1+\xi_{\Delta}^{2k+1}\right) \right) \nonumber \\
\le  Pr \left( \frac{1}{2k+1}
\sum_{i=1}^{2k+1}\left|L_{\tilde{X}^{n_i ,\delta,\Delta,k}}(x^{n_i},Y^{n_i}) - D_k(x^{n_i}) \right|> 3\epsilon +
5\delta \Lambda_{\max} + 4\xi_{\Delta}^{2k+1}\Lambda_{\max} +
4\lambda(\Delta)(1+\xi_{\Delta}^{2k+1}) \right) \nonumber \\
\le \sum_{i=1}^{2k+1} Pr \left( \left|L_{\tilde{X}^{n_i ,\delta,\Delta,k}}(x^{n_i},Y^{n_i}) - D_k(x^{n_i}) \right| > 3\epsilon +
5\delta \Lambda_{\max} + 4\xi_{\Delta}^{2k+1}\Lambda_{\max} + 4\lambda(\Delta)\left(1+\xi_{\Delta}^{2k+1}\right) \right) \nonumber \\
\le (2k+1)|\mathcal{G}_{\delta, \Delta}^k|
\left[e^{-G(\epsilon+\delta \Lambda_{\max},
\Lambda_{\max})\frac{(n-2k)}{2k+1}} +
e^{-(1-\rho)\frac{(n-2k)\gamma_k^2}{2(2k+1)}}
\right]+e^{-(1-\rho)\frac{(n-2k)\gamma_k^2}{2(2k+1)}}
\end{multline}
This is true by applying Theorem \ref{thm:loss_denoisability_dev_2kplus1} to the $2k+1$ subsequences of independent supersymbols with at most $\frac{n-2k}{2k+1}$ supersymbols in each of them. Also, the cardinality of the set of all possible proposed $2k+1$-length sliding window denoisers is bounded by the cardinality of the set of all possible quantized $k^{\text{th}}$-order probability mass
functions, $\hat{P}_{x^n}^{\delta, \Delta,k}$, i.e., $|\mathcal{G}_{\delta, \Delta}^k| \le \left[ \frac{1}{\delta}+1
\right]^{{\Delta}^{2k+1}}$.\\
\end{proof}

\section{Proof of Theorem \ref{thm:stoch_main_result}}
\label{ap:Proof of Stoch Setting}
The following claim is necessary for the proof of Theorem \ref{thm:stoch_main_result}.\\*
\begin{claim}
\begin{equation*}
\lim_{k \rightarrow \infty } \min_g E \Lambda \left( X_0, g \left(
Y^k_{-k} \right) \right)= \mathbb{D} \left( F_{\mathbf{X}}, \mathcal{C}
\right)
\end{equation*}
\label{cl:denoisability_bayes}
\end{claim}

The claim results from the following lemma. \\

\begin{lemma}
\begin{itemize}
\item For $k,l \ge 0$, $E \mathcal{U} \left( F_{X_0| Y^l_{-k}}\right)$ is
decreasing in both $k$ and $l$.
\item For any two unboundedly increasing sequences of positive
integers $\{ k_n \}$, $\{ l_n \}$,
\begin{eqnarray}
\lim_{n \rightarrow \infty} E \mathcal{U} \left( F_{X_0|
Y_{-k_n}^{l_n}} \right) = E  \mathcal{U} \left( F_{X_0|
Y_{-\infty}^{\infty}} \right)
\end{eqnarray}
\end{itemize}
\label{lem:bayes_sequence}
\end{lemma}
Equipped with Lemma \ref{lem:bayes_sequence}, the proof for Claim \ref{cl:denoisability_bayes} is very similar
to that of Claim 2 in \cite{Tsachy} but we, nevertheless, present here for completeness. \\
\subsection{Proof of Lemma \ref{lem:bayes_sequence}}
\begin{proof}

A direct consequence of the definition of the Bayes envelope
$\mathcal{U} \left( \cdot \right)$ is a concave function.
Specifically, for two distribution functions $F$ and $G$ defined
on $[a,b]$, and $\alpha \in [0,1]$,
\begin{eqnarray*}
\mathcal{U} \left( \alpha F + (1-\alpha) G \right) &=&
\min_{\hat{x} \in [a,b]} \int_{x \in [a,b]} \Lambda(x,\hat{x})d
\left( \alpha F
+ (1-\alpha) G \right)(x)\\
 &=& \alpha \min_{\hat{x} \in [a,b]} \int_{x \in [a,b]}
\left[ \Lambda(x,\hat{x})dF(x) + (1-\alpha) \Lambda(x,\hat{x})dG(x) \right]\\
 &\ge& \alpha \min_{\hat{x} \in [a,b]} \int_{x \in [a,b]} \Lambda(x,\hat{x})dF(x) + \\
 & & \quad (1-\alpha) \min_{\hat{x} \in [a,b]} \int_{x \in [a,b]}
 \Lambda(x,\hat{x})dG(x)\\
 &=& \alpha \mathcal{U} \left( F \right) + (1-\alpha) \mathcal{U} \left( G \right)
\end{eqnarray*}
where the first equality follows from the fact that the mapping,
$F \mapsto Ff$, $Ff = \int f dF$, for a bona fide distribution
function, is linear. Next, to show that $E \mathcal{U} \left(
\left[F \otimes \mathcal{C} \right]_{X|Y_{-k}^l} \right)$
decreases with $l$, observe that
\begin{equation}
\begin{split}
E \mathcal{U} \left( \phantom{{{}_{}}_{{}|{}_{}^{{}_{}}}}\left[F \otimes \mathcal{C} \right. \right. & \left. \left. \right]_{X|Y_{-k}^{l+1}} \right)  = \int_{y_{-k}^{l+k+2}}
\mathcal{U} \left( \left[F \otimes \mathcal{C}
\right]_{X|Y_{-k}^{l+1}} \right)dF_{Y_{-k}^{l+1}} \\
 &= \int_{y_{-k}^l } \left[ \int_{y_{l+1}}
\mathcal{U} \left( \left[F \otimes \mathcal{C}
\right]_{X|Y_{-k}^{l}, Y_{l+1}} \right) dF_{Y_{l+1}|Y_{-k}^{l}} \right] dF_{Y_{-k}^{l}} \\
 &\le \int_{y_{-k}^l } \mathcal{U}  \left[ \int_{y_{l+1} }
\left( \left[F \otimes \mathcal{C}
\right]_{X|Y_{-k}^{l}, Y_{l+1}} \right) dF_{Y_{l+1}|Y_{-k}^{l}} \right] dF_{Y_{-k}^{l}} \\
 &= \int_{y_{-k}^l } \mathcal{U}  \left[ \int_{y_{l+1} }
\left( \int_{a}^x
\frac{f_{Y_{-k}^{l+1}|X=\alpha}dF_X(\alpha)}{f_{Y_{-k}^{l+1}}}
\right) dF_{Y_{l+1}|Y_{-k}^{l}} \right] dF_{Y_{-k}^{l}}\\
 &= \int_{y_{-k}^l } \mathcal{U}  \left[ \int_{y_{l+1} }
\left( \int_{a}^x
\frac{f_{Y_{-k}^{l+1}|X=\alpha}dF_X(\alpha)}{f_{Y_{l+1}|Y_{-k}^l}f_{Y_{-k}^l}}
\right) dF_{Y_{l+1}|Y_{-k}^{l}} \right] dF_{Y_{-k}^{l}}\\
 &= \int_{y_{-k}^l } \mathcal{U}  \left[ \int_{y_{l+1} }
\left( \int_{a}^x
\frac{f_{Y_{-k}^{l+1}|X=\alpha}dF_X(\alpha)}{f_{Y_{-k}^l}} \right)
dy_{l+1} \right] dF_{Y_{-k}^{l}} \\
 &= \int_{y_{-k}^l } \mathcal{U}  \left[ \int_{a}^x
\left( \int_{y_{l+1} }
\frac{f_{Y_{-k}^{l+1}|X=\alpha}dF_X(\alpha)}{f_{Y_{-k}^l}} \right)
dy_{l+1} \right] dF_{Y_{-k}^{l}}\\
 &= \int_{y_{-k}^l } \mathcal{U}  \left[ \int_{a}^x
\left( \frac{f_{Y_{-k}^{l}|X=\alpha}dF_X(\alpha)}{f_{Y_{-k}^l}}
\right) \right] dF_{Y_{-k}^{l}}\\
 &= \int_{y_{-k}^l } \mathcal{U}  \left[ F \otimes \mathcal{C} \right]_{X|Y_{-k}^l}
 dF_{Y_{-k}^{l}}\\
 &= E \mathcal{U} \left( \left[F \otimes \mathcal{C}
\right]_{X|Y_{-k}^l} \right)
\end{split}
\end{equation}
where, the first inequality follows from the fact that
$\mathcal{U}$ is a concave functional mapping. The definition of
$\left[ F \otimes \mathcal{C}\right]_{X|Y}$ is bona fide from the
assumption that the family of conditional measures, $\mathcal{C}$,
is absolutely continuous. Finally, application of Fubini's theorem
permits the change of order of integration to achieve the final
inequality. The fact that $E \mathcal{U} \left( \left[ F \otimes
\mathcal{C} \right]_{X|Y_{-k}^{l+1}}\right)$ decreases with $k$ is
established similarly, concluding the proof of the first item. For
the second item, similar to the proof of Lemma 4 in \cite{Tsachy},
by the martingale convergence theorem, we have,
$F_{X|Y_{-k_n}^{l_n}} \rightarrow F_{X|Y_{-\infty}^{\infty}}$
a.s., implying $F_{X|Y_{-k_n}^{l_n}} \stackrel{d}{\rightarrow}
F_{X|Y_{-\infty}^{\infty}}$. Using the convergence of random
measures \cite[Theorem ~16.16]{Kallenberg}, we have
$F_{X|Y_{-k_n}^{l_n}}f \stackrel{d}{\rightarrow}
F_{X|Y_{-\infty}^{\infty}}f$, $\forall f \in C_{K}^{+}$, the class
of continuous positive valued functions with compact support. Here, the notation $F f = \int f dF$ for any measurable $f$ and bona fide probability distribution function, $F$. In
section \ref{sec:Analysis}, we have imposed the condition of
continuity of the loss function, $\Lambda$, and since the input
alphabet space is restricted to a closed compact interval $[a,b]$,
we satisfy the condition, $\Lambda \in C_{K}^{+}$. Hence, we have,
$F_{X|Y_{-k_n}^{l_n}} \Lambda \left (\cdot,\hat{x}
\right)\stackrel{d}{\rightarrow} F_{X|Y_{-\infty}^{\infty}}
\Lambda \left( \cdot,\hat{x} \right)$, $\forall$, $\hat{x}$. Since
$\Lambda \left(\cdot, \hat{x} \right):[a,b]\times[a,b] \rightarrow
\mathbb{R}^{+}$ is a continuous mapping, in $\hat{x}$,
$\min_{\hat{x} \in [a,b]} \int \Lambda \left( x, \hat{x} \right)
dF(x)$ is also a continuous mapping. Using the fact that $\Lambda$
is a bounded mapping and the continuous mapping theorem
\cite{Durrett}, $\mathcal{U} \left( F_{X|Y_{-k_n}^{l_n}}\right)
\stackrel{d}{\rightarrow} \mathcal{U} \left(
F_{X|Y_{-\infty}^{\infty}}\right)$ and $ E \mathcal{U} \left(
F_{X|Y_{-k_n}^{l_n}}\right) \rightarrow E \mathcal{U} \left(
F_{X|Y_{-\infty}^{\infty}}\right)$.\\
\end{proof}

\subsection{Proof of Claim \ref{cl:denoisability_bayes}}

\begin{proof}[Proof of Claim \ref{cl:denoisability_bayes}]
\begin{eqnarray}
\mathbb{D} \left(F_{X^n}, \mathcal{C} \right) &=& \min_{\hat{X}^n
\in \mathcal{D}_n} E L_{\hat{X}^n} \left( X^n, Y^n \right) =
\frac{1}{n} \sum_{i=1}^n \min_{\hat{X}: \mathbb{R}^n \rightarrow
[a,b]} E \Lambda \left( X_i, \hat{X} \left( Y^n \right)\right)
\nonumber \\
 &=& \frac{1}{n} \sum_{i=1}^n \int_{\mathbb{R}^n} \min_{\hat{x} \in
 [a,b]}E \left[  \Lambda \left( X_i, \hat{x} \right)| Y^n = y^n
 \right]dF_{Y^n} \nonumber \\
 &=& \frac{1}{n} \sum_{i=1}^n \int_{\mathbb{R}^n} \mathcal{U} \left( F_{X_i|Y^n = y^n}
 \right)dF_{Y^n} \nonumber \\
 &=& \frac{1}{n} \sum_{i=1}^n E \mathcal{U}
\left( F_{X_i|Y^n = y^n} \right) = \frac{1}{n} \sum_{i=1}^n E
\mathcal{U} \left( F_{X_0|Z^{n-i}_{1-i}} \right)
\label{eqn:proof_claim2}
\end{eqnarray}
where the last equality follows by stationarity. Since by Lemma
\ref{lem:bayes_sequence}, $E \mathcal{U} \left(
F_{X_0|Y_{1-i}^{n-i}}\right) \ge E \mathcal{U} \left(
F_{X_0|Y_{-\infty}^{\infty}} \right)$, it follows from
(\ref{eqn:proof_claim2}) that $\mathbb{D}  \left( F_{X^n},
\mathcal{C}\right) \ge E \mathcal{U} \left(
F_{X_0|Y^{\infty}_{-\infty}}\right)$ for all $n$ and, therefore,
$\mathbb{D}  \left( F_{\mathbf{X}}, \mathcal{C}\right) \ge E \mathcal{U}
\left( F_{X_0|Y^{\infty}_{-\infty}}\right)$. On the other hand,
for any $k$, $0 \le k \le n$, Lemma \ref{lem:bayes_sequence} and
(\ref{eqn:proof_claim2}) yield the upper bound
\begin{eqnarray}
\mathbb{D} \left( F_{\mathbf{X}}, \mathcal{C} \right) &\le& \frac{1}{n}
\left[ 2k \mathcal{U} \left( F_{X_0} \right) + \sum_{i=k+1}^{n-k}
E \mathcal{U} \left( F_{X_0| Y_{1-i}^{n-i}}\right)  \right] \\
 &\le& \frac{1}{n} \left[ 2k \mathcal{U} \left( F_{X_0} \right) + \sum_{i=k+1}^{n-k}
E \mathcal{U} \left( F_{X_0| Y_{-k}^{k}}\right)  \right] \\
 &=& \frac{1}{n} \left[ 2k \mathcal{U} \left( F_{X_0} \right) +
 \left(n-2k \right) E \mathcal{U} \left( F_{X_0| Y_{-k}^{k}}\right)  \right]
\end{eqnarray}
Considering the limit as $n \rightarrow \infty$ of both ends of
the above chain yields $\mathbb{D} \left( F_{\mathbf{X}}, \mathcal{C} \right)
\le E \mathcal{U} \left( F_{X_0|Y_{-k}^k}\right)$. Letting now $k
\rightarrow \infty $ and invoking Lemma \ref{lem:bayes_sequence}
implies $\mathbb{D} \left( F_{\mathbf{X}}, \mathcal{C} \right) \le E
\mathcal{U} \left( F_{X_0|Y_{-\infty}^{\infty}}\right)$.\\
\end{proof}

\subsection{Proof of Theorem  \ref{thm:stoch_main_result}}
\begin{proof}
By definition of $\mathbb{D}(F_{\mathbf{X}},\mathcal{C})$ clearly
\begin{equation*}
\liminf_{n \rightarrow \infty} EL_{\tilde{X}^n_{\text{univ}}}
\left( X^n, Y^n\right) \ge \mathbb{D} \left(F_{\mathbf{X}},\mathcal{C}\right)
\end{equation*}

On the other hand, from (\ref{2k_1_denoisability}), for
any $k$
\begin{eqnarray}
E D_k \left( X^n \right) &=& E \min_{g} E_{F^k_{x^n} \otimes
\mathcal{C}} \Lambda \left(X,g(Y_{-k}^k) \right) \nonumber \\
 &\le& \min_{g}  E \left[E_{F^k_{x^n} \otimes
\mathcal{C}} \Lambda(X,g(Y_{-k}^k)) \right] \nonumber \\
 &=&  \min_{g}  E \Lambda(X,g(Y_{-k}^k))
 \label{eqn:k_denoise_expect}
\end{eqnarray}
where, the right side $X_{-k}^k$ is emitted from the (unique)
double-sided extension of the source  $F_{\mathbf{X}}$. Using the result from
equation (\ref{eqn:k_denoise_expect}), we get
\begin{equation}
\limsup_{n \rightarrow \infty} E D_{k_n} \left( X^n \right) \le
\mathbb{D} \left( F_{\mathbf{X}}, \mathcal{C} \right)
\end{equation}
implying, by Theorem \ref{thm:mainresult2kplus1} and bounded
convergence, that
\begin{equation}
\limsup_{n \rightarrow \infty} E L_{\tilde{X}_{\text{univ}}}
\left( X^n, Y^n\right) \le \mathbb{D} \left( F_{\mathbf{X}}, \mathcal{C}
\right)
\end{equation}
and proving (\ref{eqn:stoch_mainresult}). To prove
(\ref{eqn:stoch_mainresult_a}) assume stationary ergodic
$\mathbf{X}$. We have established the continuity of $E_{F \otimes
\mathcal{C}} \Lambda \left( U_0, g(Y) \right)$ w.r.t $F \in
\mathcal{F}^{[a,b]}$ in Lemma \ref{lem:loss_continuity}  and it is
easily extendible to $\min_{g} E_{F \otimes \mathcal{C}} \Lambda
\left( U_0, g(Y) \right)$. By the ergodic theorem and continuity
of $\min_{g} E_{F \otimes \mathcal{C}} \Lambda \left( U_0, g(Y)
\right)$ in $F \in \mathcal{F}^{[a,b]}$, it follows from the
representation in (\ref{2k_1_denoisability}) that
\begin{equation}
D_k \left( \mathbf{X} \right) = \lim_{n \rightarrow \infty} D_k
\left( X^n \right) = \min_g E \Lambda \left( X_0, g\left( Y_{-k}^k
\right)\right) \quad a.s.
\end{equation}
and by Claim \ref{cl:denoisability_bayes},
\begin{equation}
D(\mathbf{X}) = \mathbb{D} \left( F_{\mathbf{X}}, \mathcal{C} \right) \quad
a.s.
\end{equation}
Thus, the fact that $\limsup_{n \rightarrow \infty} D_{k_n}
(\mathbf{x})$, $\forall$ $\mathbf{x} \in [a,b]^{\infty}$ (recall
proof of Corollary 1), combined with Theorem
\ref{thm:mainresult2kplus1}, implies
\begin{equation}
\limsup_{n \rightarrow \infty} L_{\tilde{X}^n_{\text{univ}}}
\left( X^n, Y^n \right) \le \mathbb{D} \left( F_X, \mathcal{C}
\right) \quad a.s.
\label{eqn:stoch_main_res_0}
\end{equation}
On the other hand, by Fatou's lemma and definition of $\mathbb{D}
\left( F_X, \mathcal{C}\right)$
\begin{equation}
E \left[ \limsup_{n \rightarrow \infty}
L_{\tilde{X}^n_{\text{univ}}} \left( X^n, Y^n \right) \right] \ge
\limsup_{n \rightarrow \infty} E L_{\tilde{X}^n_{\text{univ}}}
\left( X^n, Y^n \right) \ge \mathbb{D} \left( F_X, \mathcal{C}
\right)
\label{eqn:stoch_main_res_1}
\end{equation}
The combination of (\ref{eqn:stoch_main_res_0}) and (\ref{eqn:stoch_main_res_1}) completes the proof of (\ref{eqn:stoch_mainresult_a})
\end{proof}

\section{Comparison to the denoiser in \cite{Dembo} }
\label{Proof of Theorem in Dembo}

Referring to Fig. \ref{fig:dembo_scheme_equiv}, each output
alphabet is uniformly quantized to the same number of levels, $M$,
as the input (for $Y \in \mathbb{R}$, the end-intervals are
greater than quantization step size). We label the set of
quantization intervals at the output as $\mathcal{O}=\{
O_1, \cdots, O_M \}$ and let the quantization step size be
$\alpha$. Corresponding to the channel output, $Y^n$, let $Z^n$ be the corresponding quantized version. Also, let $\mathcal{A}$ denote the $M$-level finite alphabet set at the input.

As a result of the quantization, we propose mapping the
$k^{\text{th}}$-order kernel density estimate at the output,
$f_Y^{n,k}$, to the corresponding probability mass function,
$\hat{Q}_{z^n}^k$, with mass at the quantized output alphabets in
the following manner,
\begin{equation}
\hat{Q}_{z^n}^k \left[ y^n \right]\left( v_{-k}^k \right) =
\int_{y_{-k}^k \in \mathcal{O}^{2k+1}}
f_Y^{n,k}(y_{-k}^k)dy_{-k}^k \label{eqn:pmf_density_map}
\end{equation}
where, $v_{-k}^k$ is the corresponding $2k+1$-tuple of the
quantized levels. The channel conditional densities also get
correspondingly mapped to an $M \times M$ channel matrix that is
formed using,
\begin{equation}
\Pi(i,j)= \int_{y: Q_{\alpha}(y) = j} f_{Y|x=i}(y)dy
\end{equation}
where $Q_{\alpha}(\cdot)$ denotes a uniform quantizer with a quantization step size $\alpha$.\\
We compare $\hat{Q}_{z^n}^k \left[ y^n \right]\left( v_{-k}^k \right)$ to $\hat{P}_{z^n}^k \left( v_{-k}^k \right)$,
the $k$-th order distribution of the quantized output symbols,
using the notation in \cite{Dembo}.
\begin{equation}
\hat{P}_{z^n}^k \left( v_{-k}^k \right) = \frac{\mathbf{r} \left[
z^n, v_{-k}^k \right]}{n-2k} \label{eqn:Dembo_pmf}
\end{equation}
The density estimate, $f_Y^{n,k}$, we consider is the cubic
histogram estimate. The histogram estimate is defined by
\begin{equation}
f_Y^{n,k}(y) = \frac{1}{n} \sum_{i=1}^n \frac{\mathbf{1}_{[Y_i \in
Anj]}}{\lambda \left( A_{nj} \right)}, \qquad y \in A_{nj}, y \in
\mathbb{R}^{2k+1} \label{eqn:hist_dens_est}
\end{equation}
where, $\mathcal{P}_n = \{ A_{nj}, j = 1,2, \cdots \}, n \ge 1$ is
a sequence of partitions and $A_{nj}$'s are Borel sets with finite
nonzero Lebesgue measure. The sequence of partitions is rich
enough such that the class of Borel sets ($\mathcal{B}^{[a,b]}$)
is equal to
\begin{equation}
\bigcap_{n=1}^{\infty} \sigma \left( \bigcup_{m=n}^{\infty}
\mathcal{P}_m \right) \label{eqn:sigma_alg_part}
\end{equation}
where $\sigma$ is the usual notation of the $\sigma$-algebra
generated by a class of sets. In particular, the cubic histogram
estimate is constructed when we consider sets $A_{nj}$ of the
form, $\prod_{i=1}^{2k+1} \left[ a_i k_i h, a_i (k_i+1)h \right)$,
$k_i$'s are integers, $h$ is a smoothing factor as for the kernel
density estimate in (\ref{eqn:hist_dens_est}) and $a_i$'s are
positive constants s.t. $a_i k_i h \in [a,b]$, $\forall h,k_i$.
The following result similar to that in Theorem
\ref{thm:density_consist}, for $J_n$ defined in equation
(\ref{eqn:J_n}), holds for histogram density estimates.\\
\begin{theorem}
Assume that the sequence of partitions $\mathcal{P}_n$ satisfies
(\ref{eqn:sigma_alg_part}). Consider
\begin{enumerate}
\item  $J_n  \rightarrow 0$ in probability  as $n \rightarrow
\infty$, for all sequences $x^n$
\item $J_n  \rightarrow 0$ almost surely  as $n \rightarrow
\infty$, for all sequences $x^n$
\item $J_n  \rightarrow 0$ exponentially as $n \rightarrow
\infty$, for all sequences $x^n$
\item For all $A \in \mathcal{B}$ with $0 < \lambda(A) < \infty$,
and all $\varepsilon >0$ there exists $n_0$ such that for all $n
\ge n_0$, we can find $A_n \in \sigma \left( \mathcal{P}_n
\right)$ with $\lambda \left( A \Delta A_n \right) < \varepsilon$
and
\begin{equation}
\sup_{M >0 \newline \text{all sets C of finite Lebesgue measure} }
\limsup_{n \rightarrow \infty} \lambda \left(  \bigcup_{j: \lambda
\left( A_{nj \bigcap C}\right)\le \frac{M}{n} } A_{nj} \bigcap C
\right) = 0
\end{equation}
\end{enumerate}
It is then true that 4 $\Rightarrow$ 3
$\Rightarrow$ 2 $\Rightarrow$ 1.\\
\label{thm:hist_density_consist}
\end{theorem}
For the proof of this theorem, refer to \cite{Luc} with the
added condition of tightness imposed on the family of measures
associated with the channel, $\mathcal{C}$.

The condition 4) in Theorem \ref{thm:hist_density_consist}
translates to $\lim_{n \rightarrow \infty} h = 0$, $\lim_{n
\rightarrow \infty} nh^d = \infty$. It can be shown as in
\cite{Luc} that they are necessary sufficient conditions for that
specified in 4) in Theorem \ref{thm:hist_density_consist}. By
choosing the smoothing factor, $h$ to be a decreasing sequence of
numbers that are all integers fractions of the quantization step
size $\alpha$, such that $nh^d \rightarrow \infty$ is also
simultaneously satisfied, we get the mapping in equation
(\ref{eqn:pmf_density_map}) to reduce to that in equation
(\ref{eqn:Dembo_pmf}) for the subsequences described in Section
\ref{sec:2kplus1}. This is because we split the sequence $x^n$
into $2k+1$ subsequences whose $2k+1$-length super symbols are
independent so that we can apply Theorem
\ref{thm:hist_density_consist}. Now,
\begin{eqnarray}
\hat{Q}_{z^{n_i}}^k \left( v_{-k}^k \right) &=& \int_{y_{-k}^k \in
\mathcal{O}^{2k+1}} f_Y^{n_i,k}(y_{-k}^k)dy_{-k}^k \\ &=&
\int_{y_{-k}^k \in \mathcal{O}^{2k+1}} \frac{1}{\lceil
\frac{n-2k-i-1}{2k+1} \rceil} \sum_{j=0}^{\lceil
\frac{n-2k-i-1}{2k+1} \rceil}
\frac{\mathbf{1}_{\left[Y_{j(2k+1)+i}^{j(2k+1)+i+2k} \in
A_{n_il}\right]}}{\lambda(A_{n_il})} \\
 &=& \frac{1}{\lceil
\frac{n-2k-i-1}{2k+1} \rceil} \mathbf{r} \left[ z^{n_i}, v_{-k}^k
\right]
\end{eqnarray}
If we mapped the finite input-continuous output channel,
$\mathcal{C}$, to $\Pi$, the mapping in equation
(\ref{eqn:chanel_inv_2kplus1}) would then reduce to,
\begin{equation}
\hat{Q}^k_{x^{n_i}} = \arg \min_{P \in
\mathcal{F}^{\mathcal{A},k}} \sum_{v_{-k}^k} \left|
\hat{Q}_{z^{n_i}}^k \left( v_{-k}^k \right) - \sum_{u_{-k}^k \in
\mathcal{A}^{2k+1}} \prod_{j=-k}^k \Pi \left( u_j,v_j \right)P \left( u_{-k}^k
\right) \right| \label{eqn:2kplus1_map}
\end{equation}
where, $\mathcal{F}^{\mathcal{A},k}$ denote the space of all
possible $k^{\text{th}}$-order distributions on $\mathcal{A}$. If
we lift the constraints of the minimizer being a bona fide element
of $\mathcal{F}^{\mathcal{A},k}$, we get the following candidate
for the minimizer in (\ref{eqn:2kplus1_map})
\begin{equation}
\hat{Q}^k_{x^{n_i}}\left[u_{-k}^k \right] = \frac{1}{\lceil
\frac{n-2k-i-1}{2k+1}\rceil } \sum_{v_{-k}^k} \mathbf{r}
\left[z^{n_i}, v_{-k}^k \right] \prod_{j=-k}^k \Pi^{-1} \left(
v_j, u_j \right)
\end{equation}
which is exactly the same as $\hat{P}_{x^{n_i}} \left[ z^{n_i}
\right]\left( u_{-k}^k \right)$ using equation (18) in
\cite{Dembo}, also given below.
\begin{equation}
\hat{P}^k_{x^{n_i}}\left[u_{-k}^k \right] = \frac{1}{\lceil
\frac{n-2k-i-1}{2k+1}\rceil } \sum_{v_{-k}^k} \mathbf{r}
\left[z^{n_i}, v_{-k}^k \right] \prod_{j=-k}^k \Pi^{-1} \left(
v_j, u_j \right)
\end{equation}
Now, using the construction of the discrete denoiser in equation
(\ref{eqn:discrete_denoiser_2kplus1}), for $\hat{Q}_{x^{n_i}}$, we
get
\begin{multline}
g_{\text{opt}}[\hat{Q}_{x^{n_i}}]\left(y_{-k}^k \right) = \arg
\min_{\hat{x} \in \mathcal{A}}
\Lambda(\cdot,\hat{x})^T[\hat{Q}_{x^{n_i}} \otimes \mathcal{C}]_{U|y_{-k}^k} \\
 = \arg \min_{\hat{x} \in \mathcal{A}} \sum_{\hat{a} \in
 \mathcal{A}} \Lambda \left( a, \hat{x} \right) \cdot \left\{ \sum_{u_{-k}^k \in \mathcal{A}^{2k+1}:u_0 =
 a}\left[ \prod_{j=-k}^k f_{Y|x=u_j}(y_j) \hat{Q}_{x^{n_i}} \left( U_{-k}^k = u_{-k}^k
 \right)\right]\right\}
\end{multline}
which is exactly the same as $g_{opt}[P]\left(y_{-k}^k \right)$ in
equation (16) in \cite{Dembo}. Hence, the proposed denoiser with
histogram density estimate of the output symbols and quantization
gives us the same denoising rule as that of \cite{Dembo} applied to the
$2k+1$ subsequences of the output sequence $Y^n$.

\bibliographystyle{IEEEtranS}
\bibliography{IEEEfull,reportref1}

\end{document}